\documentclass{statsoc}

\usepackage[a4paper]{geometry}

\usepackage{xr}
\usepackage{xcolor}
\usepackage{algorithm}
\usepackage[noend]{algpseudocode}

 \makeatletter
 \def\BState{\State\hskip-\ALG@thistlm}
 \makeatother

\usepackage{amsmath,amssymb}

\usepackage{graphicx}
\usepackage[all]{onlyamsmath}

\newcommand{\bsl}{\baselineskip}

\newcommand{\wh}{\widehat}

\newcommand{\T}{{\mbox{\scriptsize \sf T}}}

\newcommand{\R}{\mathbb{R}}

\newcommand{\bu}{\boldsymbol{u}}
\newcommand{\bv}{\boldsymbol{v}}
\newcommand{\bgamma}{\boldsymbol{\gamma}}
\newcommand{\balpha}{\boldsymbol{\alpha}}

\newcommand{\bbeta}{\boldsymbol{\beta}}
\newcommand{\bpsi}{\boldsymbol{\psi}}
\newcommand{\bsigma}{\boldsymbol{\sigma}}
\newcommand{\bxi}{\boldsymbol{\xi}}
\newcommand{\bphi}{\boldsymbol{\phi}}
\newcommand{\bz}{\boldsymbol{z}}
\newcommand{\btau}{\boldsymbol{\tau}}
\newcommand{\btheta}{\boldsymbol{\theta}}
\newcommand{\best}{\boldsymbol{\textup{e}}}
\newcommand{\bB}{\boldsymbol{\textup{B}}}
\newcommand{\bC}{\boldsymbol{\textup{C}}}
\newcommand{\bD}{\boldsymbol{\textup{D}}}

\title[Second order semi-parametric inference for log Gaussian Cox processes]{Second order semi-parametric inference for multivariate log Gaussian Cox processes}

\author{Kristian Bj{\o}rn Hessellund}
\address{Department of Mathematical Sciences, Aalborg University, Denmark}

 \author{Ganggang Xu}
 \address{Department of Management Science, University of Miami, USA}

 \author{Yongtao Guan}
 \address{Department of Management Science, University of Miami, USA}

 \author[Hessellund {\it et al.}]{Rasmus Waagepetersen}
 \address{Department of Mathematical Sciences, Aalborg University,
   Denmark}
\coaddress{Department of Mathematical Sciences, Aalborg University,
  Skjernvej 4A, DK-9220 Aalborg}
 \email{rw@math.aau.dk}

\begin{document}

\maketitle

\begin{abstract}

This paper introduces a new approach to inferring the second order
properties of a multivariate log Gaussian Cox  process (LGCP) with
a complex intensity function. We assume a semi-parametric model for
the multivariate intensity function containing an unspecified complex
factor common to all types of points. Given this model we construct a second-order
conditional composite likelihood to infer the pair correlation and
cross pair correlation functions of the LGCP. Crucially this
likelihood does not depend on the unspecified part of the intensity
function. We also introduce a cross validation method for model
selection and an algorithm for regularized inference that can be used
to obtain sparse models for cross pair correlation functions. The
methodology is applied to simulated data as well as data examples from
microscopy and criminology. This shows how the new approach
outperforms existing alternatives where the intensity functions are
estimated non-parametrically.
\end{abstract}


\hyphenation{ex-am-ple hy-phen-a-tion short}
\hyphenation{long la-tex}
\hyphenation{gen-e-ra-li-ty}
\hyphenation{li-ke-li-ho-od-funk-ti-on}
\hyphenation{log-li-ke-li-ho-od-funk-ti-on}
\hyphenation{es-ti-ma-ti-ons-funk-ti-on-er} 
\section{Introduction}

A multivariate or multi-type point pattern is a marked point pattern
where the marks belong to a finite set corresponding to different
types of points. Equivalently, a multivariate point pattern can be
viewed as a finite collection of ordinary point patterns, where each
of these point patterns consists of points of a specific
type. In this paper we consider point pattern data from biology and
criminology. In the former case the point pattern represents locations
of different types of cells in a tumor and in the latter case, crime scenes of different types of crimes.
An obvious key point of interest is then to study possible
associations between the points of different types. 

If consistent estimates of the intensity functions are available, and
assuming second-order intensity reweighted
stationarity \citep{baddeley:moeller:waagepetersen:00} or
intensity-reweighted moment stationarity
\citep{lieshout:11}, an immediate approach is to compute non-parametric cross summary statistics
such as cross $K$, cross pair correlation, or cross $J$ functions \citep{moeller:waagepetersen:03,baddeley:jammalamadaka:nair:14,cronie:lieshout:16}. Parametric estimation of cross associations is
also possible. \cite{jalilian:etal:15}, \cite{Waage} and \cite{choiruddin:etal:19}
 used parametric models of intensity
and pair correlation functions, while
\cite{rajala:murrell:olhede:17} specified a full multivariate Markov
point process model.

In some cases it is not straightforward to estimate the intensity
function. For the cells data, the intensities
of each type appear to be very heterogeneous, possibly varying
within regions corresponding to different types of tissue. However,
it is not straightforward to delienate these regions. For the crime
data, the intensity functions depend in a complex manner on the urban
structures and the population density. 

In case of bivariate
case-control processes, \cite{diggle:etal:07} suggested  a
semi-parametric model where complex features of the case and control
intensity functions were captured by a common non-parametric factor. This
factor was estimated non-parametrically from the control point
process and next used in a
semi-parametric estimate of the intensity function for the
case process. Finally this estimate was plugged into an estimate of
the $K$-function for the case
process. This estimation method mitigates the problem of confounding
of clustering in the case process with variations in the case
intensity function. However, sensitivity to the choice of bandwidth for
the non-parametric estimation remains. Also the case and control
processes were assumed to be independent whereby the cross pair
correlation function is restricted to be one. \cite{diggle:etal:07}
assumed the control process to be Poisson and \cite{henrys:brown:09}
relaxed this assumption by allowing both case and
control processes to be clustered. They however retained independence between
the two processes. \cite{guan:waagepetersen:beale:08} used the same
framework as \cite{diggle:etal:07} but used a second-order conditional
composite likelihood to fit a parametric model to the case pair
correlation function. The composite likelihood notably did not depend
on the non-parametric part of the case intensity function and hence avoided
choosing a bandwidth for non-parametric estimation. In the same
case-control setting, \cite{xu:waagepetersen:guan:19} introduced an optimal
stochastic quasi-likelihood function for estimating the parametric
component of the intensity function for the cases.

In the context of multivariate point processes,
\cite{hessellund:etal:20} used a semi-parametric model for the
multivariate intensity function. They assumed a multiplicative structure where for
each type of points, the intensity function is a product of a common background intensity
and a log-linear factor modeling effects of covariates. Hence focus is on estimating
differences between the intensity functions (for different types of points) that
can be explained in terms of the
covariates. In the bivariate case this model coincides with the one
used in \cite{diggle:etal:07} and \cite{xu:waagepetersen:guan:19}. However, \cite{hessellund:etal:20} did
not impose any restrictive assumptions regarding the correlations within each type
of points or  between different types of points. While
the main focus in \cite{hessellund:etal:20} was inference for the
intensity function, they also obtained non-parametric estimates of
ratios of cross pair correlation functions. They were, however, not
able to obtain estimates of the individual cross pair correlation
functions. \cite{xu:etal:20} obtained a consistent
estimate of a common pair correlation function in presence of a common
completely unspecified intensity function but they considered the
particular multivariate setting of independent and
identically distributed point processes.

Our objective in this paper is to infer the full within and between
correlation structure of a multivariate point process. To do so we
adopt the parametric log Gaussian Cox process (LGCP) model for the correlation
structure proposed in \cite{Waage} and further considered in \cite{choiruddin:etal:19}.  This
model is flexible and has a very natural interpretation in terms of
latent structures. However, to deal with complex intensity functions, we replace the
parametric model for the intensity function used in
\cite{Waage} with the semi-parametric model for the
intensity function from \cite{hessellund:etal:20}. In this way we
combine the strengths of two modelling approaches. 

The presence of a non-parametric factor in the intensity function
means that ingenuity is needed for fitting the parametric part of
the model.
We generalize the approach for the bivariate case in
\cite{guan:waagepetersen:beale:08} and obtain a  second order conditional composite likelihood
function. This only depends on the parametric parts of the model and
 does not require knowledge of the non-parametric
component. Compared with \cite{guan:waagepetersen:beale:08}, we consider an arbitrary number of point processes and do not assume
that any of the point processes are Poisson, nor that any two point
processes are uncorrelated.

Some key questions we want to address for a particular data set are
whether some point processes are uncorrelated and if not, whether they
are negatively or positively correlated. We address these questions by
a model selection approach where the models considered represent
different types of correlation structures. Absence of
correlation between point processes requires that certain parameters
must be zero. To enable selection of models with parameters set to zero we combine our semi-parametric composite likelihood with a Lasso
penalization \citep{tibshirani}. This precisely facilitates that some parameters can be
estimated to be exactly zero. A similar approach was considered by \cite{choiruddin:etal:19} in the context of least squares estimation for a multivariate LGCP with a full parametric model for the multivariate intensity function.

The rest of the paper is organized as follows. Section~\ref{sec:model}  gives a brief overview of multivariate point processes with focus on the intensity functions and cross intensity functions. Next, the semi-parametric model for the intensity function and the multivariate LGCP model is described. Section~\ref{sec:esti} introduces the second order conditional composite likelihood function, an optimization algorithm based on the proximal Newton method, and a cross validation method for model selection. Section~\ref{sec:simu} contains simulation studies and Section~\ref{sec:dataanalysis} applies our methodology to cells and crimes data sets. Some concluding remarks are given in Section~\ref{sec:con}.

\section{Semi-parametric modelling of a multivariate point process}\label{sec:model}

\subsection{Background on intensity functions} \label{sec:back}
Let $\boldsymbol{X}=(X_1, \dots, X_p)$ be a multivariate spatial point process,
where $X_i$ is a spatial point process on $\R^d$ representing
points of type $i$, $i=1,\ldots,p$. Each $X_i$ is hence a random subset of $\R^d$ such that the cardinality of $X_i
\cap B$ is finite almost surely for any bounded $B \subset \R^d$.  
We observe $\boldsymbol{X}$ in a spatial window $W$, where the
window $W \subset \R^d$ is bounded with area $|W|>0$. We will assume
there exist for each $i,j=1,\dots,p$, non-negative functions
$\rho_{i}(\cdot)$ and $\rho_{ij}(\cdot)$ so that the so-called
Campbell's formulae:
\begin{align} \label{eq:campbell}
\textup{E} \sum_{\bu \in X_i} h_1(\bu) &= \int h_1(\bu)
                                        \rho_{i} ( \bu ) \textup{d}\bu \\
\textup{E} \sum_{\bu \in X_i,\bv \in X_j}^{\neq}  h_2(\bu,\bv) &=
   \int h_2(\bu,\bv) \rho_{ij}(\bu,\bv) \textup{d}\bu \textup{d}\bv, \label{eq:campbell2}
\end{align}
hold for any non-negative functions $h_1(\cdot)$ and $h_2(\cdot,\cdot)$.
Here $\sum^{\neq}$ means sum over pairwise distinct pairs $(\bu,\bv)$. The function $\rho_{i}(\cdot)$ is
called the intensity function of $X_i$. If $i=j$, then
$\rho_{ii}(\cdot)$ is called the second order intensity function of
$X_i$, while if $i \neq j$, $\rho_{ij}(\cdot)$ is called the
cross intensity function between $X_i$ and $X_j$.
The normalized cross intensity function, called cross pair correlation function (cross PCF), is denoted $g_{ij}(\cdot)$ and defined by: 
$\rho_i(\bu) \rho_j(\bv) g_{ij}(\bu,\bv)=\rho_{ij}\left(\bu,\bv \right)$.
If $i=j$ we just call $g_{ii}(\cdot)$ the pair correlation function
(PCF) for $X_i$. If $X_i$ and $X_j$ are independent, then
$g_{ij}(\bu,\bv)=1$. The case $g_{ij}(\bu,\bv)>1$ ($<1$) is indicative of positive (negative)
association between $X_i$ and $X_j$ (or between points in $X_i$ in the
case $i=j$). Hence the cross PCFs provide useful insight regarding the
dependence within and between the point processes. 
We assume that $X$ is second order
cross intensity reweighted stationary and isotropic, i.e., with an
abuse of notation, $g_{ij}(\bu,\bv)=g_{ij}(r)$, $i,j=1,\ldots,p$, where $r=||\bu-\bv||$.


\subsection{Semi-parametric regression model for the intensity} \label{sec:semi}

It may sometimes be difficult to specify a simple parametric
model for the intensity functions. One may then resort to
non-parametric estimation of the intensity functions. The
  results then
depend heavily on the choice of smoothing bandwidth where
different data driven methods may result in very different results,
 see e.g. simulation studies in \cite{cronie:lieshout:18} and \cite{shaw:moeller:waagepetersen:20}.
We instead consider a
semi-parametric model where a background intensity function
$\rho_{0}(\cdot)$ captures complex variation in the intensity functions
that is common to all the point processes $X_1,\ldots,X_p$. For the
cells data considered in Section~\ref{sec:cells}, $\rho_0(\cdot)$ may
capture variations in tissue composition that influence occurrence of
different types of cells. For the crime data in
Section~\ref{sec:crimes}, $\rho_0(\cdot)$ captures variation in
population density and dependence of the intensities on the urban
structure.
More specifically, following \cite{hessellund:etal:20}, we
consider the multiplicative model:
\begin{align} \label{eq:multi1}
\rho_{i}(\bu; \bgamma_i)=\rho_{0}(\bu) \exp(\bgamma_i^\T \bz(\bu) )
\end{align} 
for the intensity of $X_i$, where $\bz(\bu)$ denotes a vector of spatial covariates at location $\bu$ and $\bgamma_i$ is  a
regression parameter vector. Let
  $\rho^{\text{pooled}}=\sum_{i=1}^p \rho_i$ denote the intensity of
  the pooled point process  $X^{\text{pooled}}=\cup_{i=1}^p X_i$. The
  intensity function $\rho_i$ 
  can then be decomposed in a natural way as $
\rho^{\text{pooled}}(\cdot)\textup{p}_i(\cdot;\bgamma_i)$  where $\textup{p}_i (\bu;\bgamma_i) =  \rho_i(\bu;\bgamma)/\rho^{\text{pooled}}(\bu)$
is the conditional probability that a point $\bu$ is of type $i$ given
that $\bu \in X^{\text{pooled}}$. 

The parameters $\bgamma_i$ are not identifiable: replacing the $l$th entry
  $\gamma_{il}$ in $\bgamma_{i}$ by $\gamma_{il}-K$ for $i=1,\ldots,p$ while replacing
  $\rho_0(\cdot)$ by $\rho_0(\cdot)\exp(K z_l(\bu))$ does not change the
  model when $\rho_0(\cdot)$ is unspecified.
\cite{hessellund:etal:20} proposed a first order conditional composite
likelihood for estimating
contrasts $\bbeta_i=\bgamma_i-\bgamma_p$.
The first order conditional composite likelihood was
obtained as the product $\prod_{i=1}^p \prod_{\bu \in X_i}
\textup{p}_i(\bu;\bbeta_i)$ with the constraint
$\bbeta_p=\boldsymbol{0}$. Alternatively, one could impose sum-to-zero
constraints $\sum_{i=1}^p \beta_{il}=0$ on the $\bbeta_i$.


Given the semi-parametric model for the intensity functions and its
associated estimation procedure we specify in the next section a model
for the correlation structure of the multivariate point process.


\subsection{Multivariate log Gaussian Cox model}\label{sec:mlgcp}

Following the setup in \citet{Waage}, we assume that $X_i$ for
$i=1,\dots, p$, is a Cox process with random intensity function:
\begin{align} \label{eq:lgcpint}
\Lambda_i(\bu)= \rho_0 (\bu) \exp(\bgamma_i^\T \bz(\bu))\textup{exp}\left(\mu_i+ \sum_{k=1}^q\alpha_{ik} Y_k(\bu)+ \sigma_i U_i(\bu) \right),
\end{align} 
where the $Y_k$ and $U_i$ are  independent zero mean unit
variance Gaussian random fields and $\mu_i= -
\sum_{k=1}^q\alpha_{ik}^2/2 - \sigma_i^2/2$.
We interpret the $Y_k$ as latent random factors that influence
all types of points. Hence the different types of points may be
correlated due to dependence on the $Y_k$. Moreover, each $U_i$ is a type-specific random factor that only
affects the $i$th type of points. Hence $U_i$ models random
clustering within each $X_i$. 

Consider for a moment the ideal situation where the $Y_k$ are
  observed (non-random). Following the same considerations as for the $\bgamma_i$ in the
  previous section we should then impose restrictions $\alpha_{pl}=0$ or $\sum_{i=1}^p \alpha_{il}=0$,  $l=1,\ldots,q$,  in order to ensure
  identifiability. In case of unobserved $Y_k$ and hence less
  information, the need for a constraint is not less pertinent.
  In the following we impose the sum-to-zero constraint
  $\sum_{i=1}^p \alpha_{il}=0$, $l=1,\ldots,q$ which treats all $X_i$
  symmetrically.   

The intensity function  of $X_i$ 
is
$\rho_i(\bu) =\textup{E} [\Lambda_i(\bu)]
= \rho_0(\bu) \exp(\bgamma_i^\T \bz(\bu)),$
which follows from the moment generating function of a Gaussian random variable. Similarly,
\begin{align*}
\rho_{ij}(\bu,\bv)= \textup{E} [\Lambda_i(\bu)\Lambda_j(\bv)] = &  \rho_0(\bu)  \rho_0(\bv) \exp(\bgamma_i^\T \bz(\bu))\exp(\bgamma_j^\T \bz(\bu)) \\
& \times \textup{exp}\left(\sum_{k=1}^q \alpha_{ik} \alpha_{jk} c_{Y_k}(\bu,\bv) + 1[i=j] \sigma_i^2 c_{U_i}(\bu,\bv) \right),
\end{align*}
where $c_{Y_k}(\bu,\bv)=\textup{Corr}[Y_k(\bu),Y_k(\bv)]$ and
$c_{U_i}(\bu,\bv)=\textup{Corr}[U_i(\bu),U_i(\bv)]$.

For $c_{Y_k}(\cdot)$ and $c_{U_i}(\cdot)$ we use exponential correlation
functions, i.e. $c_{Y_k}(\bu,\bv)=\textup{exp}(-||\bu-\bv||/\xi_k)$
and $c_{U_i}(\bu,\bv)=\textup{exp}(-||\bu-\bv||/\phi_i)$ with correlation scale parameters $\xi_k$
and $\phi_i$. Other
parametric correlation models might of course be used instead, depending on the
application. Denote by $\btheta$ the concatenation of
$\balpha_{\cdot k}=(\alpha_{1k},\ldots,\alpha_{pk})^\T$, $k=1,\ldots,q$, $\bxi=(\xi_1,
\dots, \xi_q)^\T$, $\bsigma^2=(\sigma_1^2,\ldots,\sigma_p^2)^\T$, and $\bphi=(\phi_1,\ldots,\phi_p)^\T$.
The cross PCF between $X_i$ and $X_j$ is then given by the parametric model:
\begin{align} \label{eq:crosspcf1}
g_{ij}(r;\btheta)=\textup{exp}\Big ( \sum_{k=1}^q \alpha_{ik}
  \alpha_{jk} \textup{exp}(-r/\xi_k)+ 1[i=j]\sigma_i^2
  \textup{exp}(-r/\phi_i) \Big ).
\end{align}
If $\sum_{k=1}^q\alpha_{ik} \alpha_{jk}\textup{exp}(-r/\xi_k )$ is greater (smaller) than
$0$, this implies positive (negative) spatial correlation between
points from $X_i$ and $X_j$ at the lag $r$. If for example
$\alpha_{ik}\alpha_{jk}=0$ for all $k=1,\ldots,q$, 
then $X_i$ and $X_j$ are
independent.

The number $q$ of latent common fields controls the complexity of the
model and will be chosen according to a cross validation criterion
detailed in Section~\ref{sec:cvlambda}. In \cite{Waage} and
\cite{choiruddin:etal:19}, estimation of $\btheta$ for a chosen $q$ was based on a least squares criterion where non-parametric
estimates of the pair correlation function acted as `dependent'
variables. These non-parametric estimates were based on fully
specified regression models for the log intensity
functions. This is not possible in our setting due to the
  presence of $\rho_0(\cdot)$. 
 Section~\ref{sec:esti} therefore introduces a second order
conditional composite likelihood function for estimation of $\btheta$ 
that does not require knowledge of $\rho_0(\cdot)$.

\section{Second order conditional composite likelihood} \label{sec:esti}

We assume initially that the intensity function regression
  parameters $\bbeta_i$ are known and thus suppress dependence on
these in the notation. Recall also that $\btheta$ consists of
  the parameters for the PCFs and cross PCFs. The idea is to condition on the union  of all points regardless of type
and for each $\bu \neq \bv \in \cup_{i=1}^p X_i$ consider the
conditional probability (see Section~1 in the supplementary material) that $\bu$ is of type $i$ and
$\bv$ is of type $j$:
\begin{align} \label{eq:proba}
\textup{p}_{ij}(\bu,\bv; \btheta) = \frac{\rho_{ij}(\bu,\bv)}{\sum_{k,l} \rho_{kl}(\bu,\bv)} = \frac{f_i(\bu) f_j(\bv) g_{ij}(r; \btheta_{ij})}{\sum_{k,l} f_k(\bu) f_l(\bv) g_{kl}(r; \btheta_{kl})},
\end{align}
where $f_i(\bu)=\textup{exp}(\bbeta_i^\T\bz(\bu))$, $i=1,\ldots,p$.
 Note that $\rho_0(\bu)\rho_0(\bv)$ cancels out in
\eqref{eq:proba} so that the probabilities do not depend on the unspecified $\rho_0(\cdot)$. We
then estimate $\btheta$ by maximizing the second order
conditional composite  likelihood function given by: 
\begin{align} \label{eq:lik}
L(\btheta)=  \prod_{i,j} \prod_{\bu \in X_i \cap W, \bv \in
  X_j \cap W}^{\neq} 1_R[\bu,\bv] \textup{p}_{ij}(\bu,\bv; \btheta)
   ,
\end{align}
where $1_R[\bu,\bv] = 1[\|\bu-\bv\| \leq R]$ and $R > 0$ is a user-specified tuning parameter. Specifying an
$R<\infty$ is useful for reducing computing time. Moreover, omitting pairs of points $\bu$ and $\bv$
that are distant from each other can improve
 statistical efficiency since such pairs do not provide much
information on the correlation structure.  As a rule of thumb, $R$
should be chosen so that $g_{ij}(r;\btheta) \approx 1$ for
$||\bu-\bv||>R$.  Methods for choosing $R$ are discussed in
\cite{lavancier:poinas:waagepetersen:19}.

The cross PCFs \eqref{eq:crosspcf1} and hence the second order
conditional composite likelihood function are invariant to
simultaneous interchange of columns $\balpha_{\cdot k}=(\alpha_{ik})_k$ and corresponding
correlation scale parameters $\xi_k$, as well as to multiplication by $-1$ of
$\balpha_{\cdot k}$. This lack of identifiability is  not of much
concern from a theoretical point of view
since we are not interested in the individual $\alpha_{ij}$'s but
rather the resulting correlation structure which is invariant to the
aforementioned transformations. From a practical point of view
  one might be worried about optimization convergence problems, e.g.\  the
  optimization algorithm (Section~\ref{sec:estalg}) jumping
  from one equivalent mode to another. However, in our simulation studies and data
  examples, for any particular set of initial values, the convergence
  to one particular equivalent mode was very stable.

Following the idea of two-step estimation in
\cite{waagepetersen:guan:09}, we replace the parameters $\bbeta_i$ by
consistent estimates $\hat \bbeta_i$ obtained using the first order
conditional composite likelihood  \citep{hessellund:etal:20} described
in Section~\ref{sec:semi}.

\subsection{Optimization} \label{sec:estalg}
We denote by $l_-(\btheta)$ the negation of the log of \eqref{eq:lik} and turn
the estimation of $\btheta$ into a minimization problem. In order to
minimize $l_-(\btheta)$ with respect to $\btheta$, we consider a
cyclical block descent algorithm. The strategy is to update
$\balpha$, $\bxi$, $\bsigma^2$ and $\bphi$ in turn until a convergence
criterion is met. In the following we will, with  a convenient abuse of notation, use $\balpha$ to denote both the matrix $[ \alpha_{ij}]_{ij}$ and the vectorized version where the matrix is laid out column-wise $(\balpha_{\cdot1}^\T,\dots,\balpha_{\cdot q}^\T)^\T$. It will be clear from the context which interpretation of $\balpha$ is relevant. Denote by $\btheta^{(n)}=((\balpha^{(n)})^\T,(\bxi^{(n)})^\T,(\bsigma^{2(n)})^\T,(\bphi^{(n)})^\T)^\T$
the current value of $\btheta$. We update each parameter using a
quasi Newton-Raphson iteration with additional line search. This is
equivalent to minimizing a certain least squares problem. We give the details of this since this is also needed for solving a
regularized version of our estimation problem, see
Section~\ref{sec:lasso}. 

We denote by $\tilde \btheta$ a temporary parameter vector that keeps track of
the updates leading from $\btheta^{(n)}$ to $\btheta^{(n+1)}$ and
initialize $\tilde \btheta=\btheta^{(n)}$. Denote by  $\tilde \btau
\in \{\tilde \balpha, \tilde \bxi,\tilde \bsigma^2,\tilde \bphi \}$ the
parameter vector to be updated and by $\tilde \btheta(\btau)$ the vector
obtained by replacing  $\tilde \btau$ in $\tilde \btheta$ by $\btau$. Consider a quadratic  approximation of $l_-(\tilde \btheta(\btau))$ with
respect to $\btau$ around $\tilde \btheta$:
\begin{align} \label{eq:taylor}
q(\btau) & = l_-(\tilde \btheta) + (\btau - \tilde \btau)^\T\best(\tilde \btau) + \frac{1}{2}(\btau - \tilde \btau)^\T \textup{H} (\tilde \btau)(\btau - \tilde \btau).
\end{align}
 Here (omitting for convenience the arguments $\bu,\bv$)
\begin{align*} 
\best(\btau)  =\nabla_{\btau}l_-(\tilde \btheta(\btau))= \sum_{i , j}
  \sum_{\substack{\bu \in X_i \cap W \\ \bv \in X_j \cap W}}^{\neq} 1_R \left( \frac{\sum_{k , l} \nabla_{\btau}
  \rho_{kl}(\tilde \btheta(\btau))}{\sum_{k , l}
  \rho_{kl}(\tilde \btheta(\btau))} - \frac{\nabla_{\btau}
  \rho_{ij}(\tilde \btheta(\btau))}{\rho_{ij}(\tilde \btheta(\btau))} \right)
\end{align*}
is the gradient with respect to $\btau$ and 
\begin{align*}
\textup{H}(\btau)=\textup{E} [\nabla^2_{\btau} l_-(\tilde
  \btheta(\btau))] = \int_{W^2} 1_R \textup{Cov}
  \left(Z(\tilde \btheta(\btau)) \right) \sum_{i , j}
  \rho_{ij}(\tilde \btheta(\btau)) \textup{d} \bu \textup{d} \bv
\end{align*}
is the expected Hessian with respect to $\btau$. Here $Z(\bu,\bv,\tilde \btheta(\btau)) $
denotes a random vector which takes values $\nabla_{\btau}
\textup{log}(\rho_{ij}(\bu,\bv;\tilde \btheta(\btau)))$ with probabilities
$\textup{p}_{ij}(\bu,\bv;\tilde \btheta(\btau))$ (Lemma~2.2 in the supplementary material). We estimate
$\textup{H}(\tilde \btau)$  by 
\begin{align*}
\hat{\textup{H}}(\tilde \btau)=\sum_{i , j} \sum_{\bu \in X_i \cap W, \bv \in X_j \cap W }^{\neq} 1_R[\bu,\bv] \textup{Cov}
  \left(Z(\bu,\bv,\tilde \btheta) \right),
\end{align*}
which is unbiased by \eqref{eq:campbell2}.
Since
$\hat{\textup{H}}(\tilde \btau)$ is a symmetric, positive semi-definite
matrix, the eigendecomposition implies that
$\hat{\textup{H}}(\tilde \btau)^{1/2}=UD^{1/2}U^\T$, where $D$ is the
diagonal matrix of the (all non-negative) eigen values of
$\hat{\textup{H}}(\tilde \btau)$ and $U$ is the matrix of eigen
vectors. Assuming that all the eigen values are positive,
following Section~3 in the supplementary material, the minimizer $\hat \btau$ of
\eqref{eq:taylor} is a solution of a least squares problem:
\begin{align} \label{eq:minfct2}
\hat{\btau} = \textup{arg min}_{\btau} \bigg( \frac{1}{2} ||Y
  -X\btau ||^2 \bigg) = \left(X^\T X \right)^{-1} X^\T Y,
\end{align}
where
$Y=\hat{\textup{H}}(\tilde \btau)^{1/2}\left( -\hat{\textup{H}}(\tilde \btau)^{-1}\best(\tilde \btau)+\tilde \btau \right)$
and
$X=\hat{\textup{H}}(\tilde \btau)^{1/2}$.
Introducing a line search, we update
$\btau^{(n+1)}=\tilde \btau+t(\hat{\btau}-\tilde \btau)$, for some
$t>0$ and also update $\tilde \btheta$ by replacing $\tilde \btau = \btau^{(n)}$ by
$\btau^{(n+1)}$. When all components of $\tilde \btheta$ have been
updated, we let $\btheta^{(n+1)}=\tilde \btheta$.

 As mentioned in Section~\ref{sec:mlgcp}, we impose a sum to
zero constraint on each $\balpha_{\cdot k}$, i.e. $\sum_{i=1}^p
\alpha_{ik}=0$, $k=1,\dots,q$. The constraint is easily accommodated by the
change of variable $\bB \bpsi= \balpha$, where $\bpsi$ is a $(p-1) \times q$ matrix and $\bB^\T =[ \boldsymbol{\textup{I}}_{p-1} \textup{ -}\boldsymbol{1}]$
is a $(p-1) \times p$ matrix with $\boldsymbol{\textup{I}}_{p-1}$  the $(p-1)\times (p-1)$ identity matrix and -$\boldsymbol{1} = [\textup{-}1,\dots,  \textup{-}1]^\T \in \R^p$.
Under the sum to zero constraint, the relation between $\balpha$ and
$\bpsi$ is one-to-one. Thus in case of $\btau=\balpha$, we update
the unconstrained parameter $\bpsi$ using (by the chain rule) the
gradient $\bB^\T \best(\tilde \btheta)$ and the Hessian
$\bB^\T \textup{H}(\tilde \btheta) \bB$, and finally let $\balpha^{(n+1)}=\bB
\bpsi^{(n+1)}$. 

The cyclical block updating is iterated until relative function
convergence,
\begin{align} \label{eq:relconv}
\left\vert \left[ l_- \left( \btheta^{(n+1)} \right) -l_-\left( \btheta^{(n)} \right) \right]/l_-\left( \btheta^{(n)} \right) \right\vert < \epsilon,
\end{align}
for some $\epsilon > 0$ in which case we set $\hat \btheta=\btheta^{(n+1)}$. Algorithm \ref{fig:alg1} gives a brief overview of the cyclical block  descent algorithm.
\begin{algorithm}
\caption{Cyclical block descent algorithm}\label{fig:alg1}
\begin{algorithmic}[1]
\BState Simulate initial parameters $\hat{\bpsi}^{(0)}$,
$\hat{\bxi}^{(0)}$, $\hat{\bsigma}^{2(0)}$ and $\hat{\bphi}^{(0)}$
\BState $n:=0$
\BState $\boldsymbol{repeat}$
\BState \quad $\tilde \bpsi := \bpsi^{(n)}$,  $\tilde \bxi :=
\bxi^{(n)}$, $\tilde \bsigma^2 := \bsigma^{2(n)}$ and
$\tilde \bphi := \bphi^{(n)}$
\BState \quad update $\tilde \bpsi$, $\tilde \bxi$, $\tilde
\bsigma^2$ and $\tilde \bphi$ in turn using \eqref{eq:minfct2}
combined with line search
\BState \quad $\bpsi^{(n+1)}:= \tilde \bpsi$,
$\bxi^{(n+1)}:=\tilde \bxi$, $\bsigma^{2(n+1)} := \tilde
\bsigma^2$,  $\bphi^{(n+1)} := \tilde \bphi$, and $\balpha^{(n+1)}:=B \tilde \bpsi^{(n+1)}$
\BState \quad $n:=n+1$
\BState $\boldsymbol{until}$ relative convergence criterion \eqref{eq:relconv}
\BState $\boldsymbol{return}$ $\hat \btheta = \btheta^{(n)}$
\end{algorithmic}

\end{algorithm}
\subsection{Optimization with lasso regularization}\label{sec:lasso}

The overall model complexity is controlled by the
number $q$ of latent fields $Y_j$ with associated parameters $\alpha_{ij}$,
$i=1,\ldots,p$, $j=1,\ldots,q$. Nevertheless, for any $q$, more sparse
submodels could be obtained by restricting some $\alpha_{ij}$ to zero. Of course, if all entries in a
column $\balpha_{\cdot k}$ are restricted to zero, this just
corresponds to reducing $q$ by one. In order to look for sparse submodels
for a given $q$, we extend the estimation approach by introducing a
lasso regularization on $\balpha$.  
We express the sum to zero constraint for $\balpha$ by $\bC \balpha = \boldsymbol{0}$, where $\bC = [\bD_1 \cdots \bD_q]$
is a $q \times pq$ matrix that consists of submatrices $\bD_i$, $i=1,\dots q$, of dimension $q \times p$. Each submatrix $\bD_i$ consists of ones on the $i$th row and zeros otherwise. Here $\balpha$ should be interpreted as the vector obtained by concatenating the $\balpha_{\cdot k}$, cf.\ Section~\ref{sec:estalg}. Note that the regularization is not relevant in the bivariate case $p=2$ since in this case, by the sum to zero constraint,  $\alpha_{1k}=0$ implies $\alpha_{2k}=0$ which just corresponds to reducing $q$ by 1.

The regularized object function becomes:
\begin{equation}\label{eq:penalize}
l_-(\btheta) + \lambda \sum_{i=1}^{p}\sum_{j=1}^q | \alpha_{ij} |, \quad \bC \balpha= \boldsymbol{0},
\end{equation}
where $\lambda \sum_{i=1}^{p}\sum_{j=1}^q | \alpha_{ij} |$ is a
  lasso penalty that can lead to exact zero components in the estimate
of $\balpha$. 
We minimize this using a cyclical block descent algorithm which only differs from the one in Section~\ref{sec:estalg} by the update 
$\hat \balpha = \textup{arg min}_{\balpha} \Big( \frac{1}{2} ||Y
  -X\balpha ||^2  + \lambda \sum_{i=1}^{p}\sum_{j=1}^q| \alpha_{ij} |
  \Big ) \textup{ subject to } \bC \balpha= \boldsymbol{0}.$
To compute $\hat \balpha$ under the sum to zero constraint, we
use the augmented Lagrangian algorithm suggested in
\cite{shi:etal:16}. Details are given in
Section~4 in the supplementary material. In Section \ref{sec:cvlambda} we propose a cross validation procedure to choose $\lambda$.

The suggested algorithm can easily be extended to handle elastic net
regularization \citep{zou:hastie:05} that combines lasso with ridge regression. However,
this introduces yet a tuning parameter controlling the convex
combination of lasso and ridge regression. Also our main focus is
sparsity as provided by the lasso. For simplicity of
exposition, we therefore focus here on the lasso.

\subsection{Determination of $q$ and $\lambda$ } \label{sec:cvlambda}

We choose the values of $q$ and $\lambda$ according to a
$K$-fold ($K \ge 2$) cross validation criterion constructed so
that it targets selection of an appropriate cross correlation
structure.  Let for each $i,j$, $M_{ij}$ denote the set of pairs
$(\bu,\bv)$ with $\bu \in X_i$, $\bv \in
X_j$, and $0<\|\bu - \bv\|\le R$. We randomly split $M_{ij}$ into $K$ equally sized subsets
$M_{ij,1}, \ldots M_{ij,K}$. 
We then obtain
for each $k=1,\ldots,K$, a parameter estimate $\hat \btheta_k$ by
maximizing the regularized
conditional composite likelihood 
\[ 
l_k(\btheta)=  \sum_{i,j} \sum_{(\bu,\bv) \in M_{ij,-k}}^{\neq} \log \textup{p}_{ij}(\bu,\bv; \btheta)+ \lambda \sum_{i=1}^{p}\sum_{j=1}^q| \alpha_{ij} |, \quad \bC \balpha= \boldsymbol{0},
\]
for the training data set consisting
of the set of pairs $M_{ij,-k}= \cup_{l \neq k} M_{ij,l}$ where
  $M_{ij,k}$ is left out. 
The $k$th cross validation score based on the validation sets $M_{ij,k}$, $i \neq j$, is then
\[
CV_k(q,\lambda)=  \sum_{i \neq j} \sum_{(\bu,\bv) \in M_{ij,k}} \log
  \textup{p}_{ij}(\bu,\bv; \hat \btheta_k).
\]
We here omit the $M_{ii,k}$ to focus the cross validation on the fit of
the cross correlation structure. 
To reduce the sensitivity to Monte Carlo variation, one may compute
cross validation scores $CV_{kl}(q,\lambda)$, $l=1,\ldots,L$, based on
$L$ independent $K$-fold random splits of the data, and use the average
$\overline{CV}(q,\lambda)$ of the $CV_{kl}(q,\lambda)$,
$k=1,\ldots,K$, $l=1,\ldots,L$.
We do not have a theoretical foundation for choosing $K$. We
  hence follow the rule of thumb in the statistical learning literature
  \citep{hastie:tibshirani:friedman:13} and choose $K$ in the range of
  5 to 10. Regarding $L$, a very large $L$ can essentially eliminate
  the Monte Carlo variation due to the random splitting of the data
  into $K$ folds. However, a moderate $L$ is often used in practice
  to avoid excessive computing times. For example, we use $K=5$ and $L=10$ in
  our simulation studies in Section~\ref{sec:simu}.

 Consider the case $\lambda=0$ which is relevant for example when $p=2$. The most obvious choice of $q$ is the one that minimizes the cross
validation score, $q_{\min}=\textup{arg min}_{q} \overline{CV}(q,0)$.
We denote this the minimum (MIN) rule.
However, due to sensitivity to Monte Carlo error, a so-called one
standard error rule has  been proposed
\citep{hastie:tibshirani:friedman:13} that promotes more sparse
solutions. Let $\text{SD}(q,0)$ denote the standard deviation of a
cross validation score $CV_{kl}(q,0)$ obtained from a single
validation set. In the current framework, the one standard error
(1-SE) rule selects the smallest $q$ (denoted $q_{\text{1-SE}}$) for which 
$\overline{CV}(q,0) \leq \overline{CV}(q_{\min},0) + \text{SE}(q_{\min},0),$
where 
$\text{SE}(q,0)=\text{SD}(q_{\min},0)/\sqrt{KL}$ is the standard error
of  $\overline{CV}(q_{\min},0)$.

For joint selection of $(q,\lambda)$, the immediate choice would be the minimizer of $\overline{CV}(q,\lambda)$. However, computing $\overline{CV}(q,\lambda)$ over a two-dimensional grid of $q$ and $\lambda$ values is very time consuming. Instead we use a two-step approach where we first determine $q_{\min}$ as in the previous paragraph and next choose the $\lambda$ that minimizes $\overline{CV}(q_{\min},\lambda)$ over values of $\lambda$. Thus the initial selection of $q$ determines the overall model complexity while the subsequent possible selection of a $\lambda>0$ may introduce additional sparsity given $q_{\min}$.

\subsection{Model assessment}\label{sec:modelassess}

Assuming the model \eqref{eq:multi1} for the intensity functions,
\cite{hessellund:etal:20} obtained a consistent non-parametric
estimate of any ratio $g_{ij}(r)/g_{lk}(r)$ of cross PCFs, $r > 0$,
$i,j,l,k=1,\ldots,p$. Similarly, we can obtain semi-parametric
estimates of these ratios based on our semi-parametric estimates of
the cross PCFs. If the assumed multivariate LGCP model is valid, the
non-parametric  and the semi-parametric estimates of cross PCF ratios
should not differ much. In our data examples in
Section~\ref{sec:dataanalysis} we informally assess the models by
visual comparison of the two types of estimates. We also conduct a
so-called global envelope goodness-of-fit test
\citep{myllimakki:etal:16} based on the difference between the two
types of estimates over spatial lags $r \in [0,R]$. This requires
simulation under a null model. For this we use the fitted multivariate
LGCP where we replace the unknown background intensity $\rho_0$ by a
non-parametric estimate introduced in \cite{hessellund:etal:20}, see
also Section~5 in the supplementary
material. However, as discussed in
  Section~7, due to the impact of using an
  estimated $\rho_0$ it may be hard to interpret results of the model
  assessment - especially if a chosen summary statistic depends on the
  unknown $\rho_0$.

\section{Simulation study} \label{sec:simu}
In this section we use simulations to study the joint
  performance of the second order
conditional composite likelihood, the optimization algorithm in
Section~\ref{sec:estalg} with additional lasso regularization in
Section~\ref{sec:lasso}, and the cross validation procedure. We
conduct simulation studies based on three different settings for a
five-variate point process $\boldsymbol{X}=(X_1,X_2,X_3,X_4,X_5)^\T$ on
$W=[0,1]^2$. For the two first
settings, $\boldsymbol{X}$ is a
multivariate LGCP. For the third setting, considered in
Section~7 of the supplementary material,
$\boldsymbol{X}$ is a so-called product shot-noise Cox process. The
construction of the random intensity function of a product shot-noise
Cox process is very different from that of an LGCP. We consider the
product shot noise Cox process to study how the proposed
methodology works in case of a misspecified model. The
qualitative conclusions from the third setting are quite similar to
the first two settings. We therefore focus here on the details and results
for the first two settings. Although it would be interesting to
consider simulations from other types of multivariate models
\cite[e.g.\ Gibbs as in][]{rajala:murrell:olhede:17}, we confine ourselves
to multivariate models where the PCFs and cross PCFs are known.

For both LGCP settings we simulate one covariate $Z(\cdot)$ and one
background intensity $\rho_0(\cdot)=400\exp(0.5 V(\cdot) -0.5^2/2)$,
where $Z$ and $V$ are zero mean unit variance Gaussian random fields
with exponential and Gaussian correlation functions,
i.e. $\textup{Corr}(Z(\bu),Z(\bv))=\textup{exp}(-||\bu- \bv||/0.05)$
and $\textup{Corr}(V(\bu),V(\bv))=\textup{exp}(-(||\bu-
\bv||/0.2)^2)$. The particular realizations of $Z$ and $\rho_0$ are shown in
Figure~\ref{fig:covlambda0}. These realizations are fixed
  throughout the simulation study.
\begin{figure}[ht!]
\centering
\begin{tabular}{c c}
\includegraphics[scale=0.25]{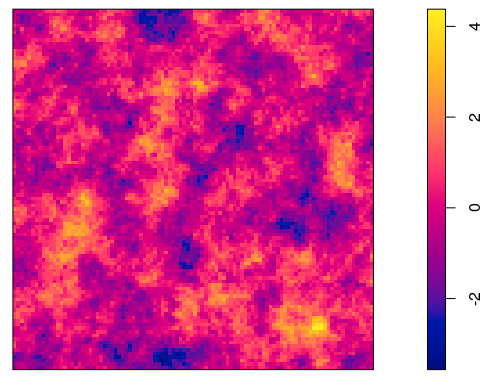}
\hspace{0.1cm}
\includegraphics[scale=0.25]{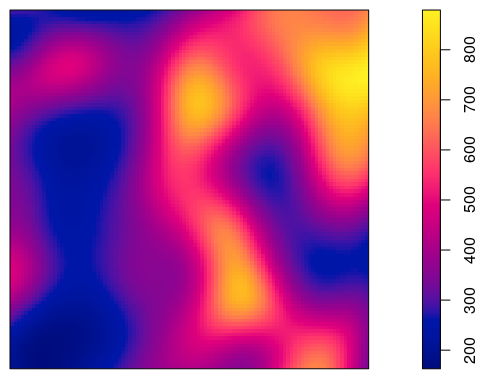}
\end{tabular}
\caption{Left: simulated covariate $Z$. Right: simulated $\rho_0$.} \label{fig:covlambda0}
\end{figure}

Table~\ref{tab:lgcp} shows the values used for the intensity
  function regression parameters $\bgamma$, and the standard deviation and
  correlation scale parameters $\bsigma$ and $\bphi$ for the
  type-specific latent fields. The table also shows the expected
  number of points $N$ for each point process.
\begin{table}
\caption{\label{tab:lgcp} Simulation settings for $\boldsymbol{X}$ in
  each setup $q=0,2$ (excluding $\balpha$ and $\bxi$). }
\centering
\fbox{%
\begin{tabular}{|c|cc|cc|c|c|cc|cc|c|}
\hline
$\boldsymbol{X}$ & $\bgamma_1$ & $\bgamma_2$ & $\bsigma$ & $\bphi$ &  $N$  &
                                                                     $\boldsymbol{X}$
  & $\bgamma_1$ & $\bgamma_2$ & $\bsigma$ & $\bphi$ & $N$\\
  \hline
  $X_1$ & 0.1 & -0.1 & $0.71$ & 0.02 & 550& $X_4$ & 0.4 & 0.1 & $0.71$ & 0.03 &750\\
  $X_2$ & 0.2 & -0.2 & $0.71$ & 0.02 & 619& $X_5$ & 0.5 & 0.2 & $0.71$ & 0.04 &830\\
  \cline{7-12}
  $X_3$ & 0.3 & 0 & $0.71$ & 0.03 & 677\\ \cline{1-6} 
  \hline
\end{tabular}}
\end{table}
Regarding the number $q$ of common latent fields with
  associated coefficients $\balpha$ and correlation scale parameters $\bxi$, we take $q=0$ for the first setting resulting in a case with independent components $X_1,\ldots,X_5$. In the second setting we let $q=2$ and choose $\balpha$ as specified in the table left in Figure~\ref{fig:simsetup2}. We moreover let $\xi_1=0.02$ and $\xi_2=0.03$. The resulting PCFs and cross PCFs are shown in the middle and right plots in Figure~\ref{fig:simsetup2}.
\begin{figure}[ht!]
\centering
\fbox{%
\begin{tabular}{|c| c c |}
\hline
$\boldsymbol{X}$ & $\balpha_{1}$ & $\balpha_{2}$ \\
  \hline
 $X_1$ & 0.5 & -1  \\
 $X_2$ & 0.5 & 0  \\
 $X_3$ & -1 & 0  \\
 $X_4$ & 0 & 0.5 \\
 $X_5$ & 0 & 0.5 \\
  \hline
\end{tabular}}
\hspace{0.1cm}
\begin{tabular}{c c}
\includegraphics[width=0.32\textwidth]{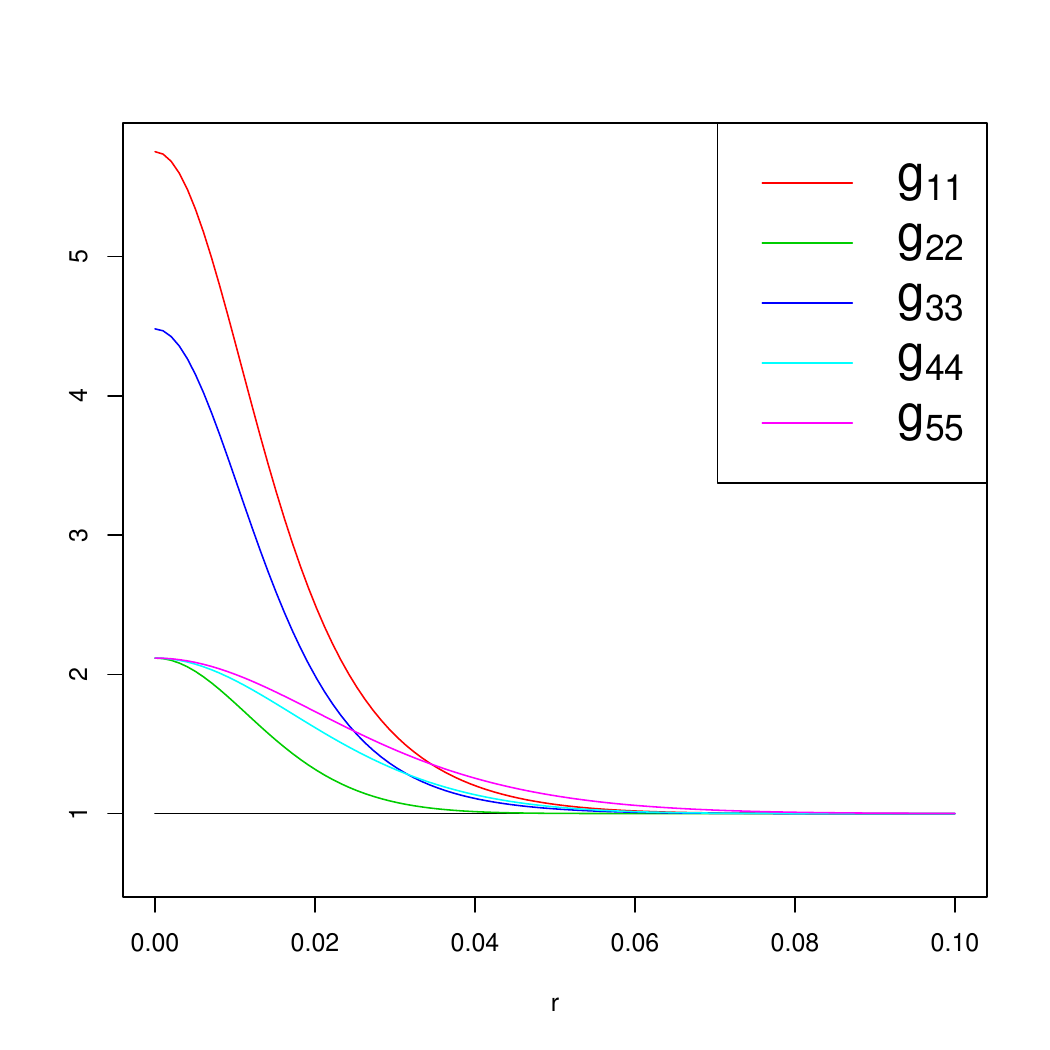}
\includegraphics[width=0.32\textwidth]{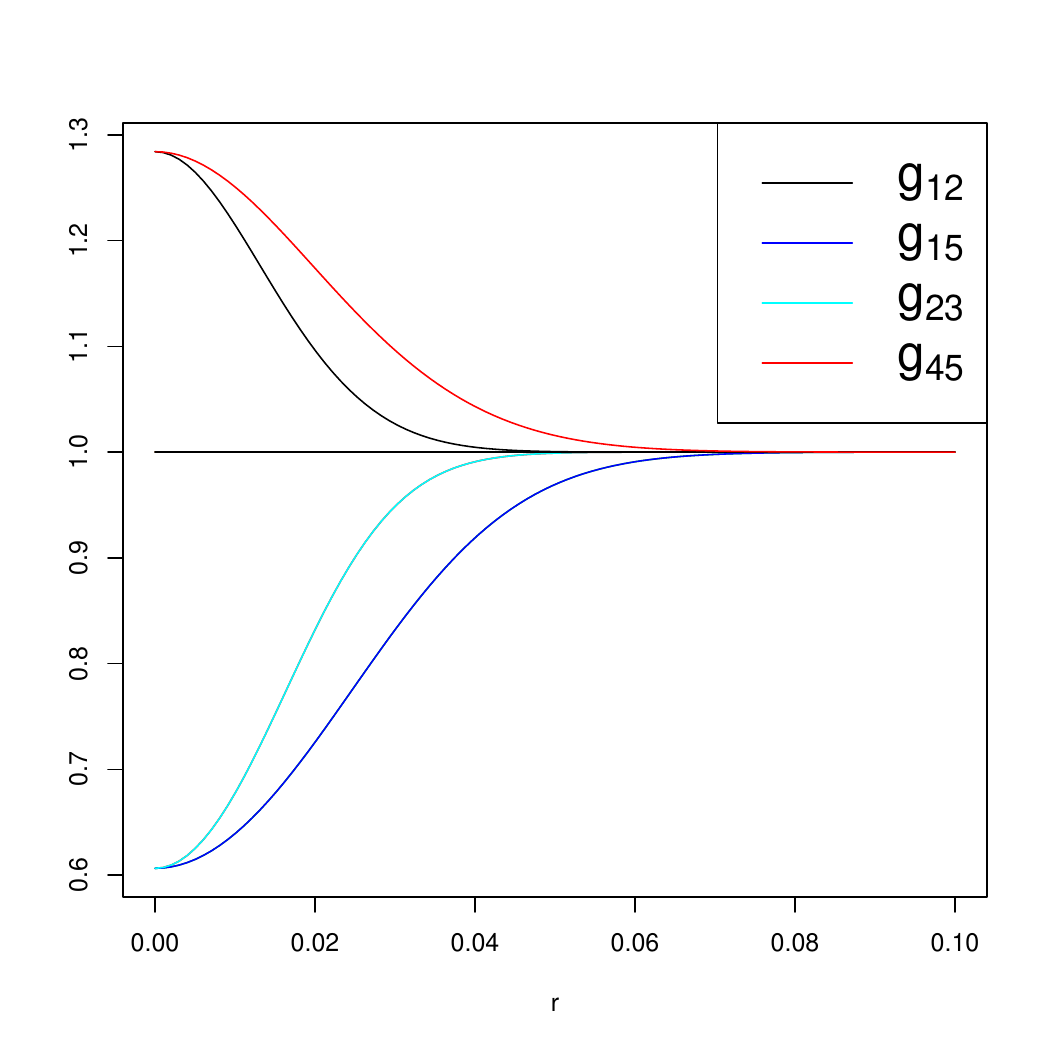}
\end{tabular}
\caption{Left: $\balpha$. Middle: true PCFs. Right: true
  cross PCFs (note $g_{13}=g_{23}$ and $g_{14}=g_{15})$.}\label{fig:simsetup2}
\end{figure}
In the case $q=2$ we have a positive spatial dependence between $X_1$
and $X_2$ and between $X_4$ and $X_5$, while there is a negative
spatial dependence between $X_3$ and $(X_1,X_2)$ and between $X_1$ and
$(X_4,X_5)$. 

For our second order conditional composite likelihood we specified
exponential correlation models for the fields $Y_k$, $k=1,\dots,q$ and
$U_i$, $i=1,\dots,5$. In practice it is rarely the case  that the true
correlation models correspond exactly to the specified ones. To  reflect this we simulate the $Y_k$ and $U_i$ using Gaussian correlation functions, i.e.\ $\textup{Corr}(Y_k(\bu),Y_k(\bv))=\textup{exp}(-(||\bu- \bv||/\xi_k)^2)$ and $\textup{Corr}(U_i(\bu),U_i(\bv))=\textup{exp}(-(||\bu- \bv||/\phi_i)^2)$. Hence the model applied is misspecified for the simulated data. For each setting we generate 100 simulated realizations of $\mathbf{X}$.

In both settings we select $q$ among the values $0,1,\ldots,5,$ and
next, for the chosen $q$, $\lambda$ among the values
$10,8,6,5,4,3,2,1,0.5,0.25,0$. For this we use cross
  validation with $K=5$ and $L=10$. We consider
results using both the MIN and the 1-SE approach for the selection of
$q$. To assess the effect of regularization in an
  over-parametrized setting, we also report results where $q=7$ is fixed with $\lambda$ chosen by cross
validition. For the second order
conditional composite likelihood we only consider distinct pairs of
points $\bu,\bv$ with $\|\bu-\bv\| \le R= 0.1$.
The initial values for the components of $\btheta=(\balpha,\bxi,\bsigma^2,\bphi)$ are simulated as $\alpha_{ik} \sim
\textup{Unif(-0.25,0.25)}$, $\xi_k,\phi_i \sim \textup{Unif(0.01,0.04)}$ and $\sigma^2_i \sim \textup{Unif(0.4,0.6)}$. 
The parameter $\epsilon$ in the relative function convergence
criterion \eqref{eq:relconv} is set to $10^{-5}$ and the
convergence parameters $\tilde{\epsilon}$ and
  $\tilde{\tilde{\epsilon}}$ for the regularized optimization ((3) in the supplementary material) are set to $10^{-10}$.

We measure the performance for each selected model using mean integrated squared error (MISE) aggregated over respectively all  PCFs and all cross PCFs, i.e.
\begin{align} \label{eq:mise}
  \text{MISE}_{\text{between}}(\hat{\btheta})=\sum_{i<j} \textup{E} \left[\int_{0.01}^{0.1} \left( g_{ij}(r;\hat{\btheta}_{ij}) - g_{ij}(r;\btheta_{ij}) \right)^2 \textup{d} r \right]. 
\end{align}
We also consider $\text{MISE}_{\text{within}}(\hat{\btheta})$
and $\text{MISE}_{\text{total}}(\hat{\btheta})$ defined as $\text{MISE}_{\text{between}}(\hat{\btheta})$ but with sum over $i=j$ or $i \leq j$. 

We compare the performance of the proposed method with two non-parametric approaches. For the first approach, referred to as `simple', we estimate the intensity functions non-parametrically using the
\texttt{spatstat} \citep{baddeley:rubak:turner:15} procedure
\texttt{density.ppp} with bandwidths selected using the method introduced in
\cite{cronie:lieshout:18}. Next the PCFs and cross PCFs are estimated using the
\texttt{spatstat} procedures \texttt{pcfinhom} and
\texttt{crosspcfinhom} with the intensity functions replaced by the non-parametric estimates. For the PCF and cross PCF estimation we manually specify reasonable bandwidths based on the knowledge of the
true PCFs and cross PCFs (note that this is in favor of the
non-parametric approach). The second approach, referred to as `Diggle', is an adaption to the multivariate case of
the method proposed in \cite{diggle:etal:07} (see Section~6 in the supplementary material for details).
To measure the performances of the non-parametric approaches we simply replace the fitted parametric cross PCFs in \eqref{eq:mise} by the non-parametric estimates. 

\subsection{Five-variate LGCP with zero common latent fields} \label{sec:simstudy1}

In the case $q=0$, the MIN rule only selects the true value $q=0$ for 1\% of the
simulated data sets while values of $q=1,2,3$ are selected for 99\%
of the simulations, see left Table~\ref{tab:simcv}. Using the 1-SE
rule, $q=0$ is selected in 77\% of the cases and a value of $q$ bigger
than 1 is only selected in two cases. The reason that the MIN rule
frequently selects $q$ larger than zero may be that neither of
the models with $q=0,\ldots,3$ are severely overparametrized. E.g.\
with $q=3$, the in total 15 PCFs and cross PCFs are parametrized using
just 25 parameters, i.e.\ less than 2 parameters on average for each
PCF or cross PCF. Hence overfitting that can be detected by the cross
validation procedure mainly occurs for $q=4,5$. The middle third
column in Table~\ref{tab:simcv} (left) shows 95\% probability
intervals for the selected $\lambda$s when $q=1,2,3$ and the last
column shows the average percentages of $\alpha_{ik}$'s that are estimated to be 0. These columns show that when a larger $q$ is selected, then also a larger $\lambda$ is selected leading to a higher percentage of zeros in the estimated $\balpha$. This makes sense since larger $q$ means more superfluous parameters and hence more need for regularization. In the case $q=7$, the selected $\lambda$s tend to be markedly larger than for the smaller $q$s up to 3. Also 52\% of the $\alpha_{ik}$ are estimated to be zero in the case $q=7$ while the percentages are quite small for $q$ up to 3. For $q=1,2,3$, the selected $\lambda$ was zero (meaning no regularization) in 59\%, 57\%, and 55\% of the cases indicating that $q=1,2,3$ already leads to a rather sparse setup and explaining the small percentages of $\alpha_{ik}$ estimated to be zero.
\begin{table}
\caption{\label{tab:simcv}(Left: true $q=0$. Right: true $q=2$) Distribution of $q$ chosen by MIN and 1-SE rules, 95\% probability interval for selected $\lambda$s, and averages over simulated data sets of percentages of estimated $\alpha_{ik}$s equal to zero.} 
\centering
\fbox{%
\begin{tabular}{|c| c c c c| c c c c|}
$q$ & MIN & 1-SE & $\lambda$ &  ($\% \alpha_{ik}=0$) & MIN & 1-SE & $\lambda$ &  ($\% \alpha_{ik}=0$) \\
\hline
$0$ & 1 & 77 & - & - & 0 & 0 & - & - \\
$1$ & 32 & 21&  (0;0.61) & 2 & 2 & 39 & 0 & 0 \\
$2$ & 56 & 2 & (0;0.41) & 7 & 60 & 61 & (0;0.38) & 1  \\
$3$ & 11 & 0 &  (0;0.88) & 6 & 36 & 0  & (0;0.25)  & 0.6 \\
$4$ & 0 & 0 & - & - & 2 & 0 &  0 & 0  \\
\hline
$7$ & - & - & (0;4.52) & 52 & - & - & (0;2) & 26 \\
\hline
\end{tabular}}
\end{table}
 Figure~\ref{fig:meancvq0} shows the average of $\overline{CV}(q,0)$ over all simulated data sets and confirms that the CV scores are quite similar across different $q$. 
\begin{figure}[ht!]
\centering
\begin{tabular}{c}
\includegraphics[width=0.4\textwidth,height=0.3\textwidth]{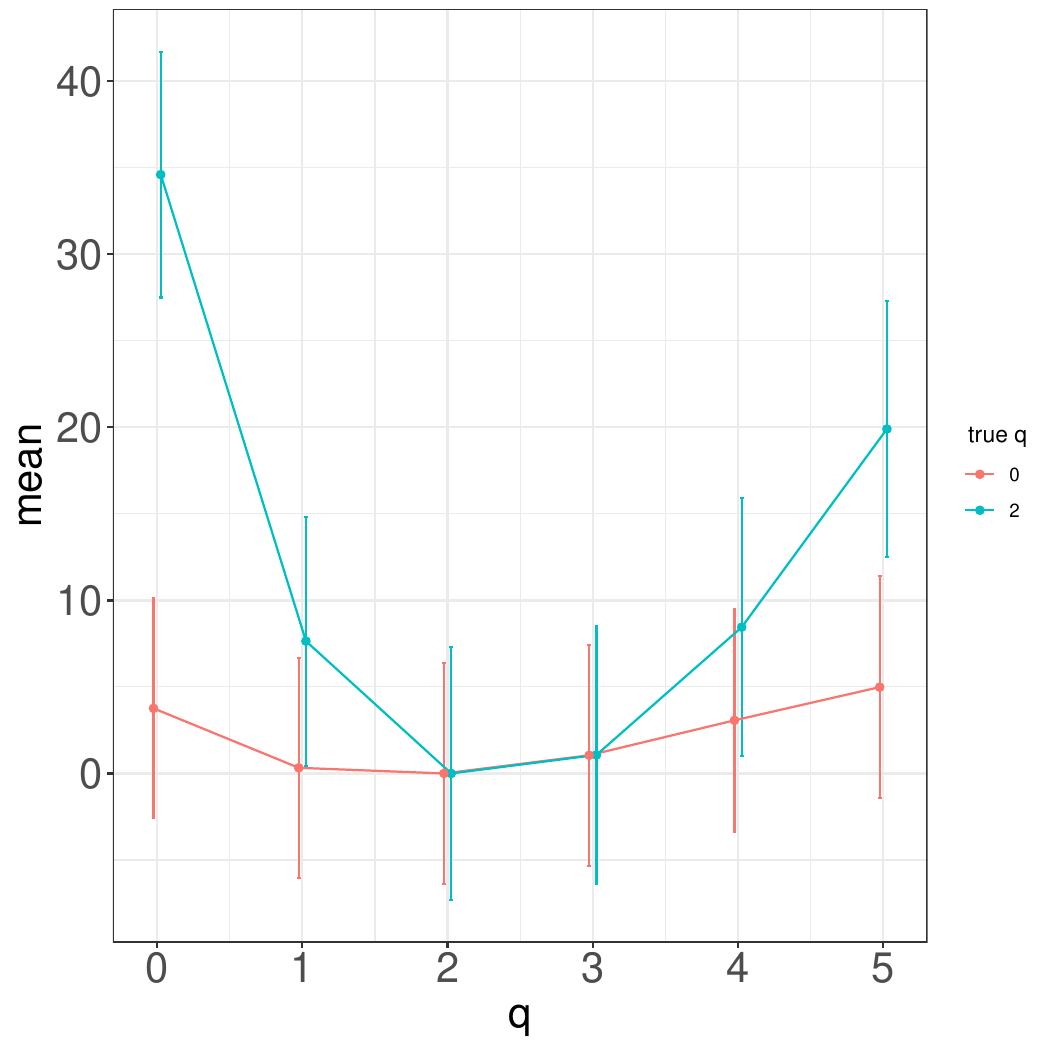}
\end{tabular}
\caption{Averages over simulated data sets of $\overline{CV}(q,0)$ scores with minimum average CV-score subtracted. The bars show the average of the standard errors $SE(q,0)$ obtained for $\overline{CV}(q,0)$ for each simulated data set. Red is for $q=0$ while blue is for $q=2$. 
}\label{fig:meancvq0}
\end{figure}

Figure~\ref{fig:simpcf1} shows means and 95\% pointwise probability
intervals for estimates of a subset of the PCFs and cross PCFs
obtained for the simulated data sets  with $q$ selected among
$0,1,2,3,4,5$ using either MIN or 1-SE and with $\lambda=0$. We only
show estimates with no regularization since the regularized estimates
are very similar. The means are quite similar for MIN and
1-SE, and the MIN and 1-SE estimates are close to unbiased for
the cross PCFs. A moderate bias is present for the PCF. This is not unexpected as we specify the wrong parametric model. However, the simple non-parametric estimates are strongly biased in all cases.
\begin{figure}[ht!]
\centering
\begin{tabular}{c c c}
\includegraphics[width=0.33\textwidth]{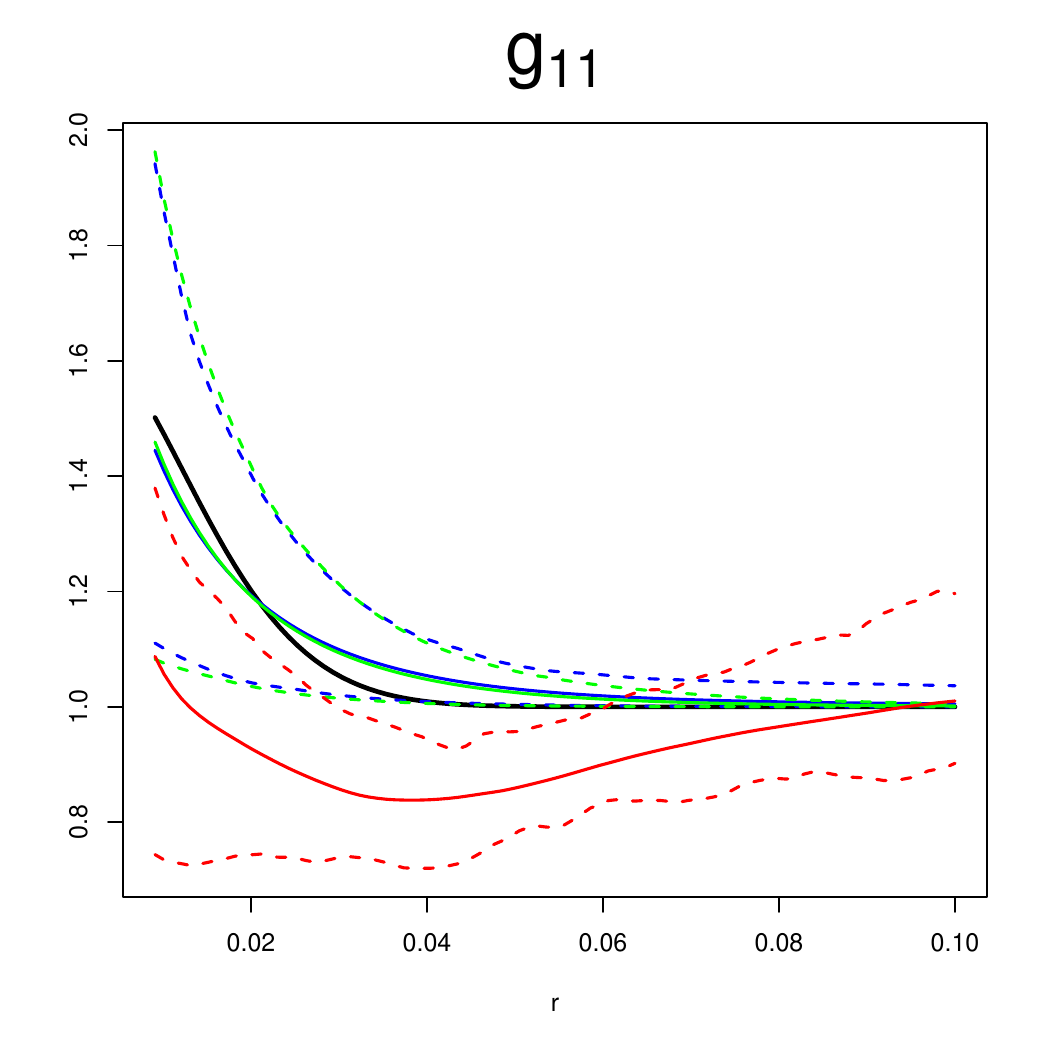}
\includegraphics[width=0.33\textwidth]{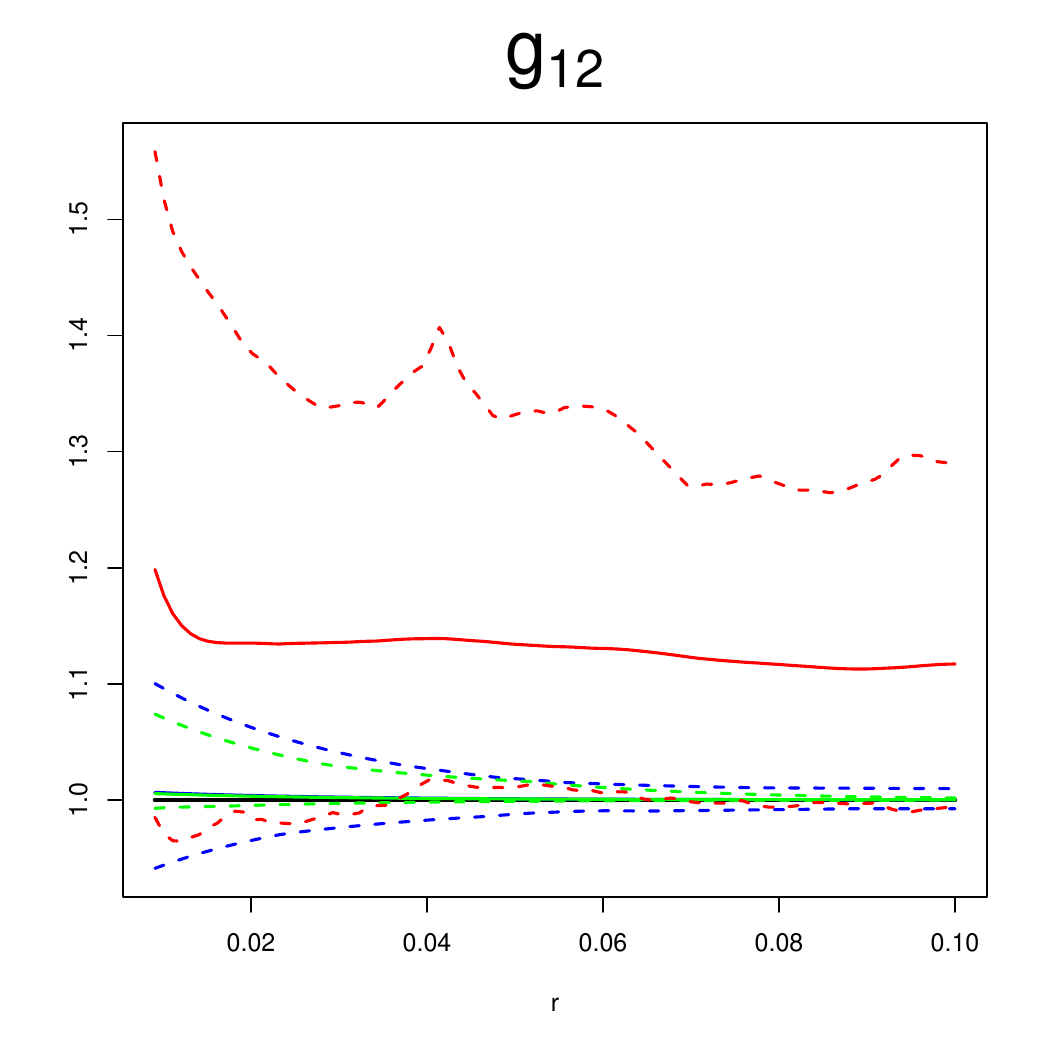}
\includegraphics[width=0.33\textwidth]{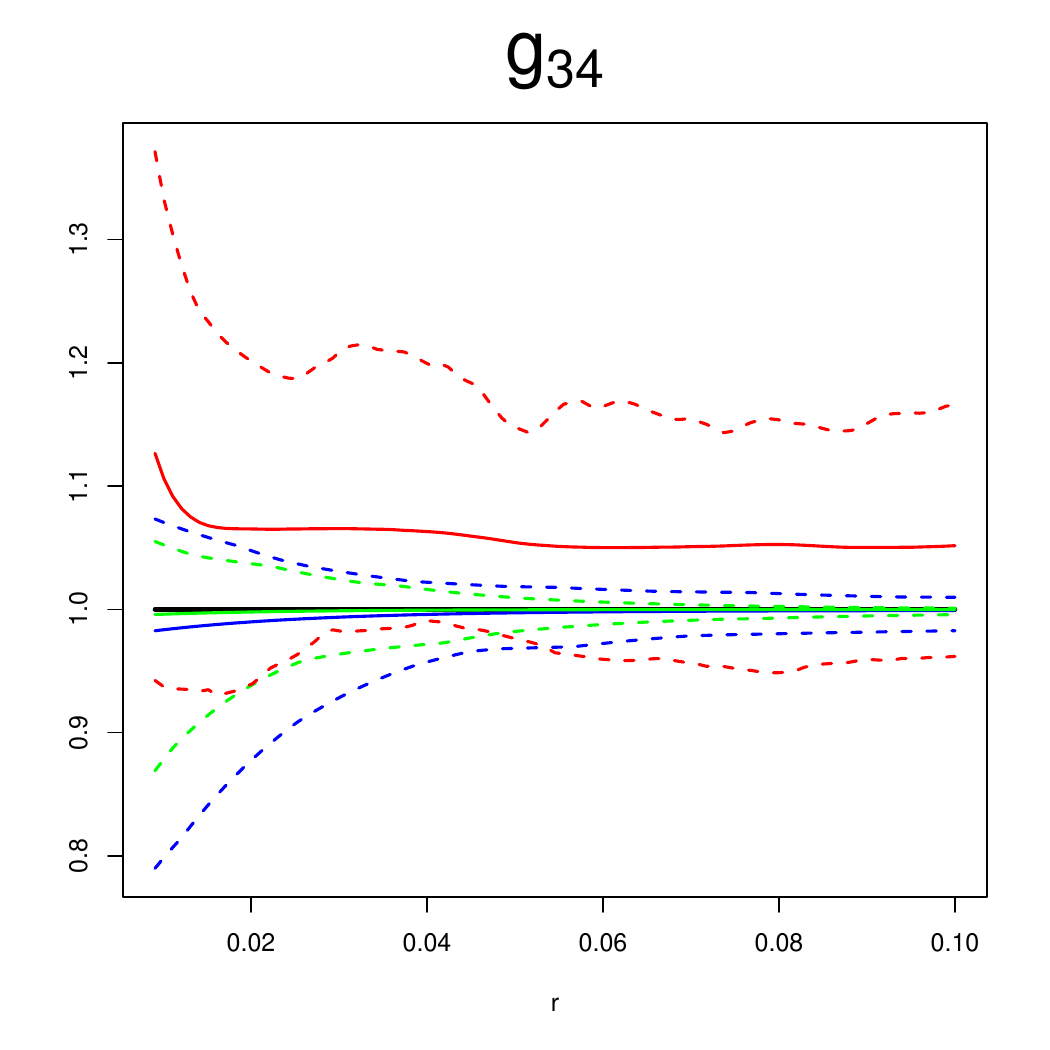}
\end{tabular}
\caption{(true $q=0$) Blue, green and red solid lines indicate pointwise means of estimates for selected cross PCFs using  MIN, 1-SE, or simple non-parametric estimation. The dotted lines indicate the corresponding $95 \%$ pointwise probability intervals. Black solid lines indicate true cross PCFs.} \label{fig:simpcf1}
\end{figure}

Table~\ref{tab:simmise1} gives total, within, and between MISEs with
different strategies for choosing $(q,\lambda)$ and for the two non-parametric approaches. The non-parametric approaches are clearly outperformed by the semi-parametric method. The results for MIN with $\lambda=0$ or $\lambda$ selected are very similar and also similar to 1-SE in case of $\text{MISE}_{\text{within}}$. However, $\text{MISE}_{\text{between}}$ for cross PCFs is more than twice as big for MIN compared to 1-SE. This is not so surprising since 1-SE chooses the true $q=0$ most of the time while MIN tends to choose larger values of $q$. The between MISEs are on the other hand on a much smaller scale than the within MISEs. 
\begin{table}
\caption{\label{tab:simmise1} MISE using $(q_{\min},0)$, $(q_{\min},\lambda)$, $(q_{\text{1-SE}},0)$, simple, or Diggle's non-parametric approach when
  true $q=0$.} 
\centering
\begin{tabular}{| l | c c c c c|}
\hline
& $(q_{\min},0)$ & $(q_{\min},\lambda)$ & $(q_{\text{1-SE}},0)$ & simple & Diggle\\
  \hline
  $\text{MISE}_{\text{total}}$ & 3.77$\cdot 10^{-4}$ & 3.78$\cdot 10^{-4}$ & 3.81$\cdot 10^{-4}$ & 2.22$\cdot 10^{-3}$ & 2.63$\cdot 10^{-3}$\\
$\text{MISE}_{\text{within}}$ & 1.02$\cdot 10^{-3}$ & 1.02$\cdot 10^{-3}$ & 1.11$\cdot 10^{-3}$ & 4.58$\cdot 10^{-3}$  & 3.36$\cdot 10^{-3}$ \\
$\text{MISE}_{\text{between}}$ & 5.45$\cdot 10^{-5}$ & 5.34$\cdot 10^{-5}$ & 2.01$\cdot 10^{-5}$ & 1.04$\cdot 10^{-3}$ & 2.27$\cdot 10^{-3}$ \\
  \hline
\end{tabular}

\end{table}

\subsection{Five-variate LGCP with two common latent fields} \label{sec:simstudy2}
 
In case of $q=2$, both MIN and 1-SE performs quite well in the sense that the chosen
$q$'s differ at most by one from the true $q$ in 98\% (MIN) or in
100\% (1-SE) of the cases and the true $q=2$ is chosen in 60\% (MIN)
or 61\% (1-SE) of the cases. The $\lambda$ column in
Table~\ref{tab:simcv} (right) shows 95\% probability intervals for the selected $\lambda$s. For $q=1,4$ the cross validation always selected $\lambda=0$. The selected $\lambda$s for $q=2,3$ are in general small and 80\% ($q=2$) or 95\% ($q=3$) of the $\lambda$s were selected to be zero. These results indicate that regularization is not pertinent in this case where the true $\balpha$ is not particularly sparse. Also the percentages of $\alpha_{ik}$ estimated to be zero are very small for $q=1,2,3,4$. In case of the overparametrized model $q=7$ we on the other hand do see an effect of regularization with larger selected $\lambda$s, and on average 26\% of the $\alpha_{ik}$s estimated to be zero.

Figure~\ref{fig:simpcf2} shows means and 95\% probability intervals
for selected estimated PCFs and cross PCFs obtained with MIN or 1-SE
without regularization.  In both cases MIN and 1-SE produce some bias
for the PCFs, which is expected as we specify the wrong model. As in
the case $q=0$, the non-parametric estimates are more biased
  than the semi-parametric estimates. 
\begin{figure}[ht!]
\centering
\begin{tabular}{c c c}
\includegraphics[width=0.33\textwidth]{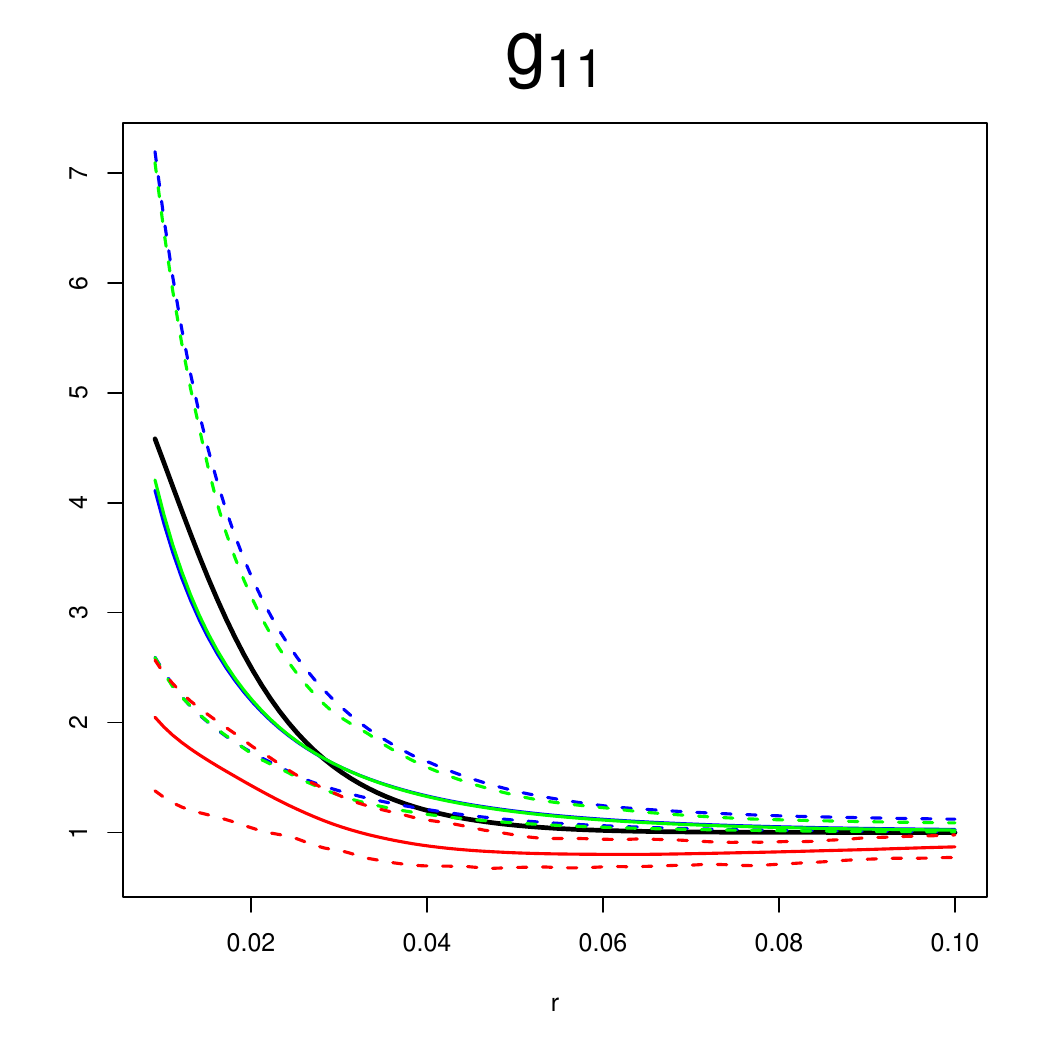}
\includegraphics[width=0.33\textwidth]{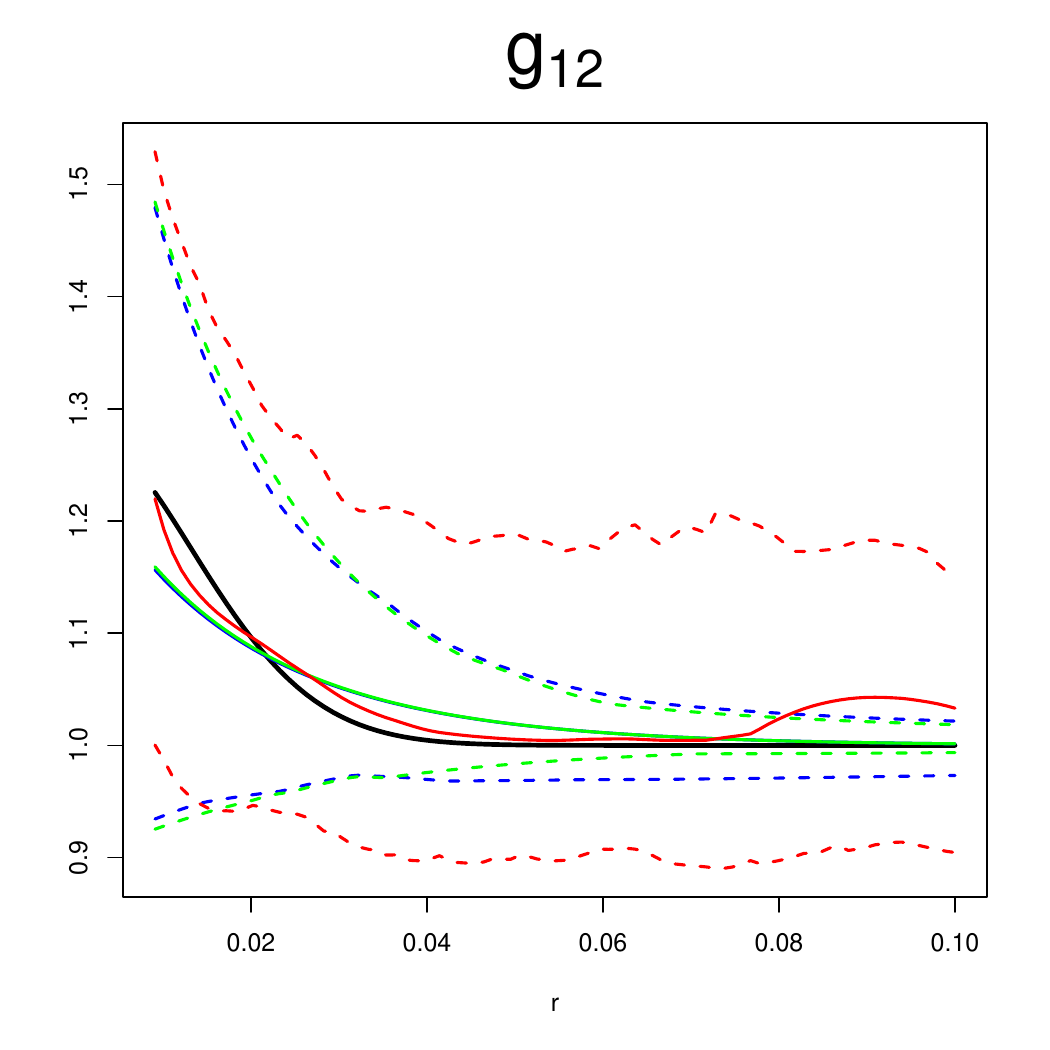}
\includegraphics[width=0.33\textwidth]{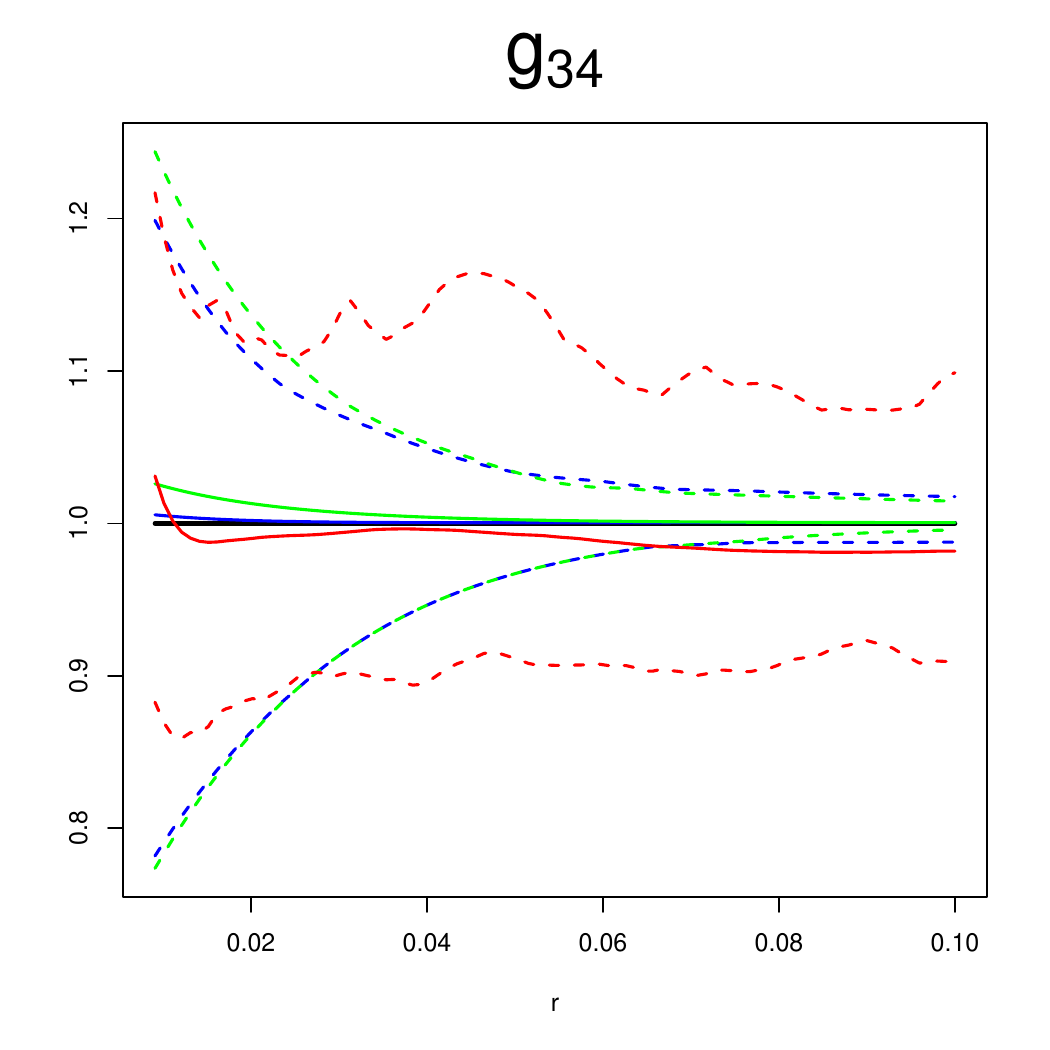} 
\end{tabular}
\caption{(true $q=2$) Blue, green and red solid lines indicate pointwise means of estimates for selected cross PCFs using  MIN, 1-SE, or simple non-parametric estimation. The dotted lines indicate the corresponding $95 \%$ pointwise probability intervals. Black solid lines indicate true cross PCFs.} \label{fig:simpcf2}
\end{figure}

Table~\ref{tab:simmise2} shows that MIN and 1-SE perform very similar regarding MISE. 
In case of $\text{MISE}_{\text{between}}$, MIN and 1-SE are somewhat better than the simple approach but much better than Diggle's approach. On the other hand, MIN and 1-SE are somewhat better than
Diggle's approach but much better than the simple approach in terms of
$\text{MISE}_{\text{within}}$. Overall ($\text{MISE}_{\text{total}}$) the semi-parametric method outperforms the non-parametric methods.
\begin{table}
\caption{\label{tab:simmise2} MISE using $(q_{\min},0)$, $(q_{\min},\lambda)$, $(q_{\text{1-SE}},0)$, the simple or Diggles non-parametric approach when
  true $q=2$.}
\centering
\begin{tabular}{| l | c c c c c|}
\hline
& $(q_{\min},0)$ & $(q_{\min},\lambda)$ & $(q_{\text{1-SE}},0)$ & simple & Diggle \\
  \hline
$\text{MISE}_{\text{total}}$ & 1.64$\cdot 10^{-3}$ & 1.64$\cdot 10^{-3}$ & 1.65$\cdot 10^{-3}$ & 7.25$\cdot 10^{-3}$ & 4.91$\cdot 10^{-3}$ \\
$\text{MISE}_{\text{within}}$ & 4.43$\cdot 10^{-3}$ & 4.42$\cdot 10^{-3}$ & 4.40$\cdot 10^{-3}$ & 2.04$\cdot 10^{-2}$ & 8.63$\cdot 10^{-3}$ \\
$\text{MISE}_{\text{between}}$ & 2.43$\cdot 10^{-4}$ & 2.43$\cdot 10^{-4}$ & 2.71$\cdot 10^{-4}$ & 6.76$\cdot 10^{-4}$  & 3.05$\cdot 10^{-3}$ \\
  \hline
\end{tabular}
\end{table}

\section{Data examples} \label{sec:dataanalysis}
In the following we apply our new methodology to point patterns of
cells in tumor tissue and crime scenes in Washington DC. We remind
that the parameter matrix $\balpha$ consists of coefficients for the common
latent fields, $\bsigma$ is the vector of standard deviations for the
type-specific latent fields, and $\bxi$ and $\bphi$ are vectors of
correlation scale parameters for the common and type-specific latent fields.

\subsection{Lymph node metastasis}\label{sec:cells}

Cancer malignancies are complex tissues containing many different
cell types that may have opposing roles in tumor growth \citep{valkenburg:groot:pienta:18}. It is now recognized that the interactions between
the multiple cell types determine the intrinsic aggressiveness
of a tumor and its response to a given anti-cancer
treatment. Studying each cell type separately
is therefore inadequate \citep{weinberg:14,tsujikawa:etal:17}. With
modern multiplexing technology it is now possible to image a large
number of cell types simultaneously. Figure~\ref{fig:datalymph} (left)
is a
specific example, which shows a fluorescence image of a lymph node
metastasis. Figure~\ref{fig:datalymph} also shows point patterns of locations of
four types of cells extracted from the image using machine learning
techniques.
\begin{figure}[ht!]
\centering
\begin{tabular}{cc}
\includegraphics[width=0.32\textwidth]{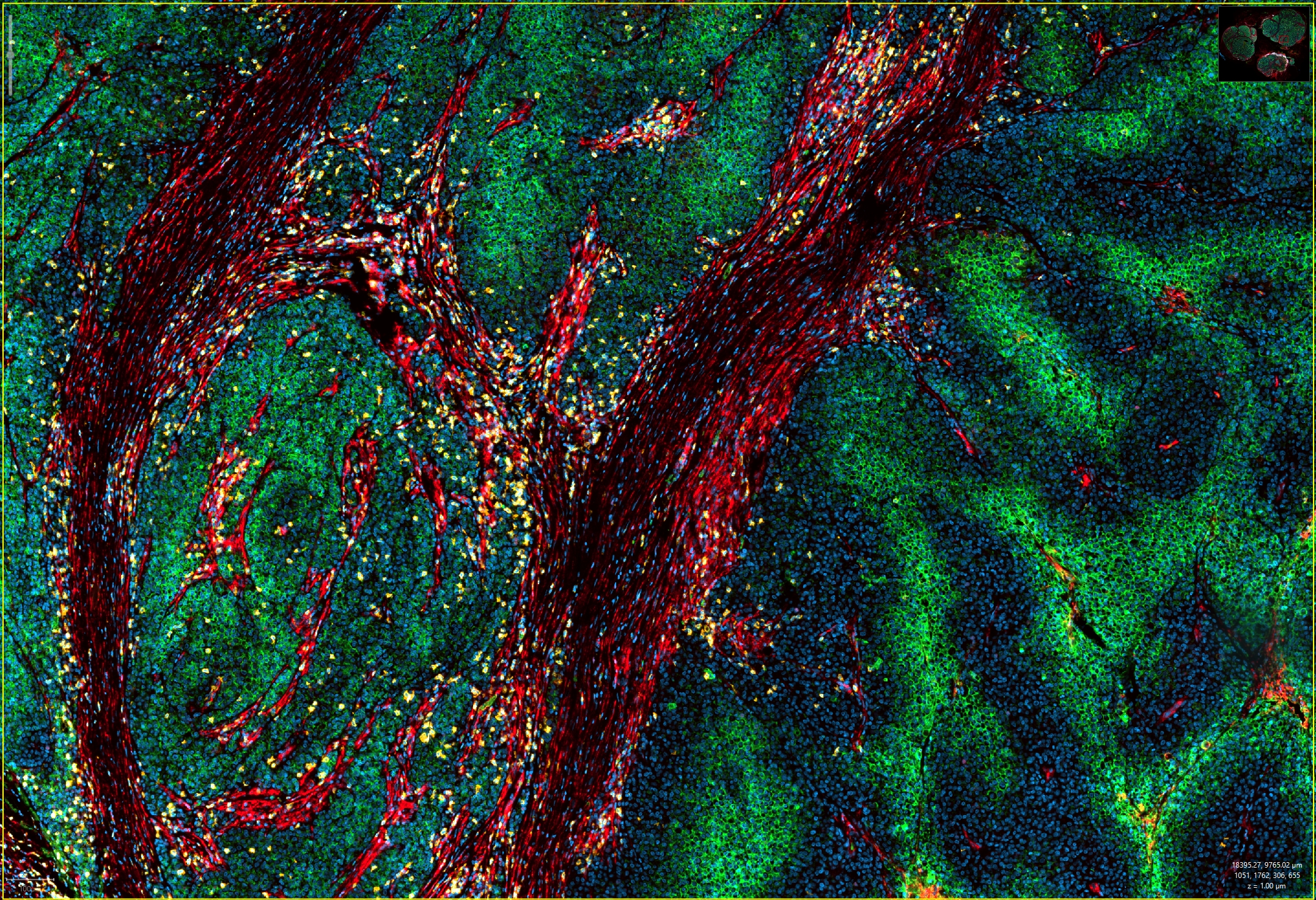}
\includegraphics[width=0.32\textwidth]{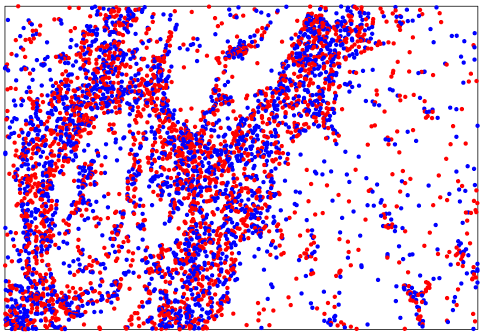}
\includegraphics[width=0.32\textwidth]{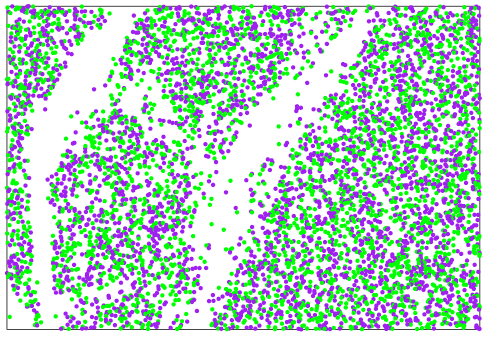}
\end{tabular}
\caption{Left: fluorescence image of a lymph node metastasis. Middle:
   bivariate point pattern of CD8 (blue) and $50 \%$ independently
   thinned Stroma (red) cells. Right: bivariate point pattern of $80
   \%$ independently thinned Hypoxic (purple) and $80 \%$
   independently thinned Normoxic (green) cells (data kindly provided by Arnulf Mayer, Dept.\ of Radiation Oncology, University Medical Center, Mainz, Germany).} \label{fig:datalymph}
\end{figure}
The four types of
cells (with abbreviated names and numbers of cells in parantheses) are Hypoxic tumor cells
(Hypoxic, 11733), Normoxic tumor cells (Normoxic, 18469), Stroma cells
(Stroma, 6015), and  Cytotoxic T-lymphocytes (CD8, 1466). For better visualization we only show random subsets
obtained by independent thinnings of the points.
Our aim is to
characterize the point patterns in terms of their intensity functions
and their PCFs and cross PCFs. We thereby quantify trends and within/between
interactions for the different types of cells. Such characterizations
may be used in subsequent studies e.g.\ for predicting outcomes of
cancer treatments \cite[see for example][]{herbst:etal:14,yan:etal:19}.

Figure~\ref{fig:kernelintensity} shows non-parametric
estimates of the intensity functions for the four point
patterns. These plots show a strong segregation between the
patterns of Stroma and CD8 cells versus the tumor cells. In the
following we study the more subtle  variation within the
bivariate point patterns of Stroma and CD8 respectively Normoxic and
Hypoxic.
\begin{figure}[ht!]
\centering
\begin{tabular}{c c}
\includegraphics[width=0.32\textwidth]{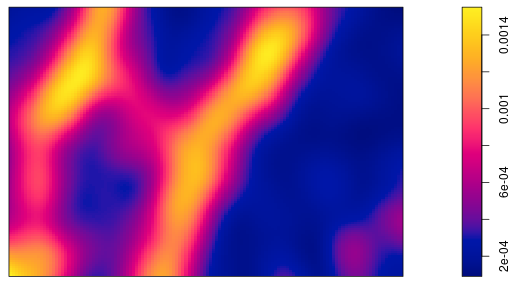} &
\includegraphics[width=0.32\textwidth]{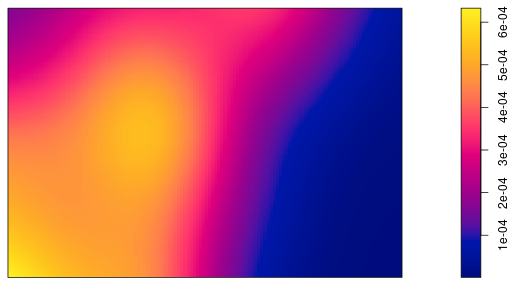}\\
\includegraphics[width=0.32\textwidth]{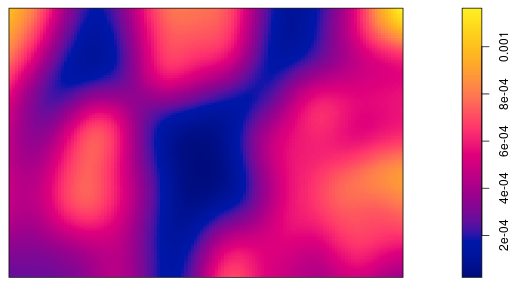}&
\includegraphics[width=0.32\textwidth]{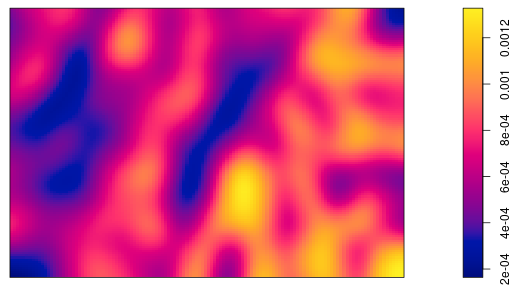}
\end{tabular}
\caption{ Kernel estimates of the intensity functions for Stroma
  (upper left), CD8 (upper right), Hypoxic (lower left) and Normoxic
  (lower right) with bandwidths  121.2, 359.5, 173,9 and 115.7, respectively.} \label{fig:kernelintensity}
\end{figure} 

There are no spatial covariates available for this data set. The
intensity functions are therefore proportional to the common component
$\rho_0(\cdot)$ both for the pairs Stroma,CD8 and Normoxic,Hypoxic. Since
the point patterns are of high cardinality we reduce computing time by
working with independent thinnings of the point patterns. The PCFs and
cross PCFs are invariant to independent thinning while the intensity
functions are only changed by a multiplication with the thinning
probability. In the following we present a detailed analysis of the Stroma-CD8 point pattern. The analysis for the Normoxic-Hypoxic tumor cells is quite similar and is presented in Section~8 of the supplementary material.

\subsubsection{Stroma and CD8} \label{sec:stroma}

For Stroma and CD8 we use all CD8 points and independently thin the
Stroma points with a thinning probability of 50\%. The point patterns
clearly show some large scale trends (Figure~\ref{fig:kernelintensity}) that are not easily fitted by simple
parametric models. We instead assume the model \eqref{eq:multi1},
choose CD8 as the baseline, 
and following \cite{hessellund:etal:20} estimate
$\bbeta=(\beta_{\textup{Str}},\beta_{\textup{CD8}})^\T=(\gamma_{\textup{Str}}-\gamma_{\textup{CD8}},\gamma_{\textup{CD8}}-\gamma_{\textup{CD8}})^\T$
by $\hat{\beta}_{\textup{Str}}= \textup{log} (3007/1466)=0.72$ and
$\hat{\beta}_{\textup{CD8}}=0$. We next choose $q$ among the values
$\{0,1,2\}$ using a $5$-fold cross validation as described in Section
\ref{sec:cvlambda}, where we resample $L=10$ times. We choose the
maximal interpoint distance $R$ for pairs of points to be $400$
$\mu$m which
corresponds to approximately 15\% of the largest observation window
side length.
\begin{figure}[ht!]
\centering
\begin{tabular}{c c}
\includegraphics[width=0.20\textwidth]{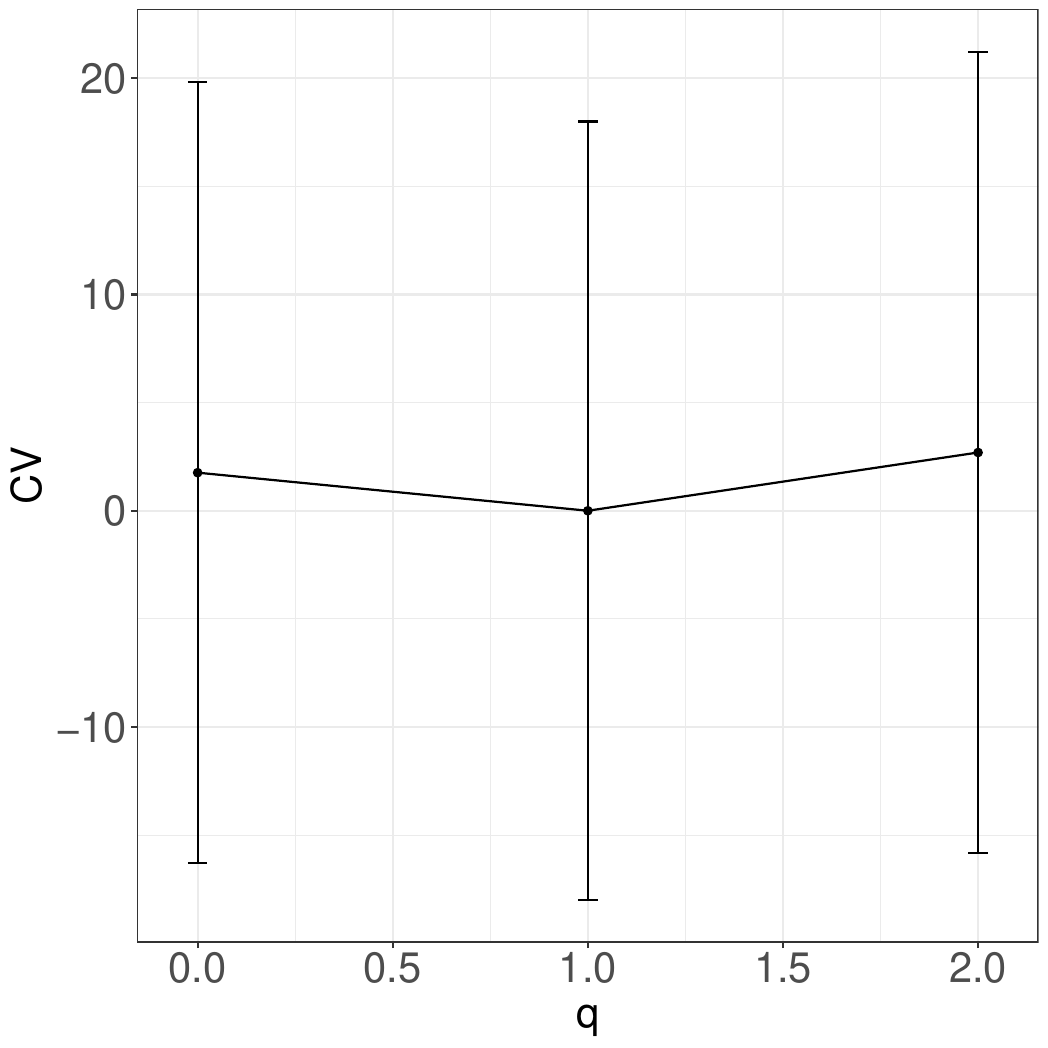}
\hspace{0.5cm}
\includegraphics[width=0.37\textwidth]{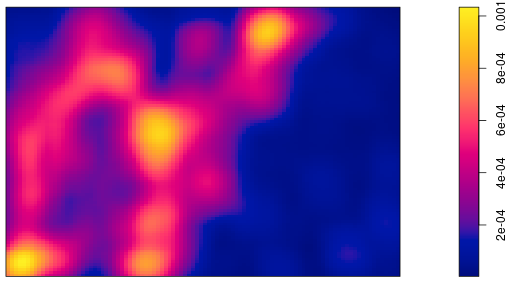}
\end{tabular}
\caption{Left: CV-scores (minus minimum CV-score) with standard errors. Right: non-parametric estimate of $\rho_0$ with bandwidth = 194.5.} \label{fig:cvlymph}
\end{figure}
According to the left panel in Figure~\ref{fig:cvlymph} we choose 
$q=1$ that minimizes the cross validation score. The right panel in Figure~\ref{fig:cvlymph} shows a non-parametric estimate of $\rho_0$ using the estimator introduced in \cite{hessellund:etal:20} with bandwidth chosen as described in Section~5 in the supplementary material.
\begin{table}
\caption{\label{fig:lymphresult} Parameter estimates for Stroma and CD8 for $q=1$.} 
\centering
\fbox{%
\begin{tabular}{| c |c c c  c| c |c c c  c|}
\hline
  & $\hat{\balpha}$ & $\hat{\bxi}$ & $\hat{\bsigma}$& $\hat{\bphi}$ & & $\hat{\balpha}$ & $\hat{\bxi}$ & $\hat{\bsigma}$& $\hat{\bphi}$ \\
\hline
Stroma & 0.52 & 63.4 & 0.32
 & 97.7 &CD8 & -0.52 & 63.4 & 0.78
& 193.6 \\
  \hline
\end{tabular}}
\end{table}

According to the parameter estimates in Table~\ref{fig:lymphresult}
and the resulting PCFs and cross PCFs shown in the left panel of
Figure~\ref{fig:lymphplot}, both
Stroma and CD8 are randomly clustered point processes. The
clustering is partly negatively correlated (cf.\
$\hat{\balpha}$ and the fitted cross PCF in Figure~\ref{fig:lymphplot}) and partly independent (cf.\ $\hat{\bsigma}$) between
Stroma and CD8. The strongest clustering is found for  CD8 due to the
higher value of $\hat \sigma_{\text{CD8}}$ than $\hat
\sigma_{\text{Str}}$, see also the fitted PCFs in Figure~\ref{fig:lymphplot}.
\begin{figure}[ht!]
\centering
\begin{tabular}{c c}
\includegraphics[scale=0.25]{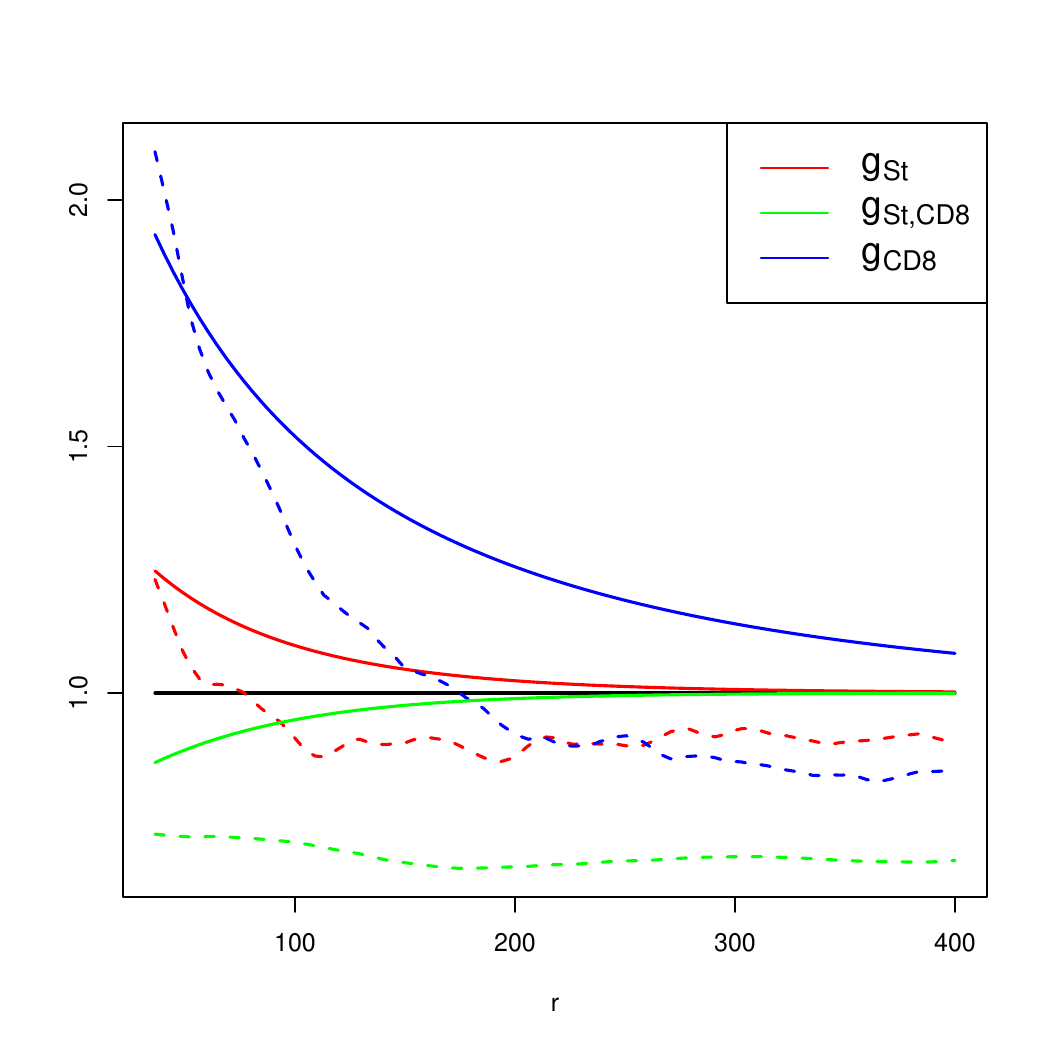}
\includegraphics[scale=0.25]{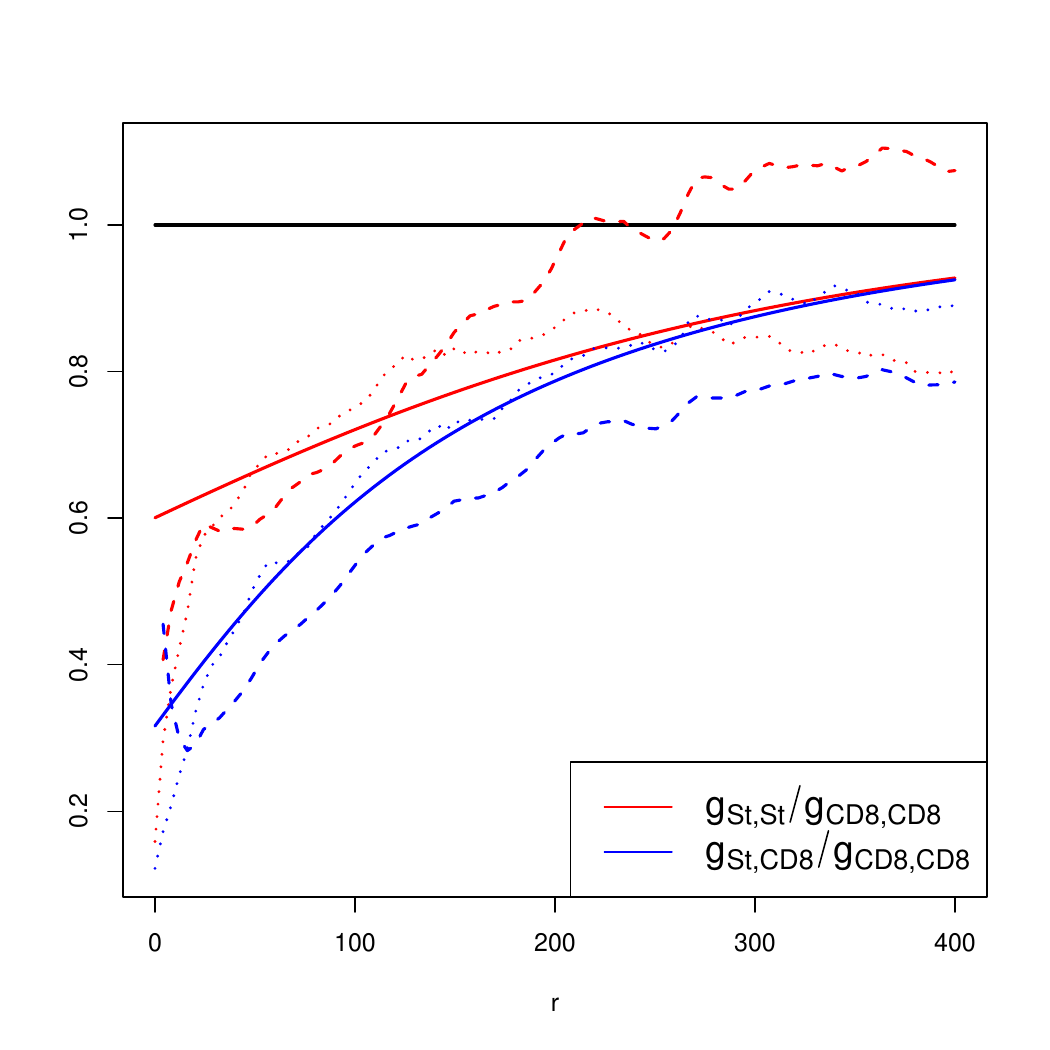}
\end{tabular}
\caption{Left: estimated (cross) PCFs  using the
  semi-parametric model (solid) and simple  approach
  (dashed). Right: estimated (cross) PCF ratio using
  semi-parametric model (solid), simple approach (dashed) and consistent approach (dotted).} \label{fig:lymphplot}
\end{figure}

The total estimated variances for the log random intensity
functions of Stroma and CD8 are rather moderate, respectively 0.37 and
0.88, while the empirical variance of $\log \hat \rho_0$ over the
observation window is 1.15. In this sense, the majority of the
variation in the random intensity functions  (especially for Stroma)
is explained by $\rho_0$. 

Following Section~\ref{sec:modelassess}, the right panel in
Figure~\ref{fig:lymphplot} compares semi-parametric estimates of cross
PCF ratios $g_{\text{Str}}/g_{\text{CD8}}$ and
$g_{\text{Str,CD8}}/g_{\text{CD8}}$ with the non-parametric estimates introduced in  
\cite{hessellund:etal:20}. The agreement seems reasonable and this is
confirmed by global envelope $p$-values of 0.05 in case of
$g_{\text{Str}}/g_{\text{CD8}}$ and 0.09 for
$g_{\text{Str,CD8}}/g_{\text{CD8,CD8}}$, see also the global
envelope plots in Section~9 in the supplementary
material. Section~11 in the supplementary material further comments on model
assessment using the inhomogeneous cross $J$-function \citep{cronie:lieshout:16}.

Figure~\ref{fig:lymphplot} (left) also shows simple non-parametric PCF and cross PCF estimates which are generally smaller than the semi-parametric estimates. In particular, the
non-parametric estimate of the cross PCF suggests a strong negative
correlation between Stroma and CD8 for all spatial lags considered.
As discussed in \cite{shaw:moeller:waagepetersen:20}, this might be due to that the selected bandwidths imply too little
smoothing in the non-parametric intensity estimates (upper plots in
Figure~\ref{fig:kernelintensity}). 
Figure~\ref{fig:lymphplot} (right) further shows that the
simple estimates of cross PCF ratios deviate more from the consistent non-parametric estimates than the semi-parametric estimates.



\subsection{Washington DC street crimes}\label{sec:crimes}
It is of great interest for criminologists and police authorities to
study the spatial patterns of crime scenes since this can lead to
better understanding of factors affecting crime and more efficient
policing strategies. 
There are several competing theories regarding how crime occurrence depends on
environmental factors \citep{weisburd:etal:93,haberman:17}. Such
theories may be studied using the intensity model \eqref{eq:multi1}. Next,
estimates of PCFs and cross PCFs \eqref{eq:crosspcf1} may be used to assess whether
crime occurrence is fully explained by the available covariates
through the intensity function. If estimated PCFs and cross PCFs differ
markedly from one, this is a sign of unexplained random variation that may
call for further criminological investigation. Further, if cross correlations exist
between different crimes, this means that predictions of one type of
crime may be enhanced by consideration of other types of crimes.

In this section, we focus on the spatial
correlation between six common types of street crimes committed in
Washington DC in January and February 2017.  
\begin{figure}[ht!]
	\centering
	\begin{tabular}{ll}
		\includegraphics[width=0.45\textwidth]{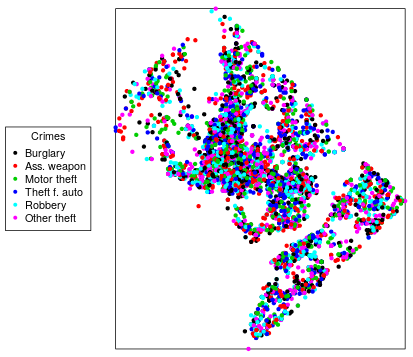} &
		\includegraphics[width=0.4\textwidth]{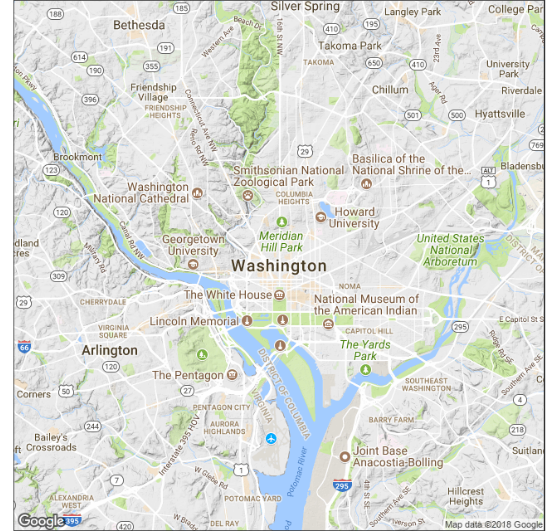} \\
	\end{tabular}
	\hspace{0.01cm}    \caption{Left: street crimes locations ($n=5378$). Right: a map of Washington DC. The data set is extracted from a
larger data set  publicly available from $\mathbf{http {:}//opendata.dc.gov/datasets/}$.}
	\label{fig1}
\end{figure}
The six types of crimes with numbers in parantheses are 1)
 Burglary (259), 2) Assault with  weapon (332), 3) Motor
vehicle theft (335),  4) Theft from automobile (1832), 5) Robbery
(366), and 6) Other theft (2254). The locations of the six types of crimes are illustrated in Figure~\ref{fig1}. This data set has previously been
considered by \cite{hessellund:etal:20} who focused on the dependency
of the street crime intensity functions on spatial covariates using
the model \eqref{eq:multi1}. In the
following we focus on the second order properties as described by the
PCFs and cross PCFs. We refer to  \cite{hessellund:etal:20}
for more details regarding the covariates and the fitted intensity
functions.

We apply the regularized estimation approach
described in Sections~\ref{sec:lasso}-\ref{sec:cvlambda}
where we first determine $q$ using cross validation without
regularization and next, for the chosen $q$, use another cross
validation to select the regularization parameter $\lambda$ to
potentially obtain a sparse submodel for the chosen $q$. For the cross
validation, we use $K=5$ and $L=10$, and choose $q$ in
$\{0,1,\ldots,5\}$ and $\lambda$ in $\{100,80,60,50,40,30,20,10,5,2.5,1,0.25,0\}$.
For the second order composite likelihood we use $R=1000$ meters.

The left panel in Figure \ref{fig:cvcrime} shows cross validation
scores for each $q$, where the 1-SE criterion leads to choosing $q=0$
while MIN chooses $q=1$. The middle panel shows CV-scores for each
$\lambda$ with $q=1$, where the minimum is obtained for $\lambda=0$.
The right panel in Figure \ref{fig:cvcrime} shows a non-parametric
estimate of $\rho_0$ using the estimator described in
\cite{hessellund:etal:20} with bandwidth chosen as described in
Section~3.3 in the supplementary
material. This is our current best proposal for
  a data-driven estimate of $\rho_0$. It may, however, seem too smooth
  since $\rho_0$ is supposed to  capture complex variation in the
  intensity due to urban structure and varying population density.
In the
following we focus on the results with $q=1$ and $\lambda=0$. Hence we obtain an estimate of $\balpha$ without regularization.
\begin{figure}[ht!]
\centering
\begin{tabular}{c c c}
\includegraphics[scale=0.2]{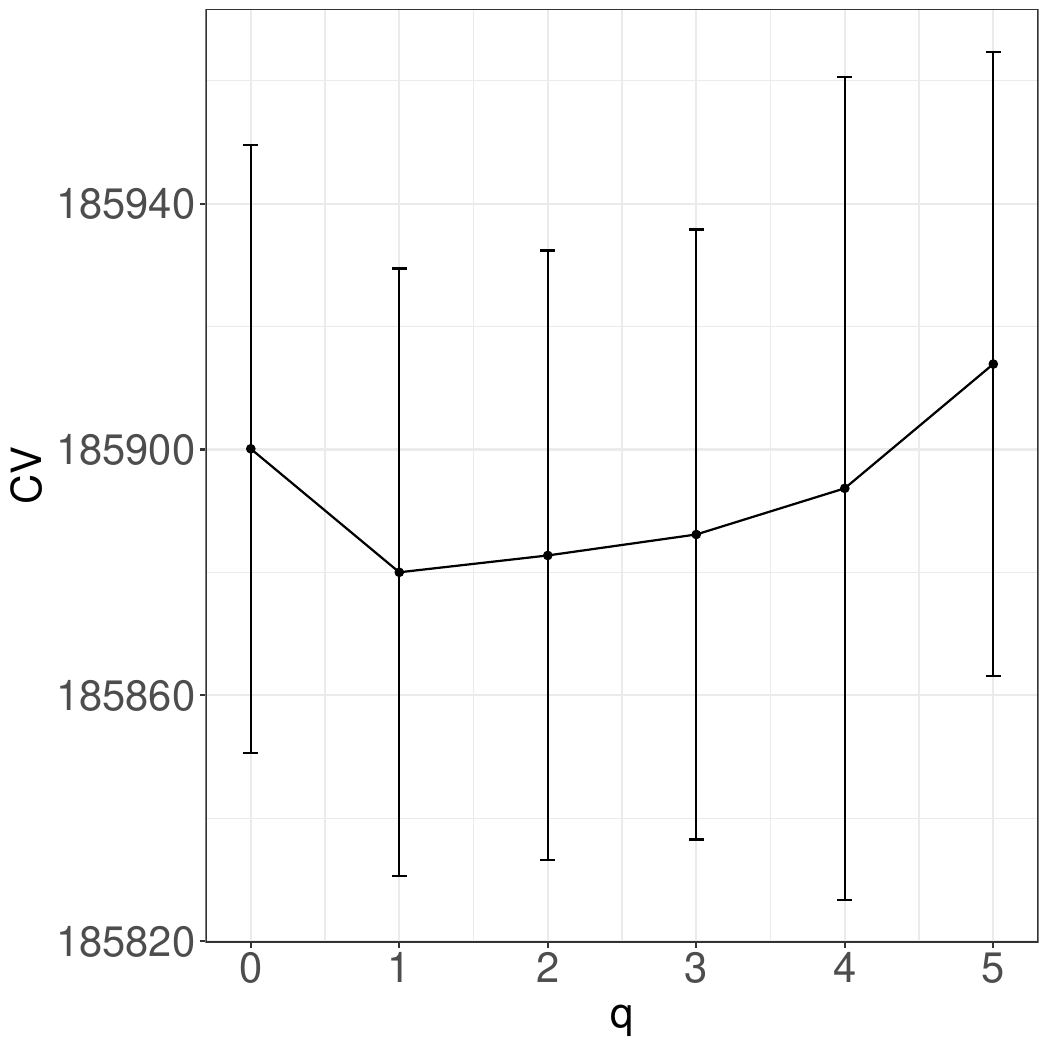}
\hspace{0.1cm}
\includegraphics[scale=0.2]{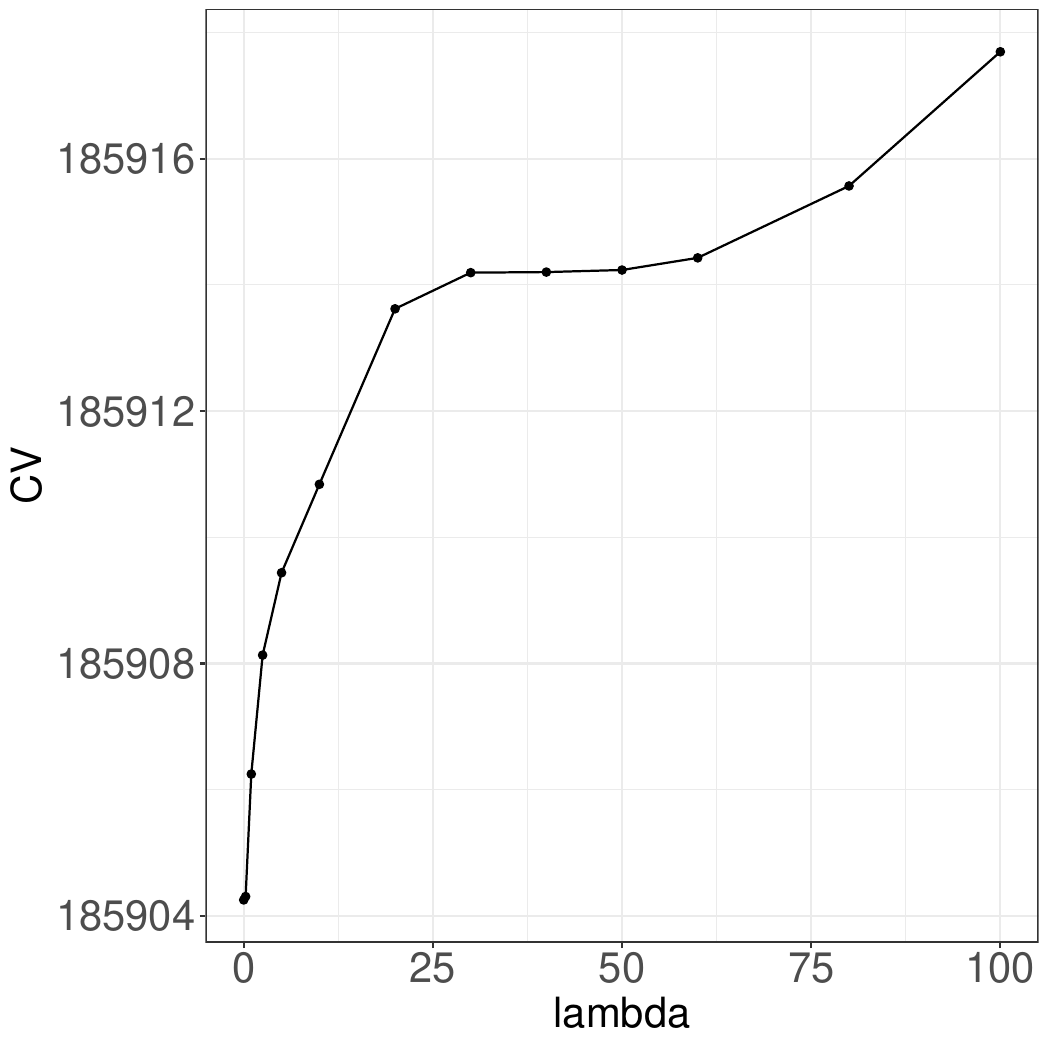}
\hspace{0.1cm}
\includegraphics[scale=0.28]{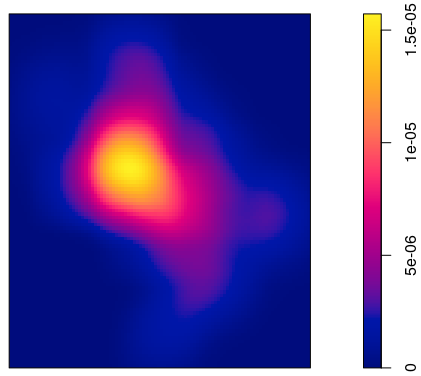}
\end{tabular}
\caption{Left: CV-score for $q$ with one standard error bars. Middle:
  CV-score for $\lambda$ given $q=1$. Right: non-parametric estimate
  of $\rho_0$ with bandwidth 2654 meters.} \label{fig:cvcrime}
\end{figure}


The parameter estimates for each street crime are given in
Table~\ref{fig:crimepar} except for the common latent field
correlation scale parameter estimate which is $\hat{\xi}=102.5$. The $\sigma_i$ estimates are
small to moderate for the first five crimes while the estimate $\hat \sigma_6$ for Other theft is about two times larger than the other $\sigma_i$ estimates. Regarding the latent field
$Y_1$, the $\alpha_{i1}$ estimates are pretty small for Assault,
Vehicle theft, Theft from auto, and Robbery while $\alpha_{11}$ for
Burglary and $\alpha_{61}$ for Other theft have fairly large estimates
$0.78$ and $-0.93$. The resulting estimated PCFs and cross PCFs are
shown in the left and middle panels of
Figure~\ref{fig:crosscrime}. The overall conclusion is that most
crimes are moderately clustered except for Burglary and Other theft
with strongest clustering for Other theft. Also the cross dependencies
seem fairly weak except for the pairs Burglary and Vehicle
  (crimes 1 and 3, positively correlated) and Burglary and Other theft
  (crimes 1 and 6, negatively correlated). 
The interpretation of these results is that except
for moderate random fluctuations, the spatial patterns of Assault, Vehicle
theft, Theft from auto and Robbery are quite well described by their intensity functions depending on
the common factor $\rho_0$ as well as covariate effects. On the other
hand, the random intensity functions for Burglary and Other theft seem
subject to more pronounced deviations from the intensity functions, and
these deviations are negatively correlated. In other words, if a
cluster of Burglaries not explained by the intensity function is
present in a certain area, then there tends to be less Other theft
committed in the same area and vice versa. 
\begin{table}
\caption{\label{fig:crimepar} Table of parameter estimates for each street crime for $(q,\lambda)=(1,0)$. Last two columns show estimates of  $\balpha_{\cdot l}$, $l=1,2$ with $(q,\lambda)=(2,2.5)$.}
\centering
\begin{tabular}{| c | c c  c | c c|}
\hline
Crime type  &  $\hat{\balpha}$ &  $\hat{\bsigma}$  & $\hat{\bphi}$  & $\hat{\balpha}_{ \cdot 1}$ &  $\hat{\balpha}_{\cdot 2}$\\
\hline
Burglary   & 0.78 & 0.50 & 245.8 & 0 & 0.76 \\
Assault & -0.12 & 0.51 & 457.5 & 0 & -0.09 \\
Vehicle Theft  & 0.49 & 0.14 & 20.5 & 0 & 0.47 \\
Theft F. Auto  & 0.09 & 0.58 & 2483.1 & 0 & 0.08 \\
Robbery  & -0.30 & 0.53 & 485.2 & 0 & -0.26 \\ 
Other theft & -0.93 & 0.96 & 20.5 & 0 & -0.97 \\
  \hline
\end{tabular}  
\end{table}

We also tried out $q=2$ for which the cross validation score is quite
close to the one for $q=1$. For $q=2$ the cross validation selected
$\lambda=2.5$. The last columns in Table~\ref{fig:crimepar} show the estimate of $\balpha$ obtained with $(q,\lambda)=(2,2.5)$. The lasso regularization has shrunk $\hat{\balpha}_{\cdot 1}$ to $\boldsymbol{0}$, while the estimate of $\balpha_{\cdot 2}$ is quite similar to the estimate of $\balpha_{\cdot1}$ for $q=1$. In view of this, one may argue that the lasso regularization makes our estimation approach more robust, since a too large selected $q$ can be counterbalanced by regularization on $\balpha$ with a $\lambda>0$.

Quite different conclusions are obtained with the simple
non-parametric analysis.
Figure~\ref{fig:crosscrime} shows that the non-parametric estimates
of the PCFs and cross PCFs are all considerably above the reference
value one, which would imply strong clustering within and between the
different types of crime. These results may well be
explained  by bias of the non-parametric estimates.
\begin{figure}[ht!]
\centering
\includegraphics[width=0.32\textwidth]{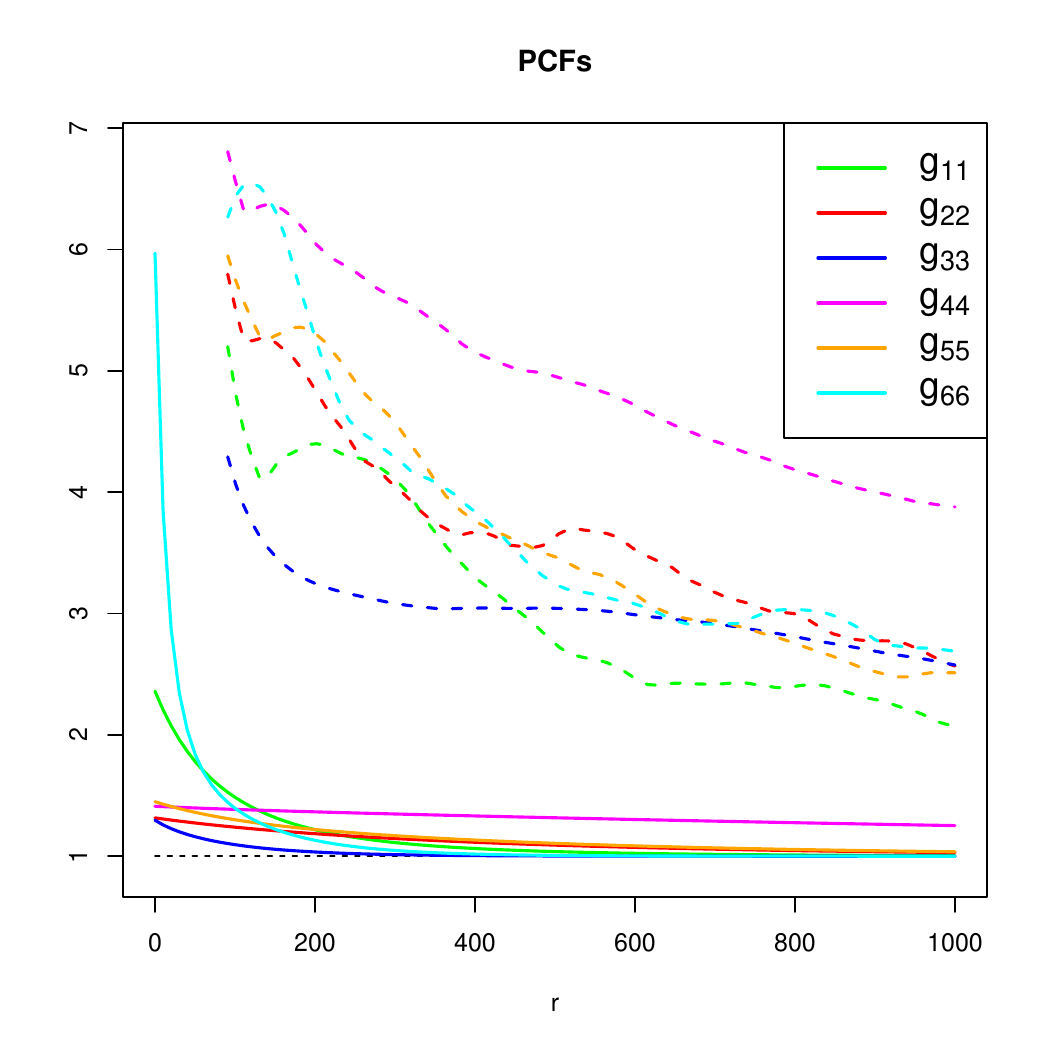}
\includegraphics[width=0.32\textwidth]{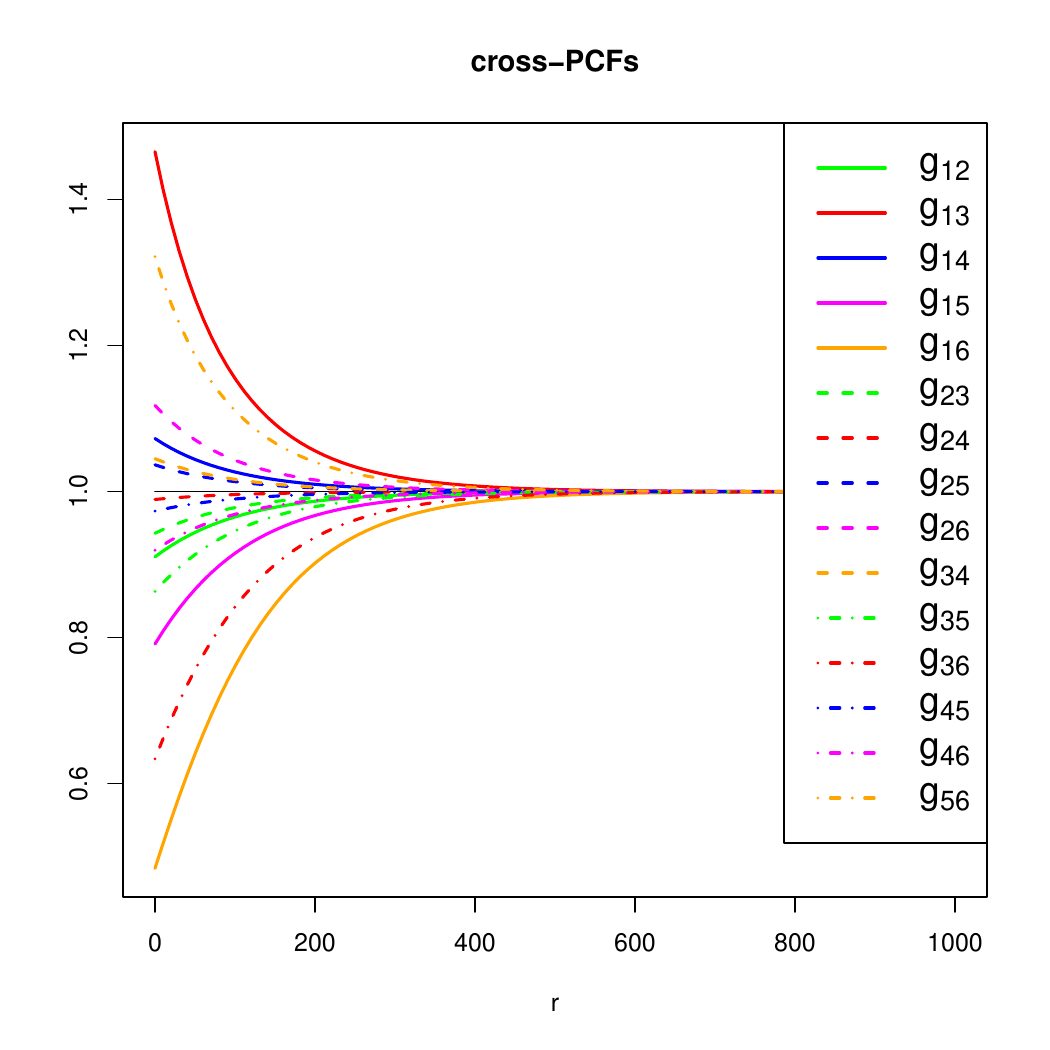}
\includegraphics[width=0.32\textwidth]{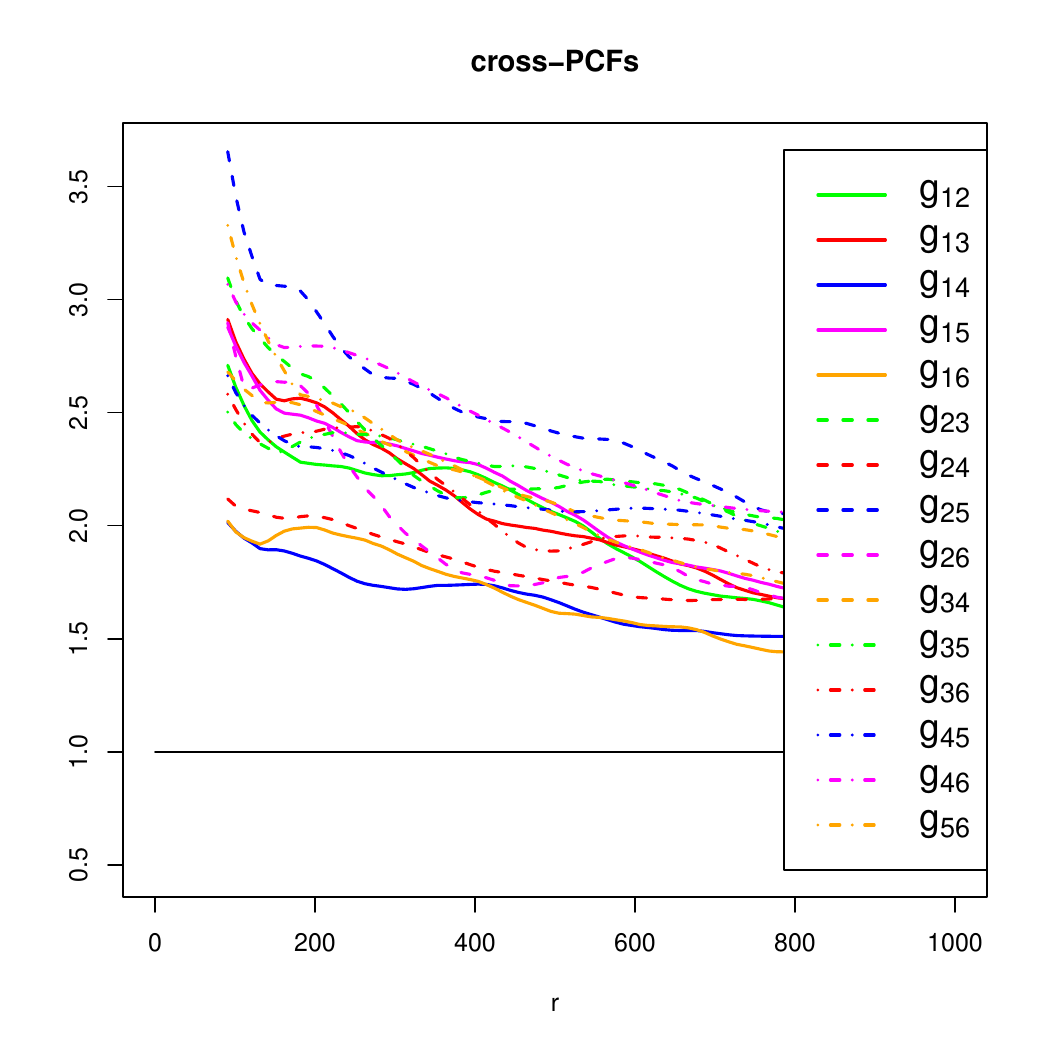} 
\caption{Left: semi-parametric (solid) and simple non-parametric (dashed) estimates of PCFs for $(q,\lambda)=(1,0)$. Middle: semi-parametric estimates of cross PCFs for $(q,\lambda)=(1,0)$. Right: simple non-parametric estimates of cross PCFs.} \label{fig:crosscrime}
\end{figure}

For model assessment, we consider global envelope tests based on
differences between semi-parametric and consistent non-parametric
estimates for all 20 ratios $g_{ij}/g_{66}$, $1\le i \le j \le 6$,
$(i,j)\neq (6,6)$. The $p$-values obtained are between 0.089 and 0.624,
and hence do not provide evidence against our model. Some
representative global envelope plots for the differences are shown in
Section~10 in the supplementary
material. Further model assessment using the inhomogeneous cross
  $J$-function is discussed in Section~11 in the
  supplementary material.

We finally consider an explorative analysis focusing on patterns in
the common latent process $Y_1$. We define `residuals' $\Delta \log \Lambda_i(\bu)$ 
by $\log \Lambda_i(\bu) - \mu_i - \bbeta_i^\T
  \bz(\bu) - \frac{1}{p} \sum_{l=1}^p [ \log
  \Lambda_l(\bu) - \mu_l - \bbeta_l^\T \bz(\bu)]$. Due to the
  sum-to-zero constraint on $\balpha$ we obtain $\Delta \log
  \Lambda_i(\bu)= \alpha_{i1} Y_1(\bu) +   \sigma_i U_i - \frac{1}{p} \sum_{l=1}^p \sigma_l U_l(\bu)$.
Estimating $\Delta \log \Lambda_i$ by replacing $\Lambda_i$
by a kernel estimate and the parameters by their conditional
likelihood estimates, we  obtain $\wh Y_1(\bu) = (\hat \balpha^\T \hat \balpha)^{-1} \hat \balpha [\wh \Delta
  \log \Lambda_1(\bu),\ldots, \wh \Delta \log \Lambda_p(\bu)]^\T.$
The left plot in Figure~\ref{fig:explorative} shows $\wh Y_1$, where
the $\Lambda_i$ are estimated by kernel smoothing using a bandwidth of
3 km. There is some resemblance between  $\hat Y_1(\bu;h)$ and the
spatial distribution of median income shown in the middle plot of
Figure~\ref{fig:explorative}. Log median income is included as a
covariate in the regression model for the log intensity. It may
  therefore be
the case that $\hat
Y_1(\bu)$ reflects nonlinear effects of the financial status of a
neighborhood, cf.\ the right plot in Figure~\ref{fig:explorative}.
\begin{figure}[ht!]
	\centering
	\begin{tabular}{lll}
	\includegraphics[width=0.3\textwidth]{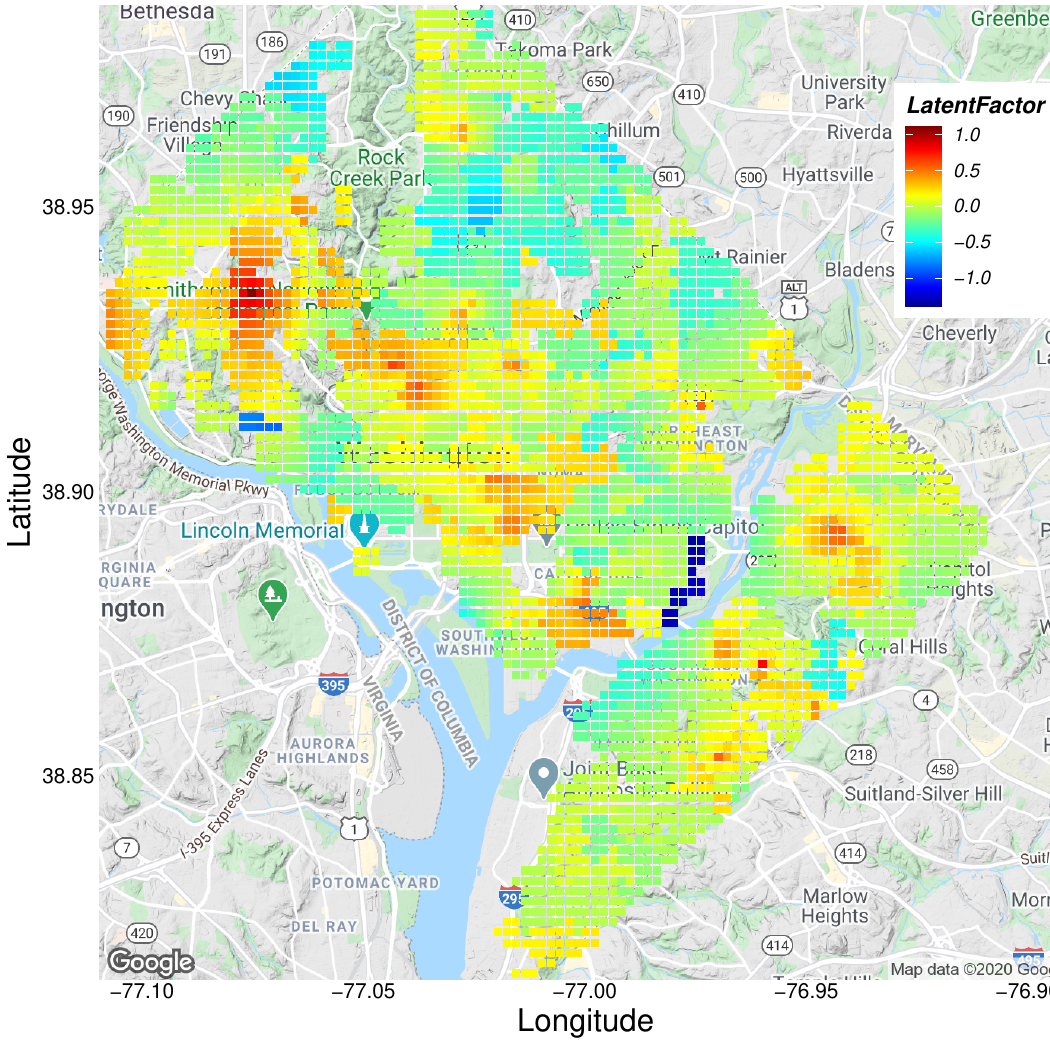} &	\includegraphics[width=0.3\textwidth]{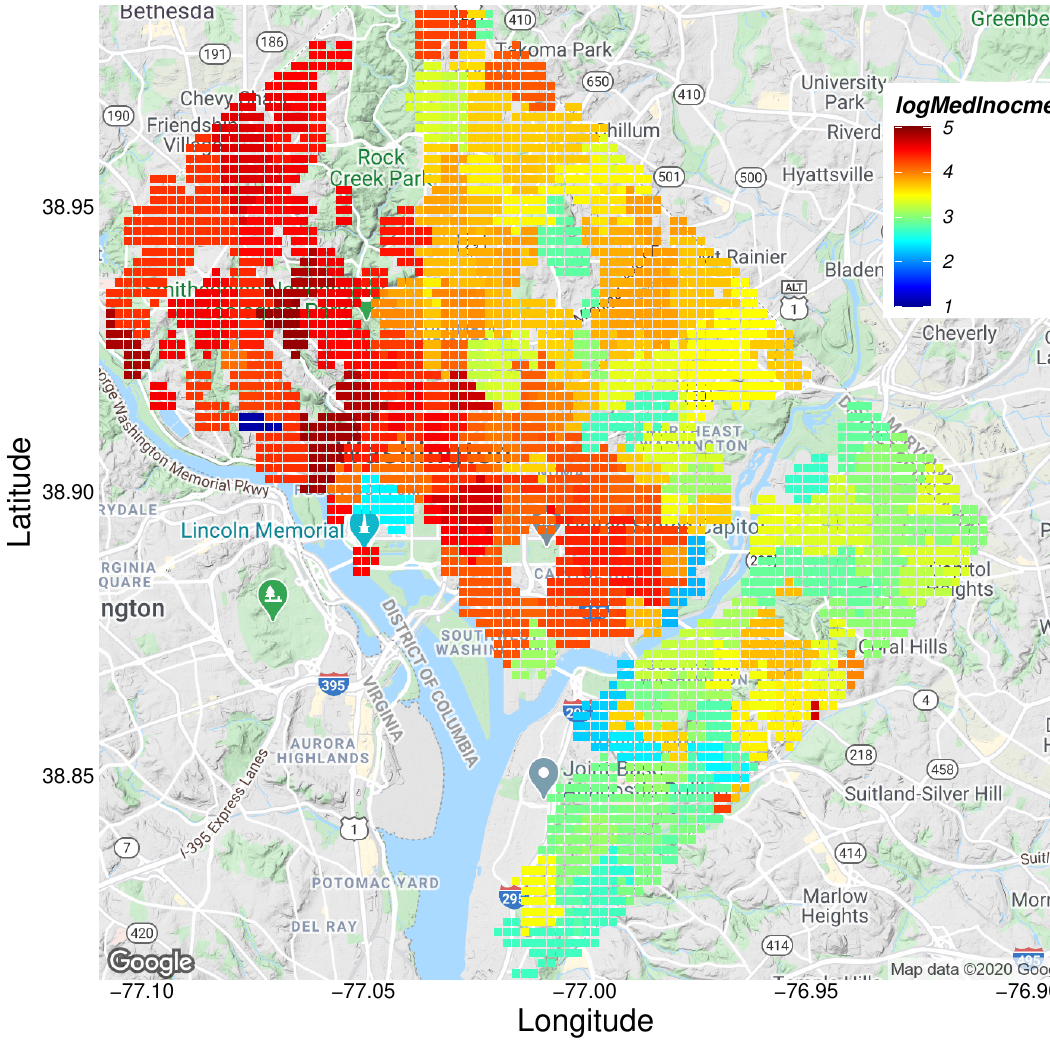} &
	\includegraphics[width=0.32\textwidth]{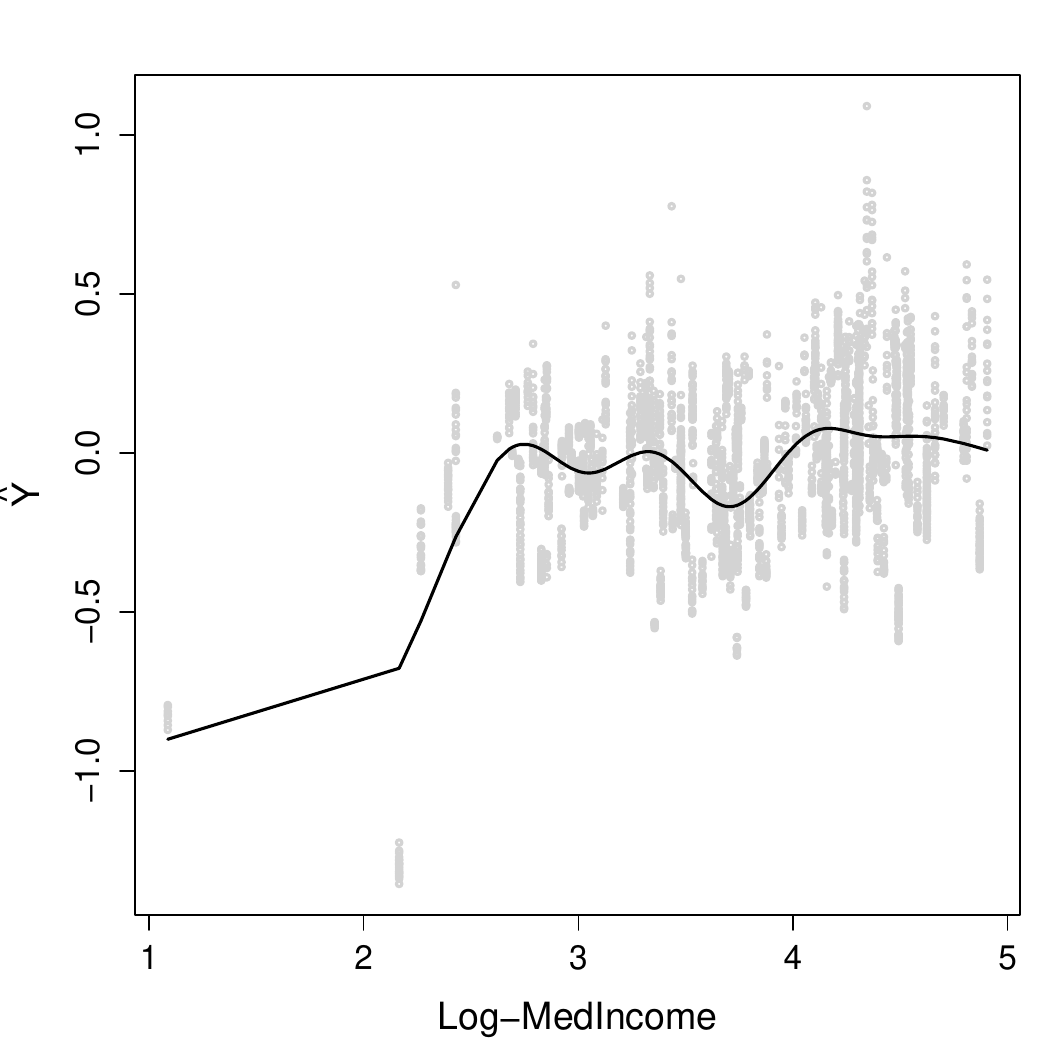}\\
\end{tabular}
	\caption{Left: latent factor $\wh Y_1$. Middle: log median income within census tracts. Right: $\hat Y_1$ versus log median income.}
	\label{fig:explorative}
\end{figure}

\section{Conclusion} \label{sec:con}

The methodology introduced in this paper provides a
major step forward regarding second order analysis of multivariate
point processes with complex intensity functions. Existing approaches
\cite[such as simple non-parametric estimation or the approach
in][]{diggle:etal:07} rely on estimating the intensity functions using
kernel estimators. This tends to result in strong bias and/or large
variance for subsequent estimation of PCFs and cross PCFs. In
contrast, in the context of the model \eqref{eq:multi1},  our approach
circumvents the need to estimate the complex unknown intensity
function factor $\rho_0$. According to our simulation studies, the resulting PCF and cross PCFs appear
to be close to unbiased. For the
data examples considered, we obtain simple and interpretable models
that may result in better understanding of the interplay between
respectively cells in tumors and different types of crimes. 

A limitation of our approach, shared with existing methods,
is that we have not provided confidence intervals for parameter
estimates or confidence bands for estimated PCFs or cross PCFs. One
topic for further research would be to establish asymptotic results
for parameter estimates within the framework of estimating function
inference. This was done by \cite{hessellund:etal:20} regarding inference for the intensity function but the current problem of
inferring cross PCFs entail considerable additional theoretical
difficulties.

The impact of using regularization was not very strong in our
simulation studies when moderate values of $q$ were
considered. However, the crimes data example indicates that the use of
regularization may add robustness to the estimation procedure if a too
large $q$ is selected.

The spatial crime point pattern data were obtained by
  aggregating data over consecutive two months. The implications of this particular choice of
  aggregation are not obvious. A highly interesting topic for
  further research would be to extend our methodology to a space-time
  analysis. One could e.g.\ envisage a space-time multivariate LGCP with
  latent common and type specific Gaussian fields evolving over time.

We finally mention that other useful cross summary statistics
  than cross PCFs are
available for studying cross dependencies in inhomogeneous
multivariate point processes. The inhomogeneous cross $K$-function
\citep{moeller:waagepetersen:03} is basically a cumulative version of
the cross PCF. Profoundly different cross summary statistics are the
inhomogeneous cross empty space function, nearest neighbour
distance function and $J$-function introduced in
\cite{cronie:lieshout:16}. As discussed in
Section~11 in the supplementary material, the need to use an estimated
intensity function $\rho_0$ complicates interpretation of results for
these summary statistics. This is a topic that needs further investigation.
\\[\bsl]
{\bf Acknowledgements}\\[\bsl]
We thank Arnulf Mayer, Dept.\ of Radiation Oncology, University
Medical Center, Mainz, Germany, for providing the fluorescence image
and the point pattern data. Kristian B.\ Hessellund and Rasmus
Waagepetersen were supported by The Danish Council for Independent
Research | Natural Sciences, grant DFF - 7014-00074 `Statistics for
point processes in space and beyond', and by the Centre for Stochastic
Geometry and Advanced Bioimaging, funded by grant 8721 from the Villum
Foundation. Ganggang Xu was supported by NSF grant SES-1902195 and
Yongtao Guan by NSF SES-1758575.

\bigskip
\begin{center}
{\large\bf SUPPLEMENTARY MATERIAL}
\end{center}

The supplementary material for this paper contains further plots, an algorithm for updating regularized $\balpha$ and auxiliary results.

\bibliography{bib/mybib,bib/cancer_ref}
\bibliographystyle{jrssc_format/rss}

\end{document}


\maketitle

\begin{abstract}
This document contains supplementary material for the paper `Second order semi-parametric inference for multivariate log Gaussian Cox processes'. Cross-references starting with `M-' refer to that paper.
\end{abstract}

\hyphenation{ex-am-ple hy-phen-a-tion short}
\hyphenation{long la-tex}
\hyphenation{gen-e-ra-li-ty}
\hyphenation{li-ke-li-ho-od-funk-ti-on}
\hyphenation{log-li-ke-li-ho-od-funk-ti-on}
\hyphenation{es-ti-ma-ti-ons-funk-ti-on-er} 

\section{Conditional probability and likelihood}\label{app:condprob}

Define $X^{\text{pooled}}=\cup_{i=1}^p X_i$ with intensity function
$\sum_{i} \rho_i$ and second order joint intensity
$\sum_{l,k} \rho_{lk}$.
Define further the measure
\[ C(A \times B \times \{i\} \times \{j\})= \EE \sum_{\bu,\bv \in
    X^{\text{pooled}}}^{\neq}1[ \bu \in A, \bv \in B, \bu \in X_i, \bv \in X_j] \]
By Radon-Nikodym's theorem,
\[ C(A \times B \times \{i\} \times \{j\}) = \int_{A \times B}
  p_{ij}(\bu,\bv) \left [\sum_{l,k} \rho_{lk}(\bu,\bv) \right ] \dd \bu \dd \bv,\]
where for almost all $(\bu,\bv)$, $p_{ij}(\bu,\bv)$ is a probability
function on $\{1,\ldots,p\} \times \{1,\ldots,p\}$. This follows
because $C(\cdot \times \cdot \times \{i\} \times \{j\})$ is
absolutely continuous with respect to the second order factorial
measure 
\[ \alpha(A \times B)= \textup{E} \sum_{\bu,\bv \in
    X^{\text{pooled}}}^{\neq}1[ \bu \in A, \bv \in B]\]
of $X^{\text{pool}}$ which has density
$\sum_{l,k} \rho_{lk}$. The sum of intensities $\sum_{l,k}
\rho_{lk}(\bu,\bv) \dd \bu \dd \bv$ is the `probability'
that $X^{\text{pool}}$ `has points at $\bu$ and $\bv$'. It is therefore natural to interpret
$p_{ij}(\bu,\bv)$ as the conditional (Palm) probability that $\bu \in X_i, \bv
\in X_j$ given that $\bu,\bv \in X^{\text{pool}}$.
On the other hand,
\[ C(A \times B \times \{i\} \times \{j\}) = \int_{A \times B}
  \rho_{ij}(\bu,\bv) \dd \bu \dd \bv .\]
Thus we obtain M-\eqref{main-eq:proba}.

Another way to arrive at M-\eqref{main-eq:proba} is to define point processes $X_{ij}=\{(\bu,\bv)| \bu \in X_i, \bv \in X_j, \bu \neq \bv\}$ with
intensity functions $\rho_{ij}(\bu,\bv)$. We can further define the union
$\tilde X^{\text{pool}} = \cup_{i, j} X_{ij}$ with intensity
function $\sum_{i,j} \rho_{ij}$. If we now condition on
$\tilde X^{\text{pool}}$ and consider a point $(\bu,\bv) \in \tilde
X^{\text{pool}}$, then M-\eqref{main-eq:proba} is the conditional probability that this point
comes from $X_{ij}$. We could also define $\tilde X^{\text{pool}}$ as
$\cup_{i \le j} X_{ij}$ in which case we would get the conditional probabilities 
\begin{equation}\label{eq:qprob} \text{q}_{ij}(\bu,\bv)= \frac{\rho_{ij}(\bu,\bv)}{\sum_{l \le k} \rho_{lk}(\bu,\bv) }.\end{equation}

\begin{remark}
  An important property of M-\eqref{main-eq:lik} is that the score function
  is unbiased, see Lemma~\ref{lem:unb}. For this to hold it is crucial that when
  we sum over all $l,k$ in the denominator of M-\eqref{main-eq:proba} we also
  use product over all $i,j$ in M-\eqref{main-eq:lik}. Alternatively, using
  \eqref{eq:qprob}, we could define
\begin{equation}\label{eq:loglik2} L(\btheta)= \prod_{i \le j} \prod_{\substack{\bu \in X_i \cap W \\ \bv \in X_j \cap W \\ ||\bu-\bv|| \leq R}}^{\neq} \text{q}_{ij}(\bu,\bv; \btheta).
\end{equation}
In this case, a pair $\{\bu,\bv\}$ with $\bu \in X_i$ and $\bv
\in X_j$ only appears once for $i \neq j$ since the sum is now only
over $i\le j$. However, a pair  $\bu \neq \bv \in
X_i$ will contribute twice to the likelihood. We tried out the two
alternatives M-\eqref{main-eq:lik} and \eqref{eq:loglik2} on a number of
data sets and got very
similar estimates.
\end{remark}

\begin{remark} Note that if an
ordered pair $(\bu,\bv)$ appears in the product in M-\eqref{main-eq:lik} then
so does $(\bv,\bu)$. Hence in a practical implementation we may
restrict the product to $i \le j$ and if $i=j$ only include unordered
pairs $\{\bu,\bv\}$ with $\bu \neq \bv \in X_i$. We can finally square to get M-\eqref{main-eq:lik}.
\end{remark}

\section{Theoretical results concerning conditional composite likelihood score and Hessian}\label{app:theoresults}

In this section $\btheta^*$ denotes the parameter vector for which the data is generated.

\begin{lemma}\label{lem:unb}
The score function $\best(\btheta)=\nabla_{\btau} l_-(\btheta)$  is unbiased meaning $\EE [\best(\btheta^*)] =0$.
\end{lemma}

\begin{proof}
\begin{align*}
\textup{E}[\best(\btheta^*)] &=  \textup{E} \left[ \sum_{i,j} \sum_{\substack{\bu \in X_i \cap W\\ \bv \in X_j \cap W}}^{\neq} \left( \frac{\sum_{k,l} \nabla_{\btau} \rho_{kl}(\bu,\bv;\btheta_{kl}^*)}{\sum_{k,l}  \rho_{kl}(\bu,\bv;\btheta_{kl}^*)} - \frac{\nabla_{\btau} \rho_{ij}(\bu,\bv;\btheta_{ij}^*)}{\rho_{ij}(\bu,\bv;\btheta_{ij}^*)} \right) \right] \\
& = \sum_{i,j} \int_{W^2} \left( \frac{\sum_{k,l} \nabla_{\btau} \rho_{kl}(\bu,\bv;\btheta_{kl}^*)}{\sum_{k,l}  \rho_{kl}(\bu,\bv;\btheta_{kl}^*)} - \frac{\nabla_{\btau} \rho_{ij}(\bu,\bv;\btheta_{ij}^*)}{\rho_{ij}(\bu,\bv;\btheta_{ij}^*)}  \right)\rho_{ij}(\bu,\bv;\btheta_{ij}^*) \textup{d}\bu \textup{d}\bv \\
&=   \int_{W^2}  \sum_{k,l} \nabla_{\btau} \rho_{kl}(\bu,\bv;\btheta_{kl}^*) - \sum_{i,j} \nabla_{\btau} \rho_{ij}(\bu,\bv;\btheta_{ij}^*) \textup{d}\bu \textup{d}\bv = \boldsymbol{0}.
\end{align*}
\end{proof}

\begin{lemma}\label{prop:sens}
The expected Hessian matrix of $l_-(\btheta^*)$ with respect to $\btau$ is given by:
\begin{align*}
\textup{H}(\btheta^*)=\int_{W^2}1_R(\bu,\bv)\textup{Cov}(Z(\bu,\bv,\btheta^*))  \sum_{i,j} \rho_{ij}(\bu,\bv;\btheta^*) \textup{d} \bu \textup{d} \bv,
\end{align*}
where for each $\bu$, $\bv$, $Z(\bu,\bv,\btheta^*)$ denotes a random vector which takes values $\nabla_{\btau} \textup{log}\rho_{ij}(\bu,\bv; \btheta^*)$ with probabilities $\textup{p}_{ij}(\bu,\bv; \btheta^*)$, $i,j=1,\ldots,p$.
\end{lemma}

\begin{proof}
We suppress the arguments $\bu$, $\bv$ and $\btheta^*$ in order to save space. The Hessian matrix for $l_-(\btheta)$ is:
\begin{align*}
&\frac{\partial}{\partial \btau^\T} e(\btheta)= \nabla^2_{\btau} l_-(\btheta)= \\
& \sum_{i,j} \sum_{\substack{\bu \in X_i \cap W \\ \bv \in X_j \cap W}}^{\neq} \frac{( \sum_{k,l} \nabla^2_{\btau}  \rho_{kl} ) (\sum_{k,l} \rho_{kl}) - (\sum_{k,l} \nabla_{\btau}  \rho_{kl})( \sum_{k,l} \nabla^\T_{\btau}  \rho_{kl})  }{(\sum_{k,l} \rho_{kl})^2} \quad -\\
& \sum_{i,j} \sum_{\substack{\bu \in X_i \cap W \\ \bv \in X_j \cap W}}^{\neq}  \frac{ (\nabla^2_{\btau} \rho_{ij}) \rho_{ij} - (\nabla_{\btau} \rho_{ij}) (\nabla^\T_{\btau} \rho_{ij})}{\rho_{ij}^2} 
\end{align*}
The expected Hessian is then given by:
\begin{align*}
& \textup{H}(\btheta^*) \\
= & \sum_{i,j} \int_{W^2} 1_R \frac{( \sum_{k , l} \nabla^2_{\btau}  \rho_{kl} ) (\sum_{k  , l} \rho_{kl}) - (\sum_{k , l} \nabla_{\btau}  \rho_{kl})( \sum_{k , l} \nabla^\T_{\btau}  \rho_{kl})  }{(\sum_{k , l} \rho_{kl})^2} \rho_{ij} du dv \quad -\\  & \sum_{i , j} \int_{W^2} 1_R \frac{ (\nabla^2_{\btau} \rho_{ij}) \rho_{ij} - (\nabla_{\btau} \rho_{ij}) (\nabla^\T_{\btau} \rho_{ij})}{\rho_{ij}^2} \rho_{ij} du dv \\
= & \int_{W^2} 1_R \sum_{i , j} \frac{ (\nabla_{\btau} \rho_{ij}) (\nabla^\T_{\btau} \rho_{ij})}{\rho_{ij}} - \frac{(\sum_{k,l} \nabla_{\btau}  \rho_{kl})( \sum_{k,l} \nabla^\T_{\btau}  \rho_{kl})  }{\left( \sum_{k,l} \rho_{kl} \right)} du dv \\
= & \int_{W^2} 1_R \left( \sum_{k,l} \rho_{kl} \right) \left( \sum_{i , j} \nabla_{\btau} \textup{log}(\rho_{ij}) \nabla^\T_{\btau} \textup{log} (\rho_{ij}) \textup{p}_{ij}\right) du dv\\
& - \int_{W^2} 1_R  \left( \sum_{k,l} \rho_{kl} \right) \left( (\sum_{k,l} \nabla_{\btau}  \textup{log} ( \rho_{kl})\textup{p}_{kl}) ( \sum_{k,l} \nabla_{\btau} \textup{log} (\rho_{kl})\textup{p}_{kl} )^\T \right) du dv \\
= & \int_{W^2} 1_R \left( \sum_{i,j} \rho_{ij} \right) \textup{Cov}(Z(\btheta^*)) du dv.
\end{align*}
\end{proof}

\section{Quadratic approximation and least squares} \label{app:pnm}
Omitting the first term not depending on $\btau$ and letting $K=-\textup{H}^{-1}(\tilde \btau)\best(\tilde \btau)$, the quadratic approximation M-\eqref{main-eq:taylor} can be rewritten as follows:
\begin{align} 
  &   (\btau - \tilde \btau)^\T \best(\tilde \btau)+ \frac{1}{2}(\btau - \tilde \btau)^\T \textup{H} (\tilde \btau)(\btau - \tilde \btau) \nonumber\\
 = &   -(\btau - \tilde \btau)^\T  \textup{H} (\tilde \btau)K + \frac{1}{2}(\btau - \tilde \btau)^\T \textup{H} (\tilde \btau)(\btau - \tilde \btau)  \nonumber\\
 =  &   \frac{1}{2}(K-(\btau - \tilde \btau))^\T \textup{H} (\tilde
      \btau)(K-(\btau - \tilde \btau))  -\frac{1}{2}K^\T \textup{H}(\tilde \btau)K.\nonumber
\end{align}
Hence,  minimizing M-\eqref{main-eq:taylor} is a least squares problem:
\begin{align*} 
\hat{\btau} = \textup{arg min}_{\btau} \bigg( \|\textup{H} (\tilde \btau)^{1/2}(K-(\btau- \tilde \btau))\|^2 \bigg) 
=  \textup{arg min}_{\btau} \bigg( \|Y -X\btau \|^2 \bigg) ,
\end{align*}
where
\begin{align*}
Y=\textup{H}(\tilde \btau)^{1/2}\left( -\textup{H}(\tilde \btau)^{-1}\best(\tilde \btau)+\tilde \btau \right) \quad \textup{and} \quad
X=\textup{H}(\tilde \btau)^{1/2}.
\end{align*}

\section{Update of regularized $\balpha$} \label{app:updatealp}
Following \cite{shi:etal:16}, we  update $\balpha$ by minimizing the augmented Lagrangian object function 
\begin{align*}
Q_{\lambda,\mu}(\balpha,\boeta) = \frac{1}{2} ||Y -X\balpha ||^2  + \lambda \sum_{i=1}^{p}\sum_{j=1}^q| \alpha_{ij} | + \boeta \bC \balpha + \frac{\mu}{2} ||\bC \balpha||^2,
\end{align*}
where $\boeta \in \R^q$ is the Lagrange multiplier, and $\mu >0$ is a
penalty parameter that we set to 1 as in \cite{shi:etal:16}. Letting $\balpha^{\text{current}}$, $\boeta^{\text{current}}$, $\balpha^{\text{new}}$, and $\boeta^{\text{new}}$ denote temporary vectors used in the iterative algorithm, we initialize $\balpha^{\text{current}}=\balpha^{(n)}$ and $\boeta^{\text{current}}=\boldsymbol{0}$. We then iterate updates
\begin{align*}
\balpha^{\text{new}} \leftarrow \textup{arg min}_{\balpha} Q_{\lambda,\mu}(\balpha,\boeta^{\text{current}},\bpi^{(n)}) \quad \boeta^{\text{new}} \leftarrow \boeta^{\text{current}} + \mu \bC \balpha^{\text{new}}  .
\end{align*}
The updating is terminated if for some $\tilde{\epsilon}$ and $\tilde{\tilde{\epsilon}}$,
\begin{align} \label{eq:innercrit}
|| \balpha^{\text{new}}-\balpha^{\text{current}}|| < \tilde{\epsilon} \quad \textup{and} \quad || \boeta^{\text{new}}-\boeta^{\text{current}}|| < \tilde{\tilde{\epsilon}},
\end{align}
in which case $\hat \balpha:=\balpha^{\text{new}}$. Otherwise $\balpha^{\text{current}} \leftarrow \balpha^{\text{new}}$ and $\boeta^{\text{current}} \leftarrow \boeta^{\text{new}}$ and a new iteration takes place.

The update leading to $\balpha^{\text{new}}$ is conducted using cyclical updating of the entries $\alpha_{ij}^{\text{new}}$ in $\balpha^{\text{new}}$. The update of the $ij$th entry is 
\begin{align} \label{eq:alphaupdate}
c_1 & \leftarrow  X_{\cdot ij}^\T \left(Y - \sum_{lk \neq ij} X_{\cdot lk} \alpha_{lk}^{\text{new}} \right) \nonumber \\
  c_2 & \leftarrow \mu   \left(\sum_{lk \neq ij} \alpha_{lk}^{\text{new}} \bC_{\cdot lk}^\T \bC_{\cdot ij} + \bC_{\cdot ij}^\T \boeta/\mu\right) \nonumber \\
\alpha_{ij}^{\text{new}} & \leftarrow \frac{S \left( c_1 - c_2, \lambda \right)}{ X_{\cdot ij}^\T X_{\cdot ij} + \mu \bC_{\cdot ij}^\T \bC_{\cdot ij}},
\end{align}
where  $X_{\cdot ij}$ and $C_{\cdot ij}$ are the columns of $X$ and
$C$ corresponding to $\alpha_{ij}$ (i.e.\ when $\balpha$ is laid out
columnwise). The  soft-thresholding operator $S(\cdot,\cdot)$ is  given by:
\begin{align*}
S(c,\lambda)= \begin{cases} c - \lambda &\textup{ if } c>0 \textup{ and } \lambda < |c| \\
 c+\lambda &\textup{ if } c<0 \textup{ and } \lambda < |c| \\
 0 &\textup{ if } \lambda > |c|.
\end{cases}
\end{align*}

Algorithm~\ref{fig:alg2} gives an overview of the regularized cyclical block descent algorithm.
\begin{algorithm}
\caption{Regularized cyclical  block descent algorithm}\label{fig:alg2}
\begin{algorithmic}[1]
  \BState Simulate initial parameters $\hat{\balpha}^{(0)}$, $\hat{\bxi}^{(0)}$, $\hat{\bsigma}^{2(0)}$ and $\hat{\bphi}^{(0)}$
  \BState $n:=0$
  \BState $\boldsymbol{repeat}$
\BState \quad $\tilde \balpha :=\balpha^{(n)}$,  $\tilde \bxi := \bxi^{(n)}$, $\bsigma^2 := \bsigma^{2(n)}$ and  $\bphi := \bphi^{(n)}$
\BState \quad update $\tilde \balpha$ using augmented Lagrangian method followed by line search
\BState \quad update $\tilde \bxi$, $\tilde \bsigma^2$ and $\tilde \bphi$ in turn using M-\eqref{main-eq:minfct2} followed by line search
\BState \quad $\balpha^{(n+1)}:=\tilde \balpha$, $\bxi^{(n+1)}:=\tilde \bxi$, $\bsigma^{2(n+1)} := \tilde \bsigma^2$, and  $\bphi^{(n+1)} := \tilde \bphi$
\BState \quad $n:=n+1$
\BState $\boldsymbol{until}$ relative convergence for object function M-\eqref{main-eq:penalize}
\BState $\boldsymbol{return}$ $\hat \btheta=\btheta^{(n)}$
\end{algorithmic}
\end{algorithm}

\section{Bandwidth selection for $\rho_0$}\label{app:bandwidth}
Following \cite{hessellund:etal:20}, we can estimate
$\rho_0$ using the semi-parametric kernel estimator:
\begin{align} \label{eq:rho0est}
\hat{\rho}_0(\bu)= \frac{1}{p} \sum_{i=1}^p \sum_{\bv \in X_i \cap W} \textup{exp}(-\hat{\boldsymbol{\beta}}_i^\T \bz(\bv))\frac{k((\bu-\bv)/b)}{b^dc_b(\bv)},
\end{align}
where $k$ is a $d$-dimensional kernel, $b$ is the bandwidth, and
$c_b(v)=\int_W k(\bu-\bv) \dd \bu$ is an edge correction factor. 

We
suggest to choose the bandwidth according to a criterion inspired by
\cite{cronie:lieshout:16} who consider the squared difference between the observation window area
and an estimate (depending on the bandwidth) of this area. However,
exact knowledge of the observation window area may not always be
available (this is e.g.\ the case for the crime data, where the
observation window depends on the complex urban structure of
Washington DC). To handle this we take advantage of our 
 multivariate setup.  Define $X^{\text{pooled}} = \cup_{i=1}^p X_i$ with intensity function $\rho^{\text{pooled}}(\bu)= \sum_{i=1}^p \rho_i(\bu)$. An estimator of $\rho^{\text{pooled}}$ is simply given by $\hat{\rho}^{\text{pooled}}(\bu)=\hat{\rho}_0(\bu)\sum_{i=1}^p \textup{exp}(\hat{\boldsymbol{\beta}}_i^\T \bz(\bu))$. We can then define two different estimators, $\hat{\omega}$ and $\hat{w}$, for the area of the observation window:
\begin{align*}
\hat{\omega}(b) = \frac{1}{p} \sum_{i=1}^p \sum_{\bu \in X_i \cap W}
  \frac{1}{\hat{\rho}_0(\bu)
  \textup{exp}(\hat{\boldsymbol{\beta}}_i^\T \bz(\bu))} \quad
  \textup{and} \quad\hat{w}(b) = \sum_{\bu \in X^{\text{pooled}} \cap W}\frac{1}{\hat{\rho}^{\text{pooled}}(\bu)},
\end{align*} 
where the dependence on $b$ is through $\hat \rho_0(\cdot)$.
We then select the bandwidth $b$ that minimizes $(\hat{\omega}(b) - \hat{w}(b))^2$. Hence, the bandwidth can be selected  without specifying the observation window.

\section{Modification of \cite{diggle:etal:07}s second order analysis} \label{app:diggle}

As pointed out in \cite{diggle:etal:07}, non-parametric estimation of
both first and second order properties from the same point pattern
data is an ill-posed problem due to confounding between variations in the intensity function and random clustering. In case of a bivariate case-control point process, and assuming the model M-\eqref{main-eq:multi1}, \cite{diggle:etal:07} suggested to estimate $\rho_0$ using the control points and plug in this estimate when inferring the clustering properties of the case process. This approach can be extended to the multivariate ($p>2$) setting as follows. For each type $i=1,\dots, p$, we modify \eqref{eq:rho0est} to obtain the estimator
\begin{align*}
\hat{\rho}_{0,-i}(\bu)= \frac{1}{p-1} \sum_{\substack{k=1\\ k \neq i}}^p \sum_{\bv \in X_k} \frac{k((\bu-\bv)/b)}{\textup{exp}(-\hat{\boldsymbol{\beta}}_k^\T \bz(\bv))b^d c_b(\bv)},
\end{align*}
that does not utilize the $i$th point pattern. To ease the computation
time we estimate one common bandwidth $b$ for all $i=1,\dots,p$, by the
bandwidth selection criterion described in
Section~\ref{app:bandwidth}. To estimate the PCFs and cross PCFs we
use the \texttt{spatstat} procedures \texttt{pcfinhom} and
\texttt{crosspcfinhom}, where we specify the intensity functions by
\begin{equation}\label{eq:hatrhoi}
\hat{\rho}_i(\bu)=\hat{\rho}_{0,-i}(\bu) \textup{exp}(\hat{\bbeta}_i^\T \bz(\bu)).
\end{equation}
We manually choose the
bandwidth for the PCFs and cross PCFs.

\section{Simulation study for product shot-noise Cox process}\label{sec:productshotnoise}

Multivariate product shot-noise Cox processes were introduced in
\cite{jalilian:etal:15}. A $p$-variate product shot-noise Cox process
$\boldsymbol{X}=(X_1,\ldots,X_p)$ is driven by $p$ parent Poisson
processes $\Phi_1,\ldots,\Phi_p$. The random intensity function of the
$i$th component $X_i$ is a product of a usual shot-noise field $S_i$
depending on $\Phi_i$ as well as certain product fields $F_{il}$, $l
\neq i$, that via $\Phi_l$ model positive or negative impact of the
$l$th component process $X_l$ on $X_i$, see \cite{jalilian:etal:15}
for details. \cite{jalilian:etal:15} derived closed-form expressions
for the PCFs and cross PCFs of a product shot-noise Cox process. Using
simulations from a product shot-noise Cox process we can therefore
assess the performance of our estimation methodology when the assumed
multivariate LGCP model is misspecified. 

To consider a practically
relevant case of a product shot-noise Cox process, we use a parameter
setting for the interaction and clustering parameters corresponding to
the model fitted in \cite{jalilian:etal:15} to point patterns of locations
of five rain forest tree species. We next rescale the model to the
unit square. The PCFs and cross PCFs are shown in
Figure~\ref{fig:prodpcfs}.
\begin{figure}[ht!]
\centering
\begin{tabular}{c c}
\includegraphics[width=0.45\textwidth]{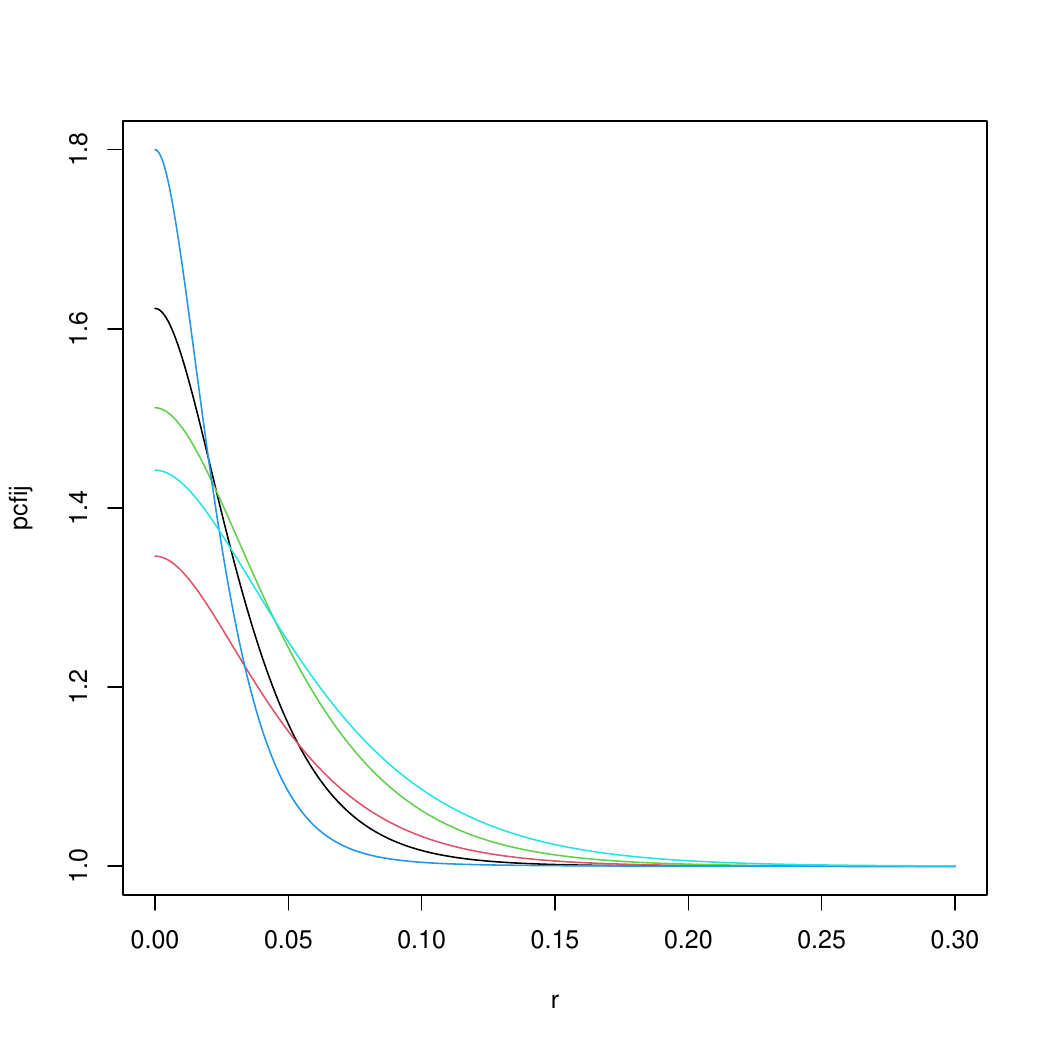}
\includegraphics[width=0.45\textwidth]{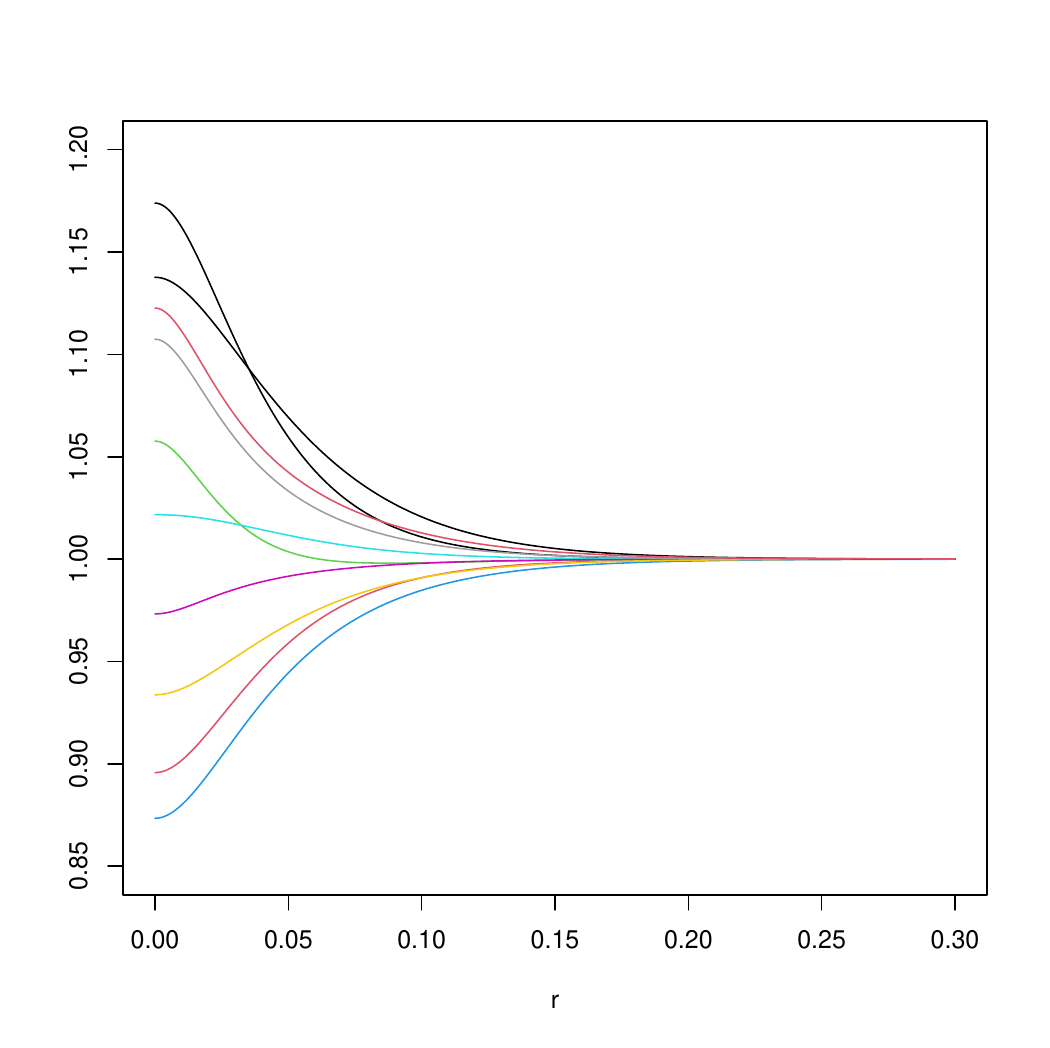}
\end{tabular}
\caption{PCFs (right) and cross PCFs (left) for five-variate product shot-noise Cox process used in simulation study.} \label{fig:prodpcfs}
\end{figure}
We simulate $\rho_0$ as a log Gaussian random field and
do not include any covariates. Each component process contains on
average 425 points. 

We apply our estimation methodology as
detailed in Section~M-\ref{main-sec:simu} using cross validation to select
$q$  and $\lambda$ (with $q$ chosen according to the
MIN rule). As in
Sections~M-\ref{main-sec:simstudy1}-M-\ref{main-sec:simstudy2} we compare
the performance of our methodology with non-parametric
alternatives (the `simple' method and `Diggle's' method) as discussed
in Section~M-\ref{main-sec:simu}.

The cross validation procedure selected $q$ to be 1, 2 or 3 in
respectively 34\%, 53\% or 13\% of the cases. The value $\lambda=0$
was chosen in 64\% of the cases implying no lasso regularization in
these cases. As discussed in the first paragraph of Section~M-\ref{main-sec:simstudy1} values of $q$ up to 3 yield fairly parsimonius parametrizations
of the in total 15 PCFs and cross PCFs.

Figure~\ref{fig:simpcf} shows means and 95\% probability intervals for
the semi- and non-parametric estimates for a selection of PCFs and
cross PCFs.
\begin{figure}[ht!]
\centering
\begin{tabular}{c c c}
\includegraphics[width=0.33\textwidth]{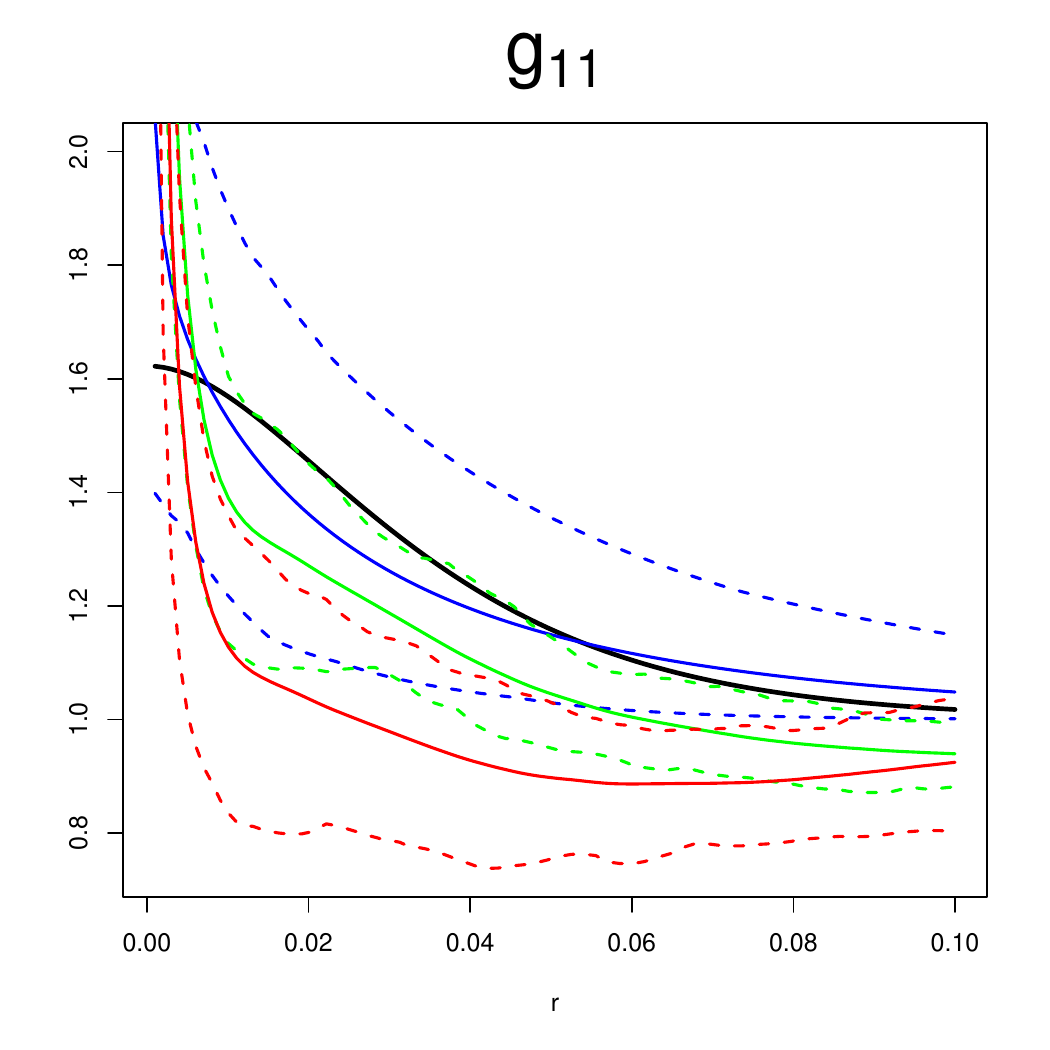}
\includegraphics[width=0.33\textwidth]{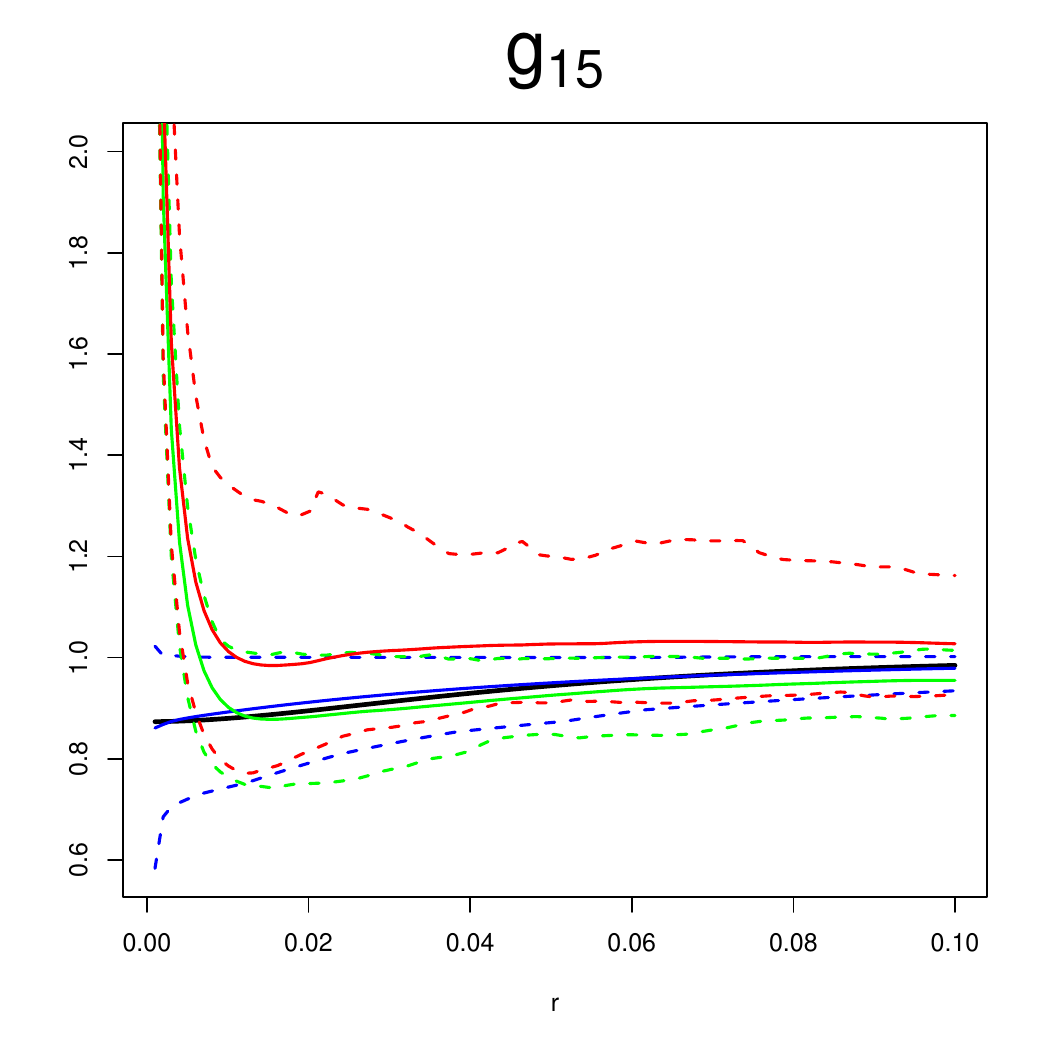}
\includegraphics[width=0.33\textwidth]{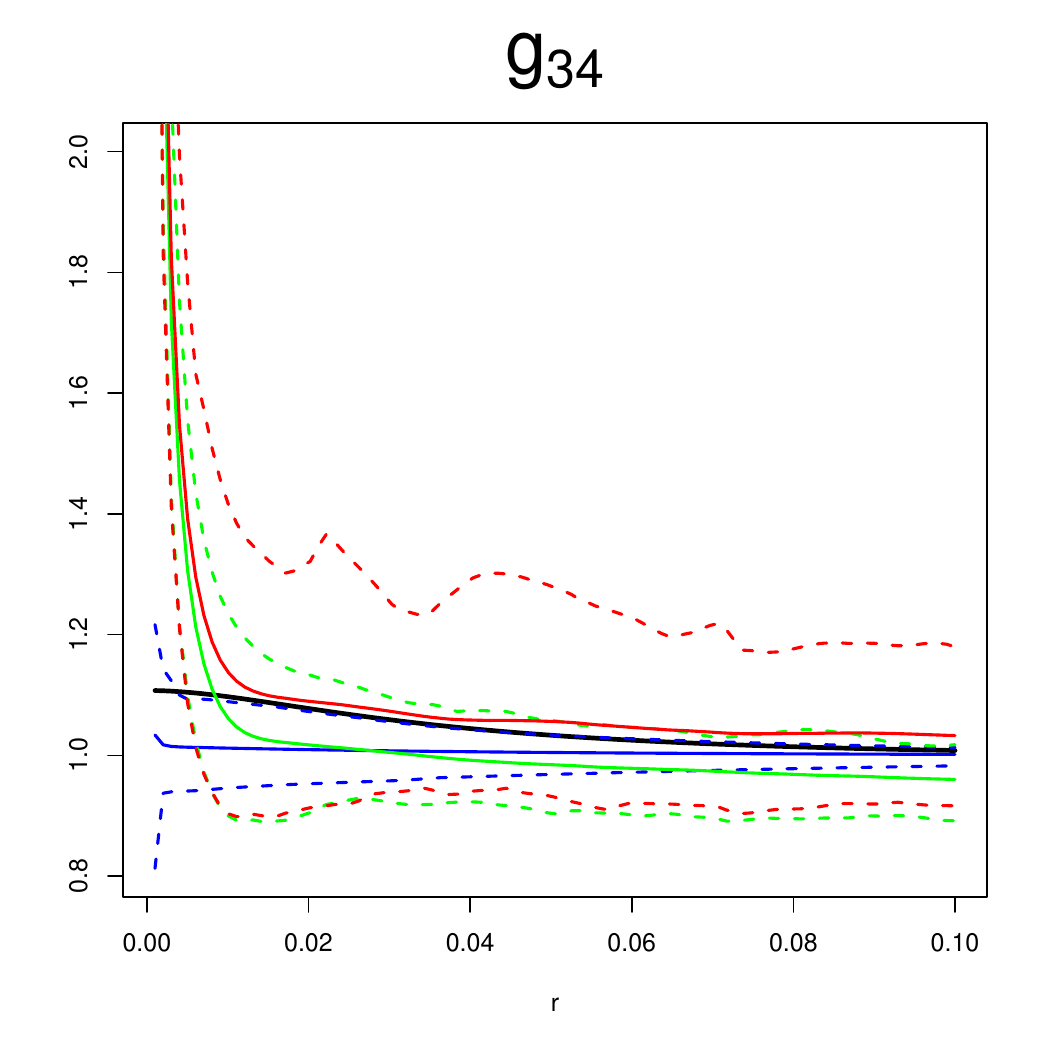}
\end{tabular}
\caption{(true $q=0$) Blue, red and green solid lines indicate
  pointwise means of estimates for selected cross PCFs using  MIN,
  simple, or Diggle's method. The dotted lines indicate the corresponding $95 \%$ pointwise probability intervals. Black solid lines indicate true cross PCFs.} \label{fig:simpcf}
\end{figure}
Not surprisingly, given the model misspecification, some bias is
present for the semi-parametric estimates. The semi-parametric
estimates nevertheless seem to capture the behaviour of the true PCFs
and cross PCFs reasonably well. Moreover, even stronger bias seem to
occur for the non-parametric methods, and especially so for the simple
method. Also the sampling variability of the simple approach seems
markedly larger than for the semi-parametric estimates.

Table~\ref{tbl:simmise} shows MISE aggregated over all PCFs and cross
PCFs ($\text{MISE}_{\text{total}}$), MISE for PCFs
($\text{MISE}_{\text{within}}$), and MISE for cross PCFs
($\text{MISE}_{\text{between}}$). In accordance with the observations
from Figure~\ref{fig:simpcf}, the MISEs for the simple non-parametric method
are 6-8 times larger than for the semi-parametric method. Diggle's
method works better than the simple method but still with MISEs up to
2 times bigger than for the semi-parametric method.
\begin{table}
\caption{MISE using $(q_{\min},0)$, simple, or
  Diggle's non-parametric approach when point patterns are simulated
  from a product shot-noise Cox process.}\label{tbl:simmise}
\centering
\begin{tabular}{| l | c c c |}
\hline
  & $(q_{\min},\lambda)$ & simple & Diggle\\ \hline
  $\text{MISE}_{\text{total}}$ & 5.7$\cdot 10^{-4}$ & 3.7$\cdot 10^{-3}$ &9.1$\cdot 10^{-4}$ \\
$\text{MISE}_{\text{within}}$ & 1.3$\cdot 10^{-3}$ & 8.0$\cdot 10^{-3}$ & 2.0$\cdot 10^{-3}$ \\
$\text{MISE}_{\text{between}}$ & 2.0$\cdot 10^{-4}$ & 1.5$\cdot 10^{-3}$ & 3.7$\cdot 10^{-4}$ \\
  \hline
\end{tabular}
\end{table}

\section{Analysis for tumor cells}\label{app:tumor}
For the tumor cells, we use 80\% independent thinning of the points but
otherwise proceed precisely as in Section~M-\ref{main-sec:stroma} to which
we refer for details. We choose Normoxic
as the baseline and estimate
$\bbeta=(\beta_{\textup{Hyp}},\beta_{\textup{Nor}})^\T$ by
$\hat{\beta}_{\textup{Hyp}}= \textup{log}(2346/3693)=-0.45$ and
$\hat{\beta}_{\textup{Nor}}=0$.
\begin{figure}[ht!]
\centering
\begin{tabular}{c c}
\includegraphics[scale=0.16]{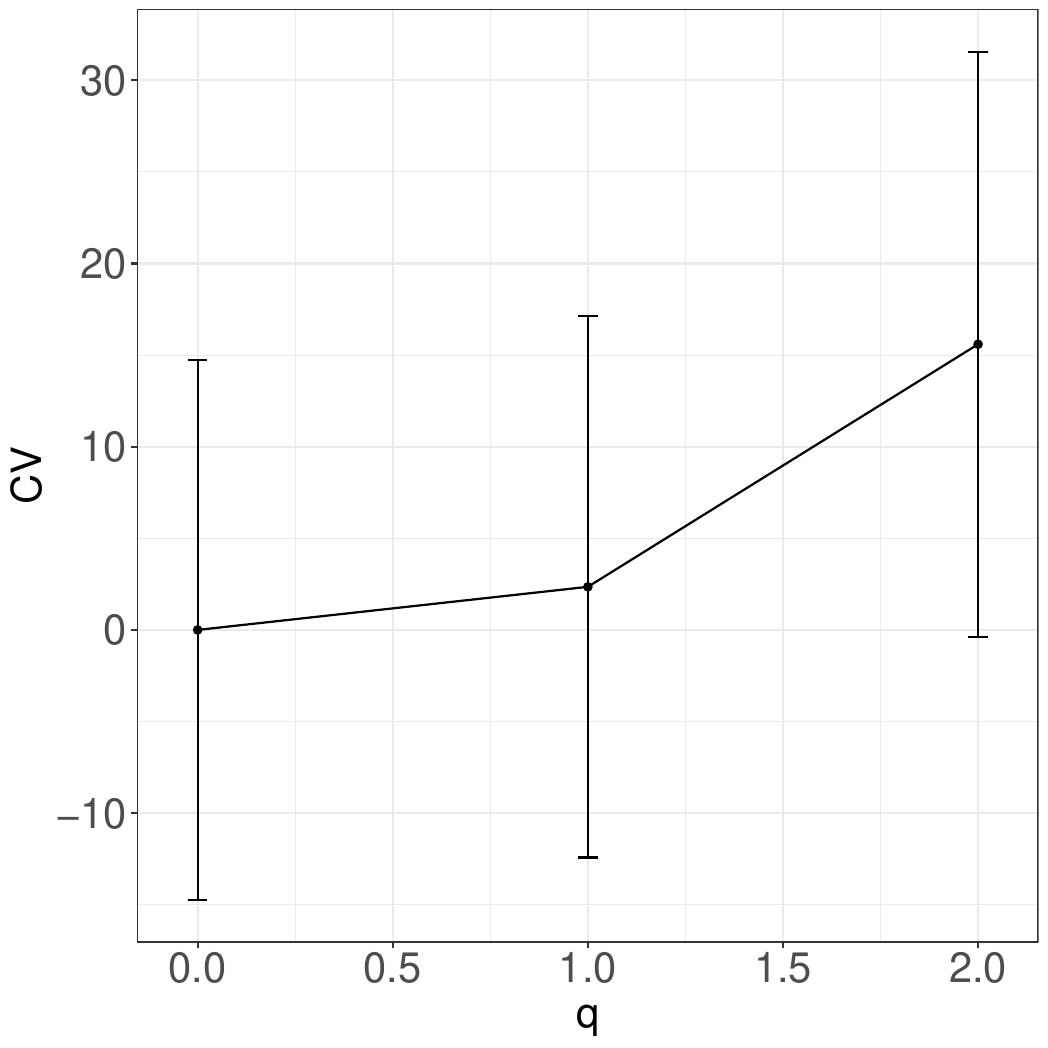}
\hspace{0.5cm}
\includegraphics[scale=0.3]{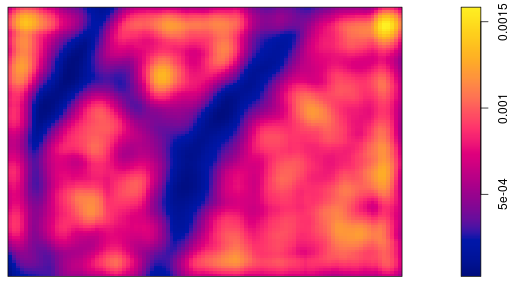}
\end{tabular}
\caption{Left: CV-scores (minus minimum CV-score) with standard errors. Right: Non-parametric estimate of $\rho_0$ with bandwidth = 153.3.} \label{fig:cvtumor}
\end{figure}
 The left panel in Figure~\ref{fig:cvtumor} shows that the cross
 validation score is minimized for $q=0$. Hence, we model the
 bivariate point process as two independent LGCPs. The right panel in
 Figure~\ref{fig:cvtumor} shows the non-parametric estimate of
 $\rho_0$. 

The parameter estimates with $q=0$ are $\hat
 \sigma_{\text{Hyp}}=1.45$, $\hat \sigma_{\text{Nor}}=1.31$, $\hat
 \phi_{\text{Hyp}}=66.1$, and $\hat \phi_{\text{Nor}}=46.4$.
The estimates of $\sigma_{\text{Hyp}}$ and $\sigma_{\text{Nor}}$
show that in addition to the variation caused by $\rho_0$, both the
Hypoxic and Normoxic cells are highly clustered with strongest
clustering for Hypoxic. This is also
illustrated by the fitted PCFs in the left panel of Figure~\ref{fig:tumorplot}.
\begin{figure}[ht!]
\centering
\begin{tabular}{c c}
\includegraphics[scale=0.25]{figures/all_pcfs.pdf}
\includegraphics[scale=0.25]{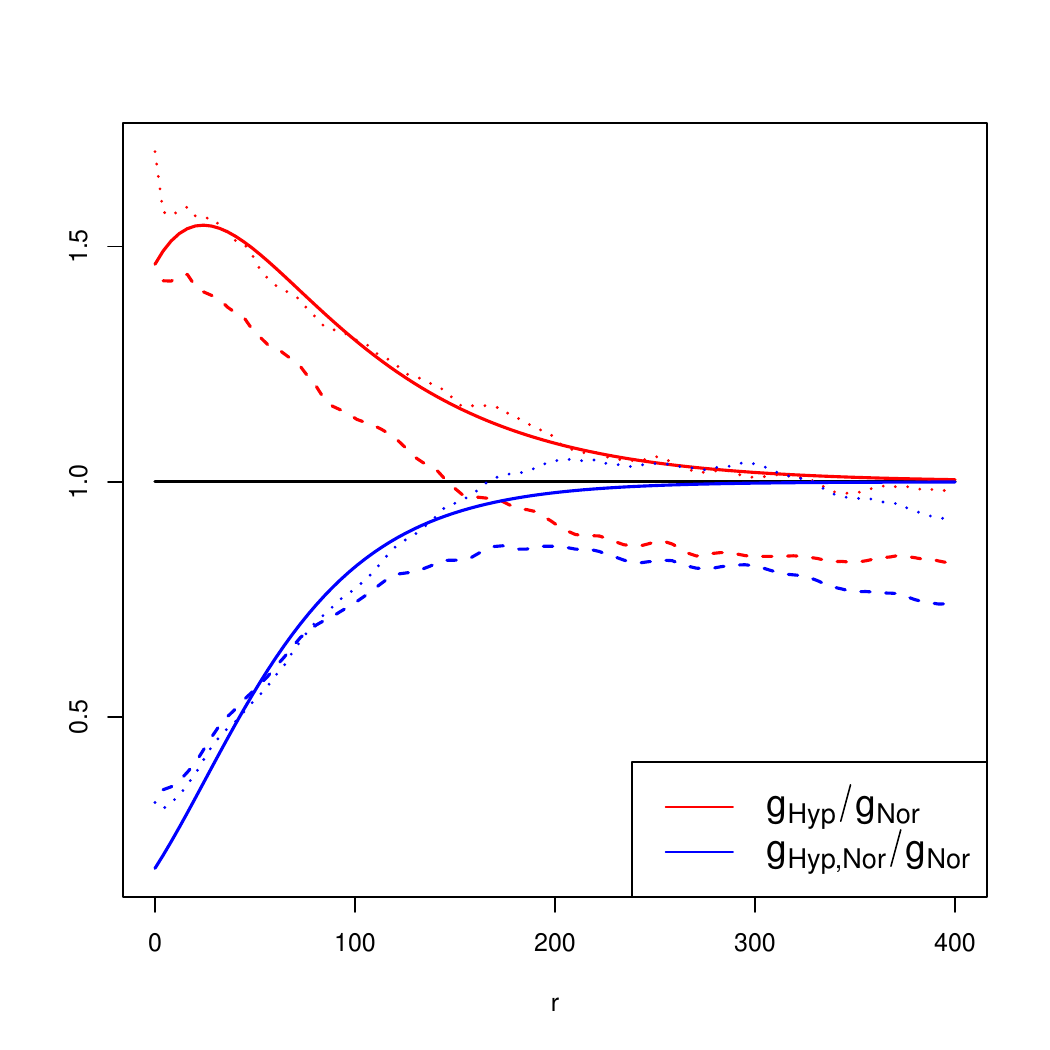}
\end{tabular}
\caption{Left: estimated (cross) PCFs  using the
  semi-parametric model (solid) and simple approach
  (dashed). Right: estimated (cross) PCF ratio using
  semi-parametric model (solid), simple approach (dashed) and consistent approach (dotted).} \label{fig:tumorplot}
\end{figure}


Figure~\ref{fig:tumorplot} shows that the agreement between the
semi-parametric and consistent non-parametric estimates of cross
PCF ratios $g_{\text{Hyp}}/g_{\text{Nor}}$ and
$g_{\text{Hyp,Nor}}/g_{\text{Nor}}$ is very good. This is confirmed by global envelope $p$-values of 0.98 in case of
$g_{\text{Hyp}}/g_{\text{Nor}}$ and 0.178 for
$g_{\text{Hyp,Nor}}/g_{\text{Nor}}$, see also the global
envelope plots in Section~\ref{app:get}. Further comments on model
assessment using the inhomogeneous cross $J$-function
\citep{cronie:lieshout:16} are given in
Section~\ref{sec:jinhom}. The conclusions regarding the simple non-parametric estimates shown
in Figure~\ref{fig:tumorplot} (left) are completely analogous to those for
Stroma and CD8: the non-parametric estimates seem biased and the simple
non-parametric estimates of cross PCF ratios deviate more from the
consistent non-parametric estimates than the semi-parametric
estimates (Figure~\ref{fig:tumorplot}, right).

\section{Model assessment for lymphoma data using cross PCF ratios}\label{app:get}

For the lymphoma data, global envelopes for
ratios of cross PCFs are shown in Figures~\ref{fig:diffratiocd8}-\ref{fig:diffratiotumor}.

\begin{figure}[ht!]
\centering
\begin{tabular}{c c}
\includegraphics[scale=0.35]{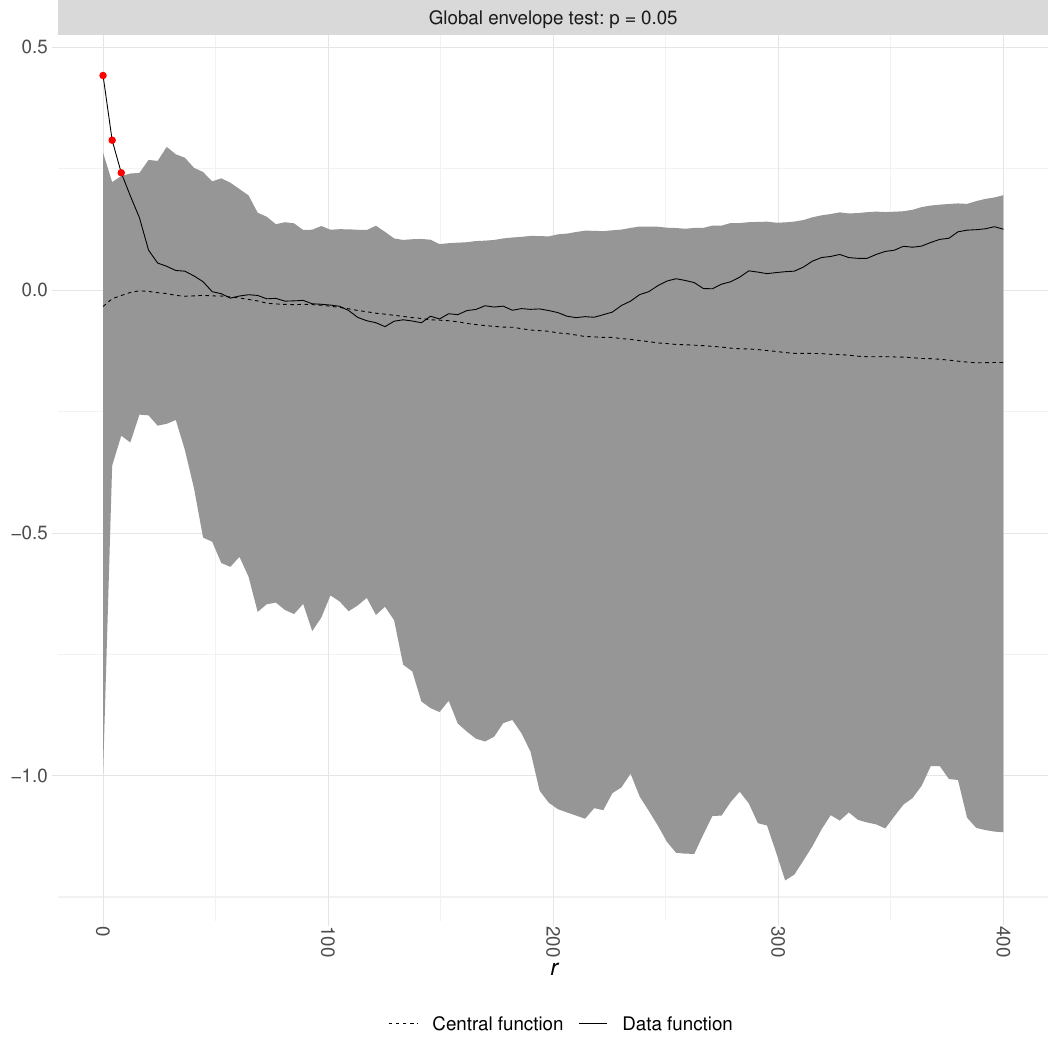}
\hspace{0.1cm}
\includegraphics[scale=0.35]{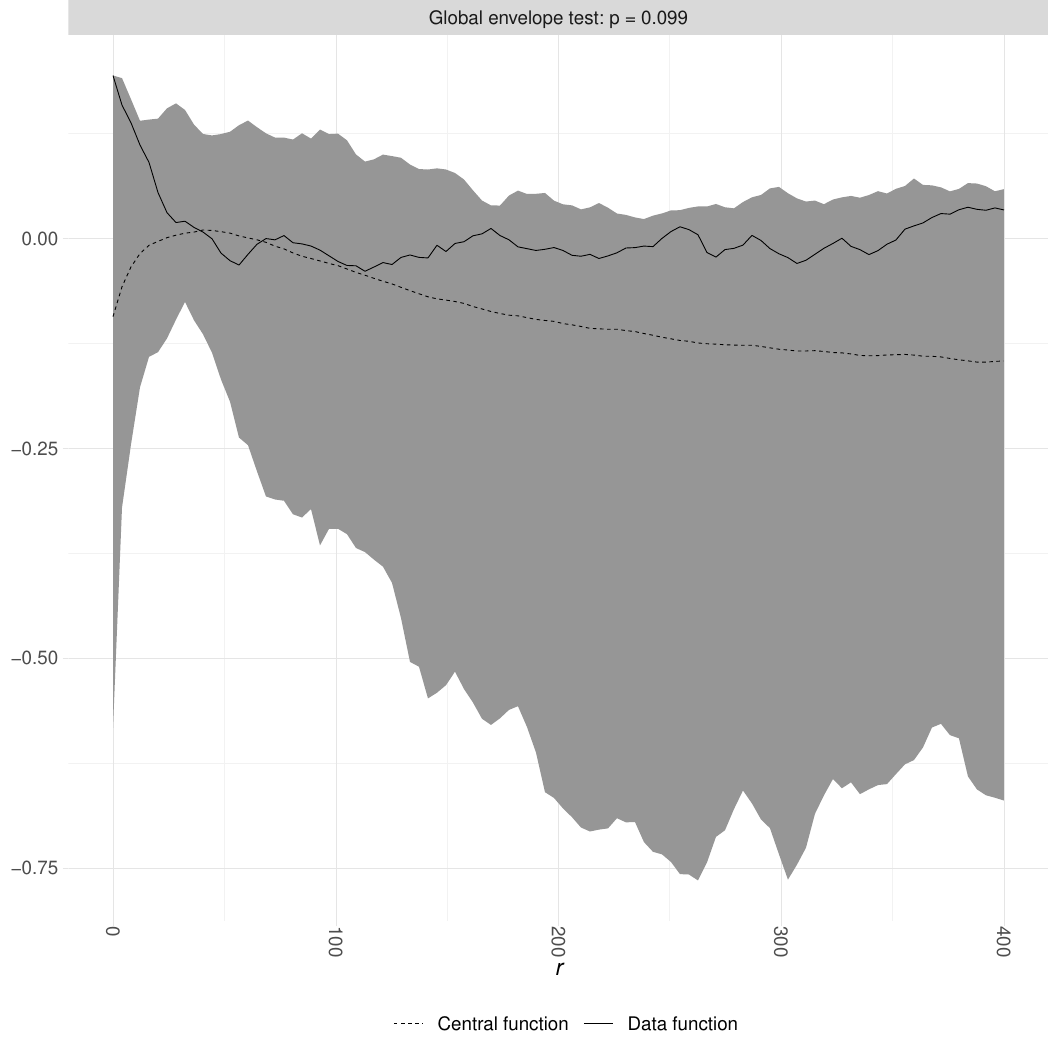}
\end{tabular}
\caption{Differences between semi-parametric and consistent non-parametric estimates of cross PCF ratios (solid curves) with global 95\% envelopes (gray shaded areas). Left: $g_{\textup{Str}}/g_{\textup{CD8}}$. Right: $g_{\textup{Str},\textup{CD8}}/g_{\textup{CD8}}$.} \label{fig:diffratiocd8}
\end{figure}

\begin{figure}[ht!]
\centering
\begin{tabular}{c c}
\includegraphics[scale=0.35]{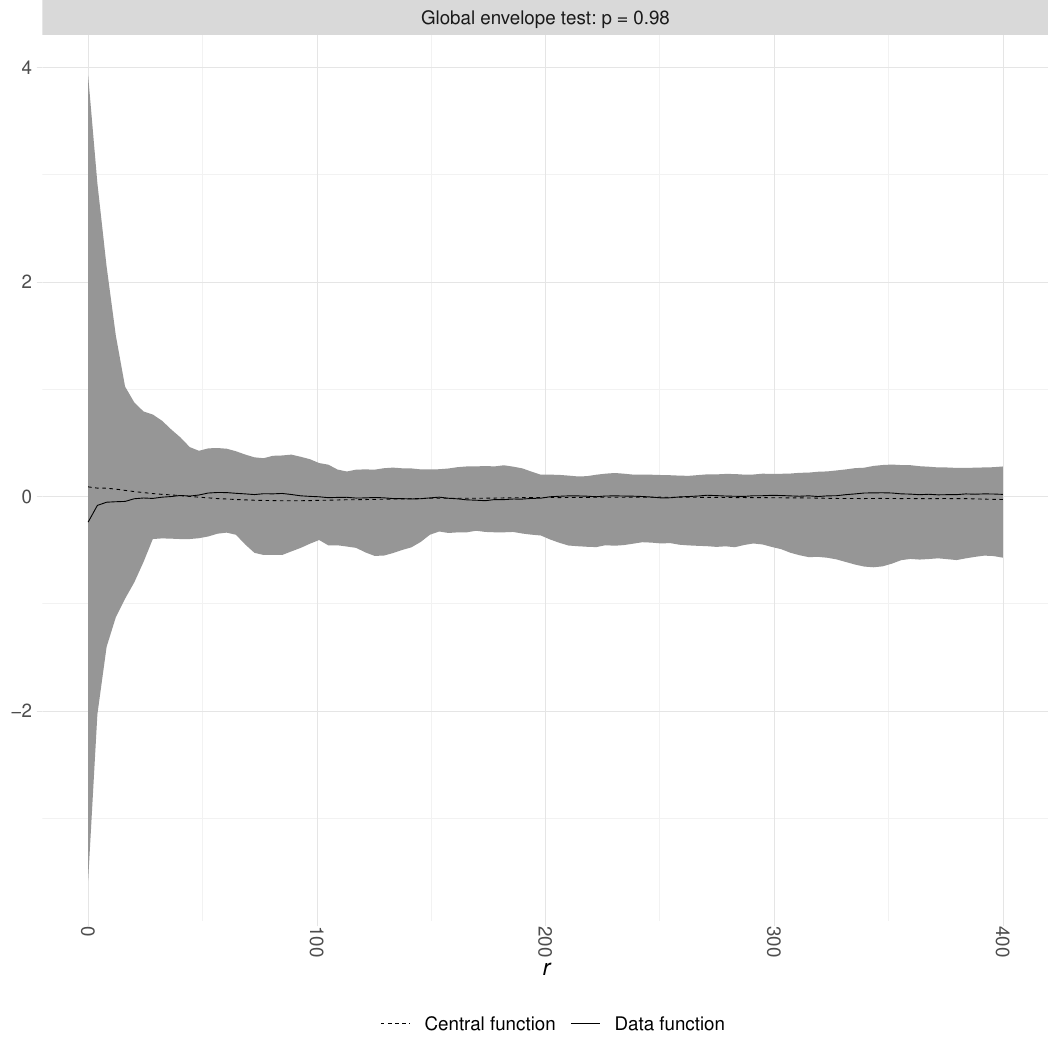}
\hspace{0.1cm}
\includegraphics[scale=0.35]{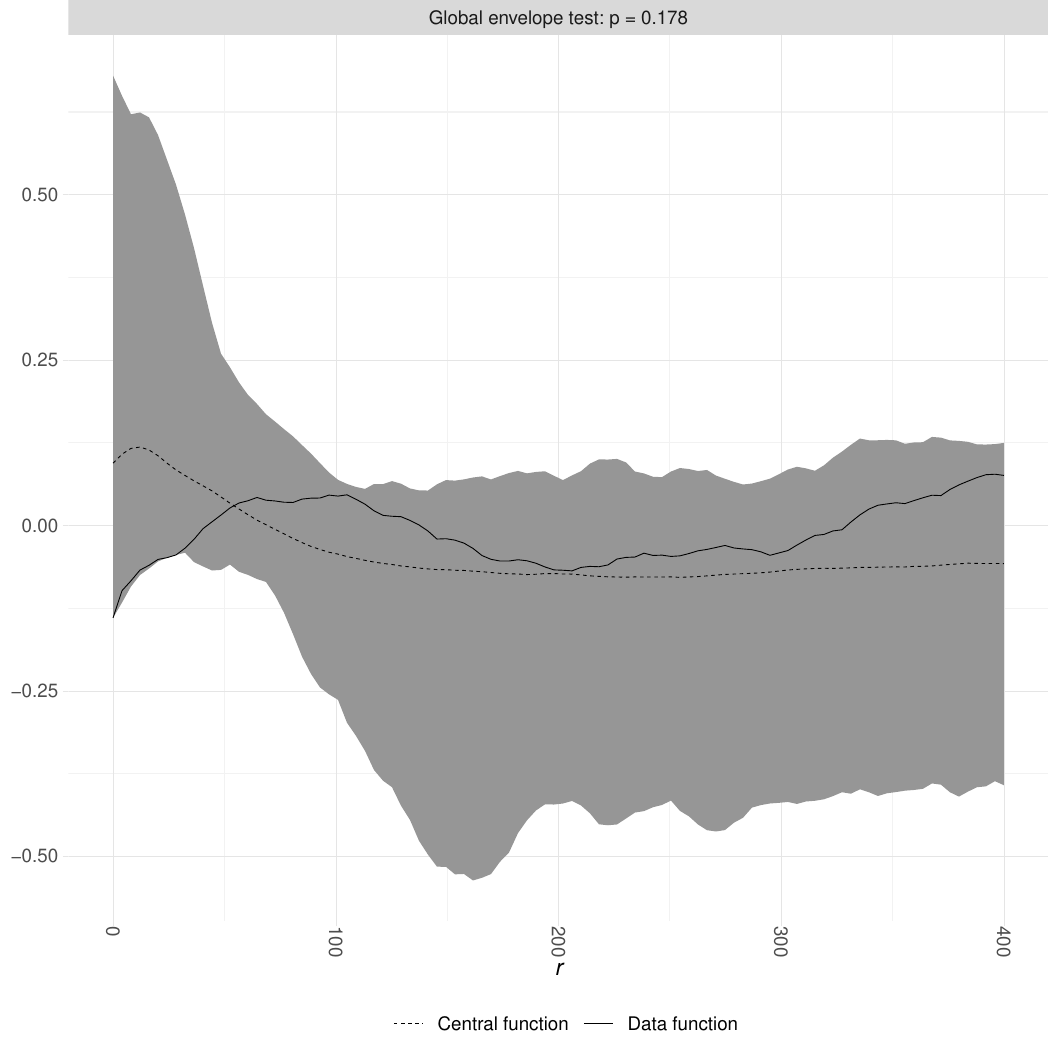}
\end{tabular}
\caption{Differences between semi-parametric and consistent non-parametric estimates of cross PCF ratios (solid curves) with global 95\% envelopes (gray shaded areas). Left: $g_{\textup{Hyp}}/g_{\textup{Nor}}$. Right: $g_{\textup{Nor},\textup{Hyp}}/g_{\textup{Nor}}$.} \label{fig:diffratiotumor}
\end{figure}

\section{Model assessment for crime data using cross PCF ratios}\label{app:getcrime}

For the crime data, representative examples of global envelopes for
ratios of cross PCFs are shown in Figures~\ref{fig:diffratiocrime1}-\ref{fig:diffratiocrime2}.
\begin{figure}[ht!]
\centering
\begin{tabular}{c c}
\includegraphics[scale=0.35]{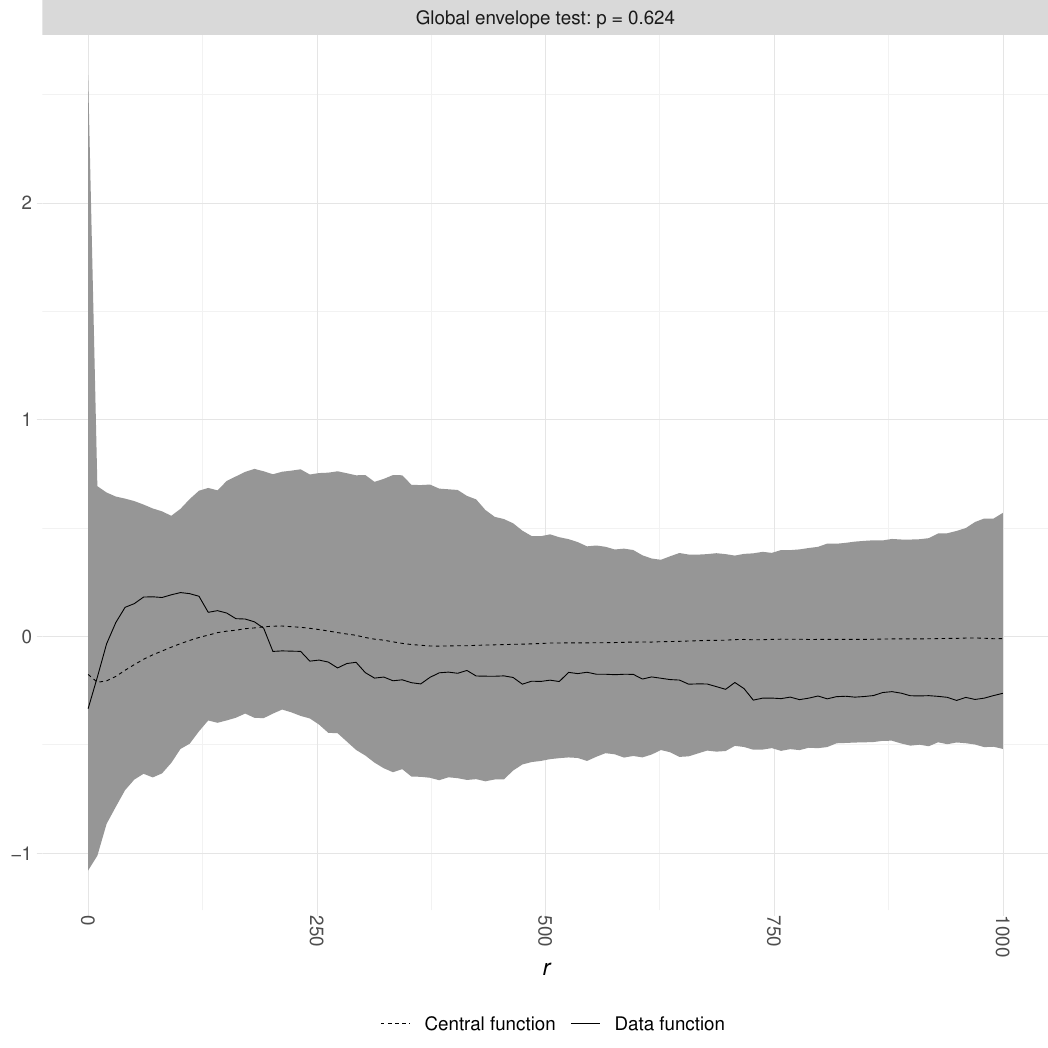}
\hspace{0.1cm}
\includegraphics[scale=0.35]{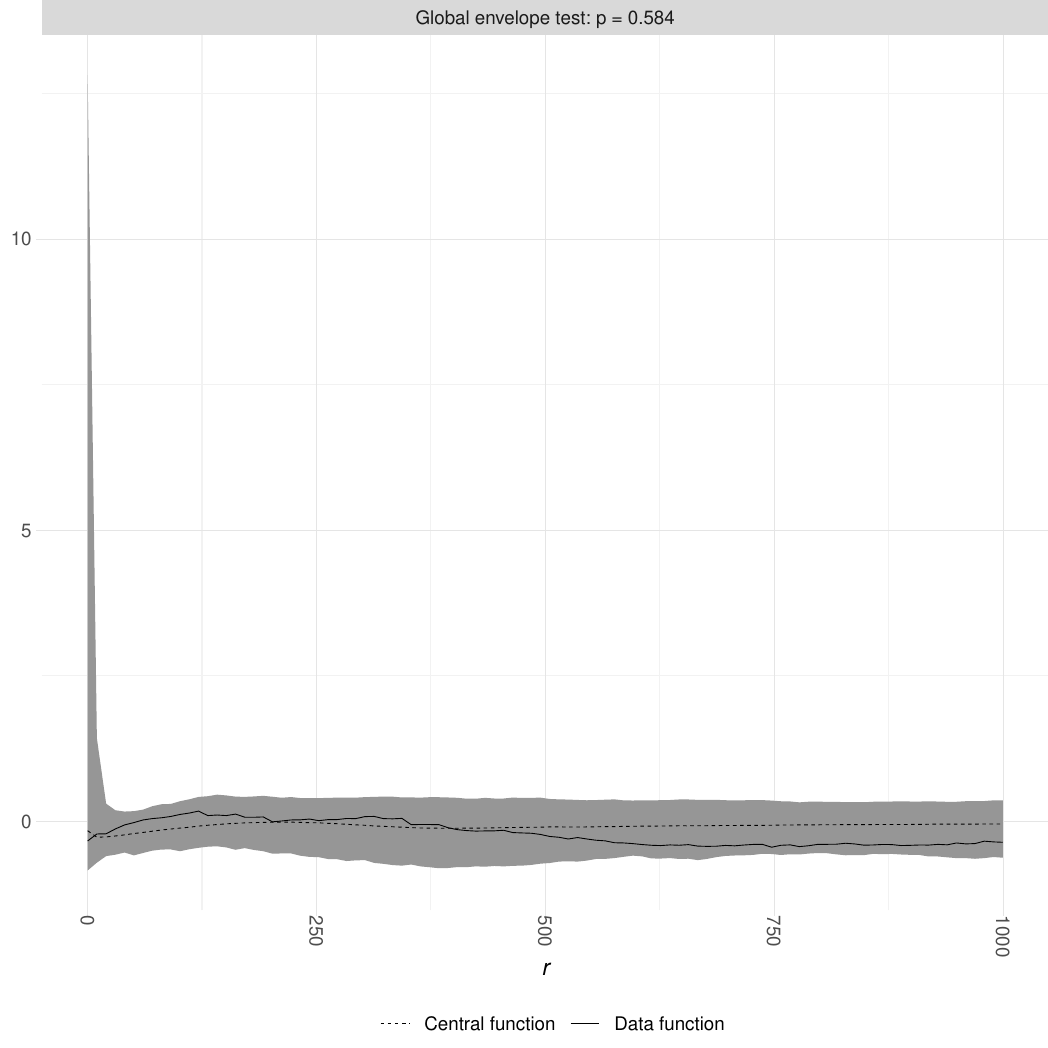}
\end{tabular}
\caption{Differences between semi-parametric and consistent non-parametric estimates of cross PCF ratios (solid curves) with global 95\% envelopes (gray shaded areas). Left: $g_{1}/g_{6}$. Right: $g_{2}/g_{6}$.} \label{fig:diffratiocrime1}
\end{figure}

\begin{figure}[ht!]
\centering
\begin{tabular}{c c}
\includegraphics[scale=0.35]{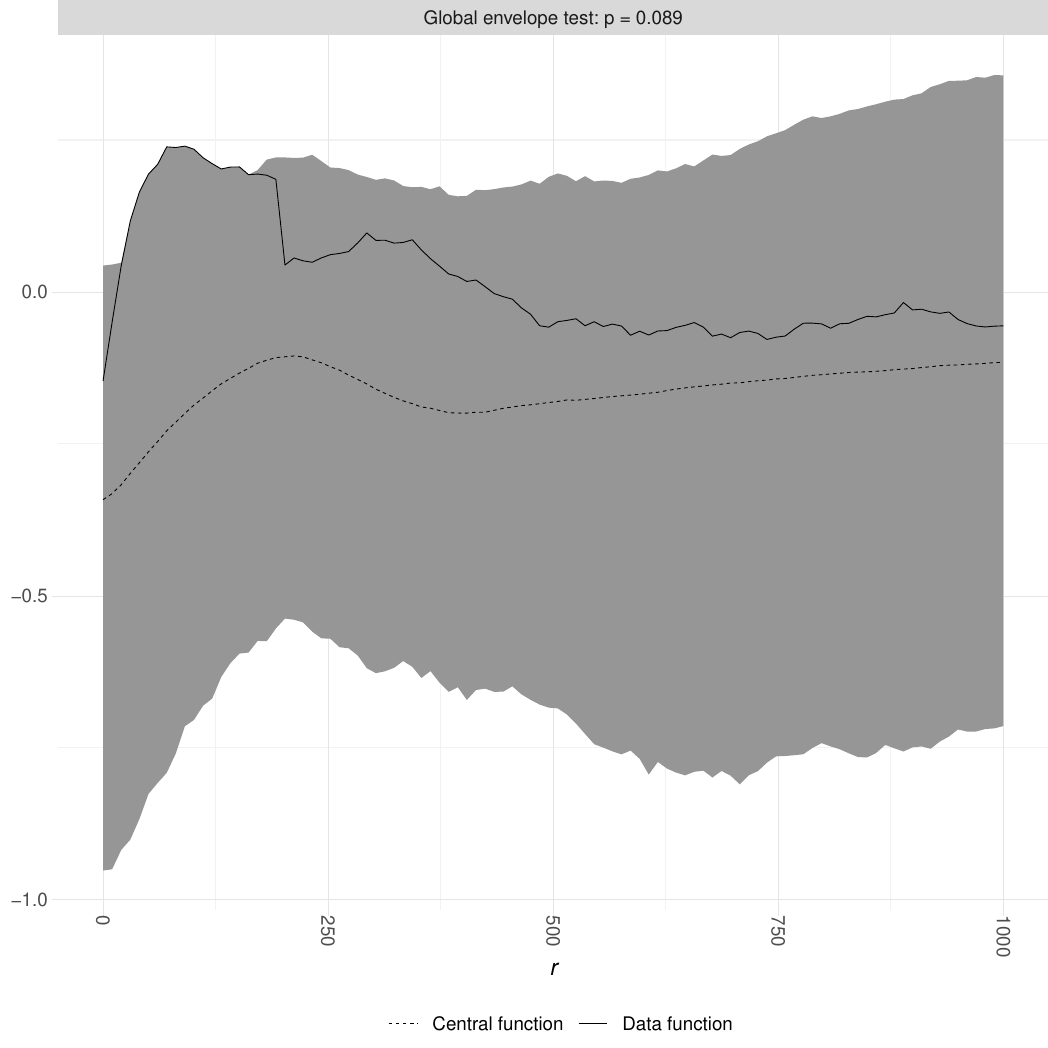}
\hspace{0.1cm}
\includegraphics[scale=0.35]{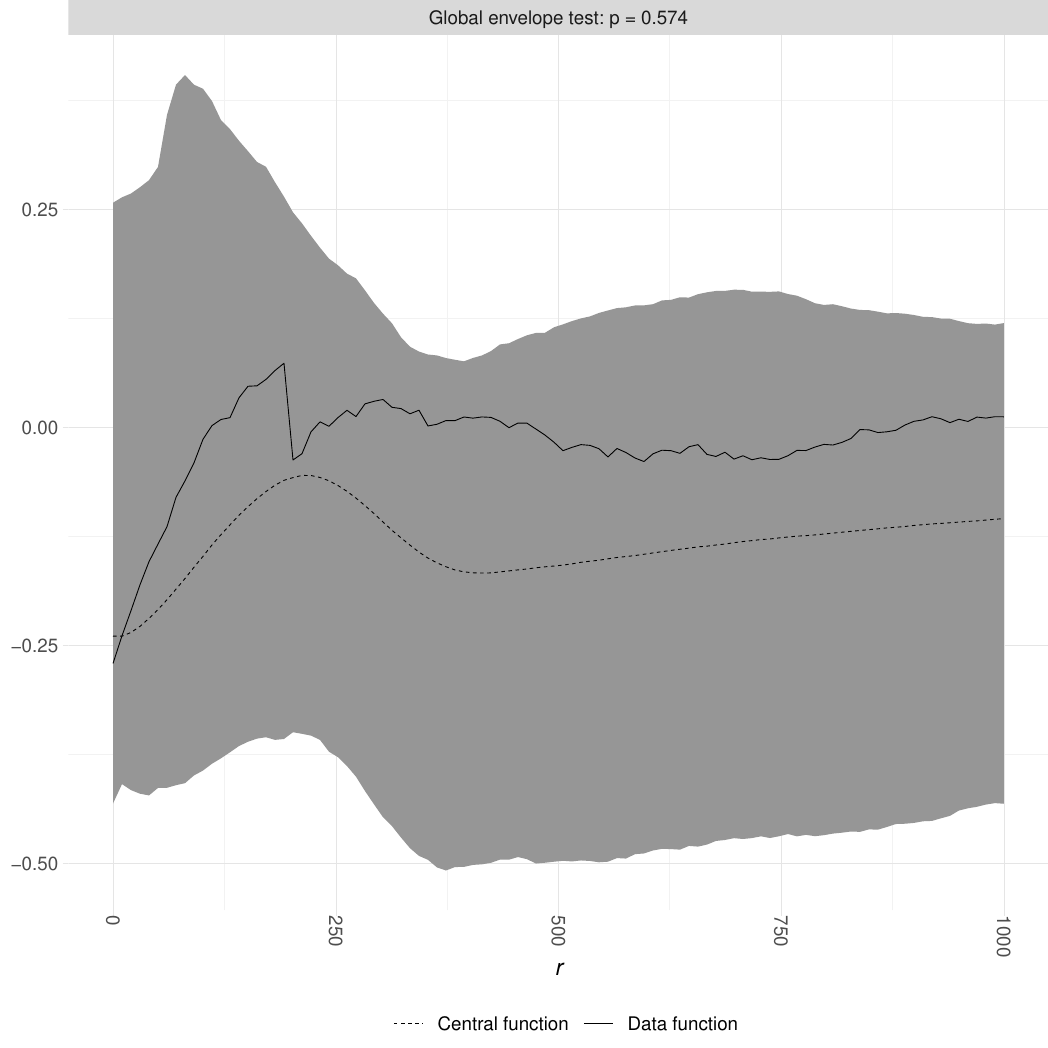}
\end{tabular}
\caption{Differences between semi-parametric and consistent non-parametric estimates of cross PCF ratios (solid curves) with global 95\% envelopes (gray shaded areas). Left: $g_{13}/g_{6}$. Right: $g_{16}/g_{6}$.} \label{fig:diffratiocrime2}
\end{figure}


\section{Model assessment using inhomogeneous cross
    $J$-function}\label{sec:jinhom}

As a supplement to the model assessment using non-parametric estimates
of cross PCF ratios, we also considered the inhomogeneous cross
$J$-function \citep{cronie:lieshout:16} which can be computed in terms
of the inhomogeneous cross $G$ and $F$-functions. Estimates of the
inhomogeneous $G$ and $F$ are implemented
in the \texttt{R} package \texttt{spatstat}
\citep{baddeley:rubak:turner:15}. The $G$ and $F$-function estimates
require estimates of the intensity functions $\rho_i$ and we here use
the leave-one-type out estimate \eqref{eq:hatrhoi} to mitigate issues
with confounding of random clustering and variation in $\rho_0$. 

For both the lymphoma data and the crime data the results are
mixed. For $J_{\textup{Str},\textup{CD8}}$,
$J_{\textup{CD8},\textup{Str}}$, $J_{\textup{Hyp},\textup{Nor}}$ and
$J_{\textup{Nor},\textup{Hyp}}$, we obtain global envelope test
$p$-values 0.28, 0.40, 0.34 and 0.03. Plots for $J_{\textup{Hyp},\textup{Nor}}$ and
$J_{\textup{Nor},\textup{Hyp}}$ are shown in Figure~\ref{fig:jinhomlymph}.
\begin{figure}[ht!]
\centering
\begin{tabular}{c c}
\includegraphics[scale=0.35]{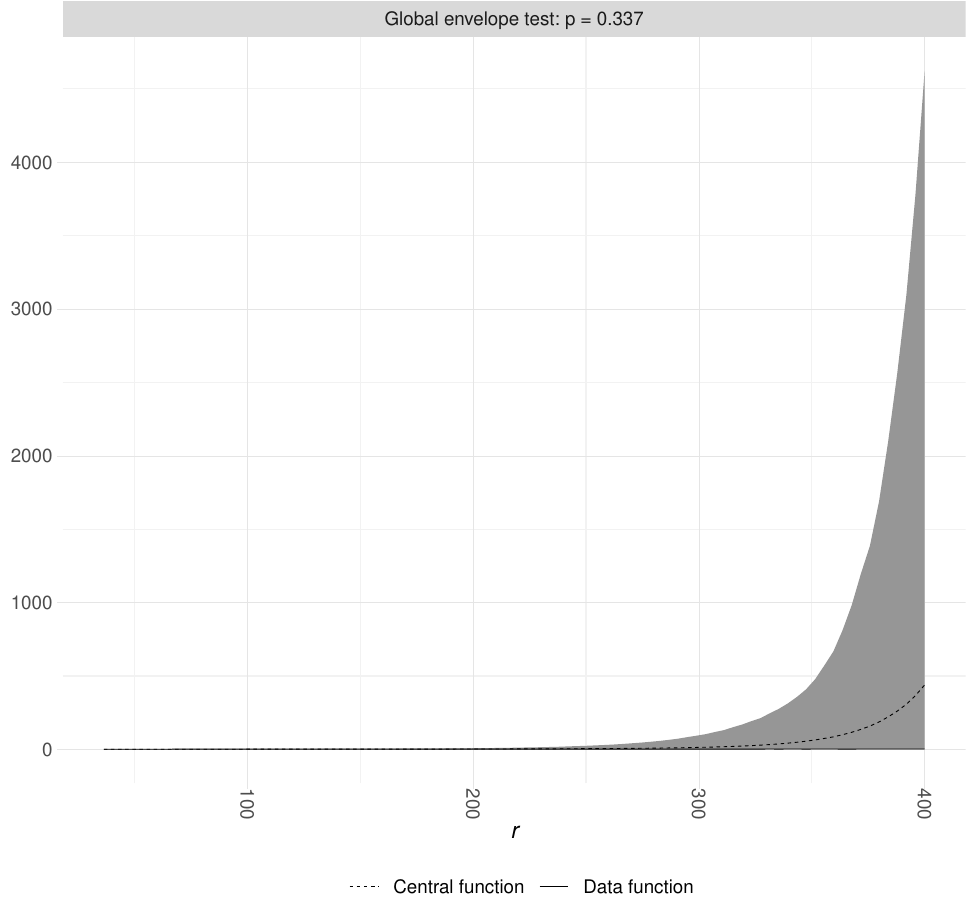}
\hspace{0.1cm}
\includegraphics[scale=0.35]{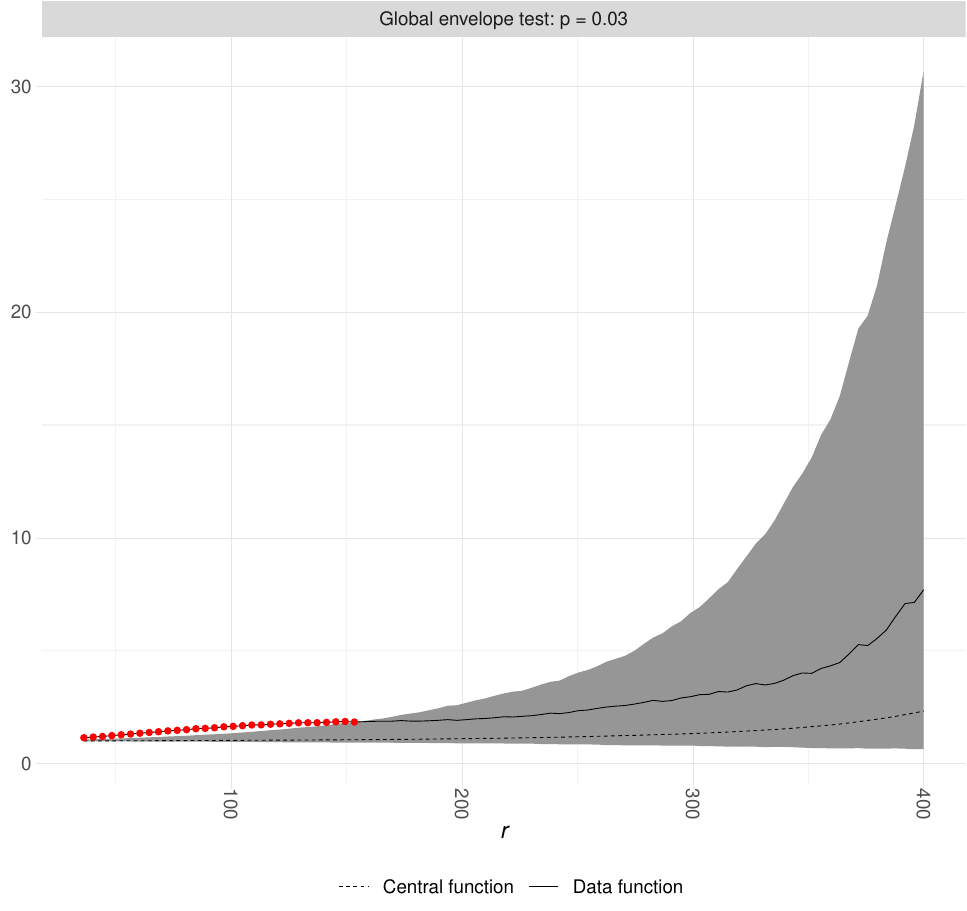}
\end{tabular}
\caption{Tumor cells data. Global envelope plots for inhomogeneous cross $J$-functions  (solid curves) with global 95\% envelopes (gray shaded areas). Left: $J_{\textup{Hyp},\textup{Nor}}$. Right: $J_{\textup{Nor},\textup{Hyp}}$.} \label{fig:jinhomlymph}
\end{figure}

Overall the inhomogeneous cross $J$-results for the lymphoma data do
not seem to provide serious evidence against the fitted model.

For the crimes data we obtain in total 30 global envelope test
$p$-values summarized in Table~\ref{tbl:pvalues}. Examples of two global envelopes
plots (for $J_{12}$ and $J_{15}$) are shown in Figure~\ref{fig:crossJcrime}.
 \begin{table}
\caption{\label{tbl:pvalues}Distribution of global envelope $p$-values for 30 cross  inhomogeneous $J$-functions}
 \centering
\fbox{
 \begin{tabular}{l | c  c c c  c}
 $p$-value & 0.01 & 0.02 & 0.03 & 0.04 & $>$0.05 \\ \hline
 frequency & 8 & 3 &  2 & 2 & 15
 \end{tabular}}
 \end{table}
\begin{figure}[!htb]
\centering
\begin{tabular}{cc}
\includegraphics[width=0.45\textwidth]{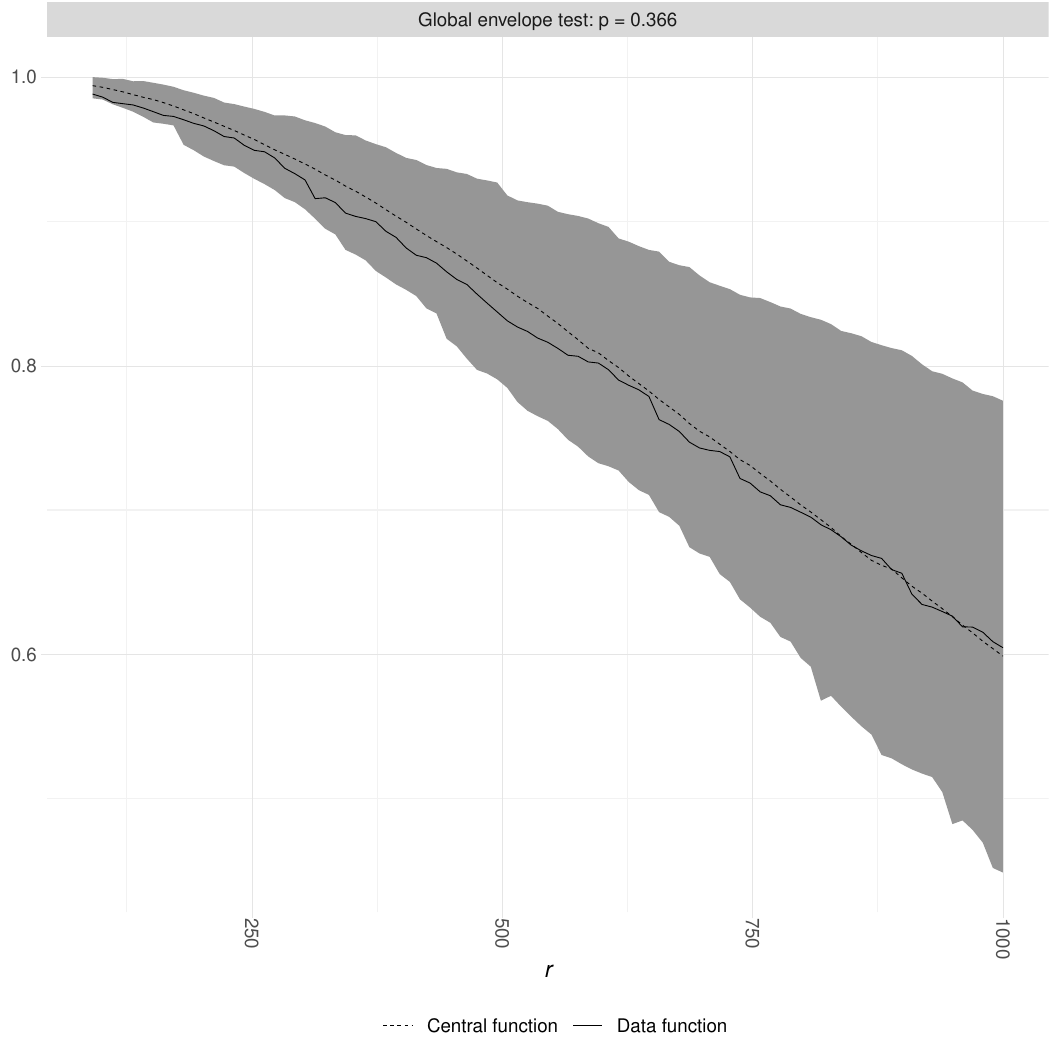} &                                                       \includegraphics[width=0.45\textwidth]{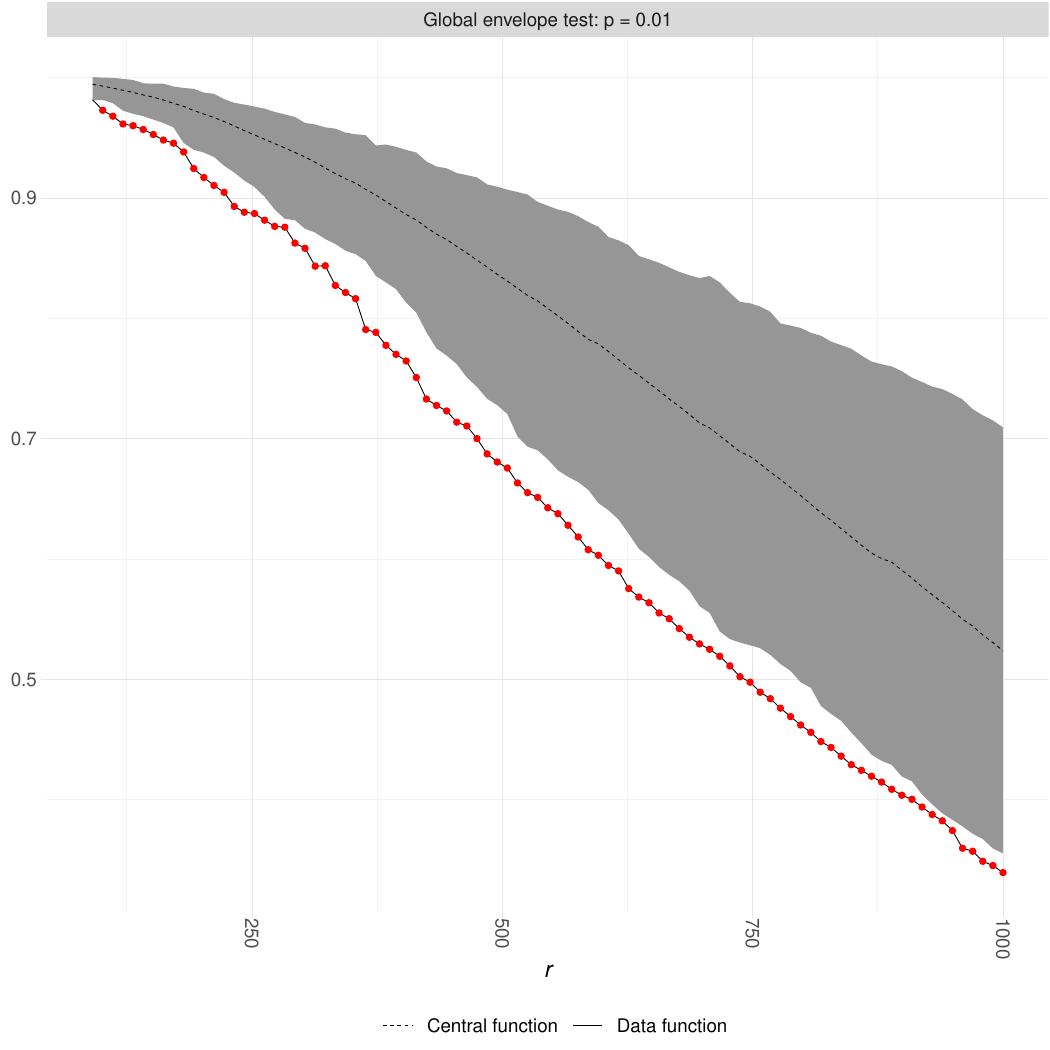} 
\end{tabular}
\caption{Global envelope plots for inhomogeneous $J_{12}$ (left) and
  $J_{15}$ (right).}\label{fig:crossJcrime}
\end{figure}

The large fraction of small $p$-values (15 out of 30 less than
5\%) raises doubt about the fitted model for the crimes
data. 

In each
case we obtain a small $p$-value, the estimated $J$-function falls below the
simulated $J$-functions (as in the right plot of
Figure~\ref{fig:crossJcrime}) indicating stronger clustering for the data
than for the simulations. However, it is difficult to exactly pinpoint
what is the cause of the apparent model deficiency. This is because the results depend
both on the estimate of $\rho_0$ and the fitted multivariate LGCP
structure not depending on $\rho_0$. The non-parametric kernel
estimate in the right
plot of Figure~11 is our best suggestion for a data-driven estimate
of $\rho_0$. However, this estimate is most likely
too smooth compared to the true $\rho_0$ that for the crime data depends on the complex urban structure and population density. Hence
it is likely that the resulting cross $J$-function obtained for the
data overestimates clustering by picking up
variation due to $\rho_0$. On the other hand, this does not happen for
the simulated data, where the smoothness of the $\rho_0$ used for the
simulations likely agrees better with the smoothness of the $\rho_0$ estimates
subsequently obtained from the simulated data. In other words, we do
not know whether the low global envelope $p$-values are due to misspecification of
the model for cross-interactions or due to using a wrong estimate of
the intensity function.

Our model assessment using  non-parametric estimates of
cross PCF ratios is less prone to the above issues since these
estimates do not require knowledge of $\rho_0$.

\bibliography{../bib/mybib}
\bibliographystyle{../jrssc_format/rss}